\documentclass[12pt]{article}

\usepackage{microtype}

\usepackage{graphicx}
\usepackage{fullpage}
\usepackage{dsfont}
\usepackage{color}
\usepackage{enumerate,authblk}

\usepackage{amsmath,amsfonts,amssymb,amsthm}
\newcommand{\fP}{\mathfrak{A}}
\newcommand{\fQ}{\mathfrak{B}}
\long\def\ignore#1{}

\newtheorem{theorem}{Theorem}
\newtheorem{lemma}[theorem]{Lemma}
\newtheorem{corollary}[theorem]{Corollary}

\newtheorem{definition}[theorem]{Definition}

\newtheorem{observation}[theorem]{Observation}
\usepackage{amsthm}

\makeatletter
\newtheorem*{rep@theorem}{\rep@title}
\newcommand{\newreptheorem}[2]{%
\newenvironment{rep#1}[1]{%
 \def\rep@title{#2 \textbf{\ref{##1}}}%
 \begin{rep@theorem}}%
 {\end{rep@theorem}}}
\makeatother

\newreptheorem{theorem}{\textbf{Theorem}}

\newcommand{\commentout}[1]{}
\newcommand{\eat}[1]{}

\newcommand{\calH}{{\mathcal H}}
\newcommand{\dualH}{{\mathcal H}^\star}

\newcommand{\calL}{{\mathcal L}}
\newcommand{\calI}{{\mathcal I}}

\newcommand{\calP}{{\mathcal P}}
\newcommand{\calK}{{\mathcal K}}
\newcommand{\bcalK}{\overline{\mathcal K}}
\newcommand{\calS}{{\mathcal S}}

\newcommand{\calR}{{\mathcal R}}
\newcommand{\calE}{{\mathcal E}}
\newcommand{\dualE}{{\mathcal E}^\star}

\newcommand{\TK}{\mathsf{TK}}

\newcommand{\extent}{{\mathfrak E}}

\newcommand{\Prob}{{\mathrm{Pr}}}
\renewcommand{\Pr}{{\mathrm{Pr}}}
\newcommand{\Exp}{{\mathbb{E}}}
\newcommand{\obj}{\mathsf{obj}}

\newcommand{\E}{{\mathbb{E}}}

\newcommand{\R}{{\mathbb{R}}}

\newcommand{\dist}{\mathrm{d}}
\renewcommand{\d}{\mathrm{d}}

\newcommand{\ver}{\mathrm{Ver}}
\newcommand{\polar}{\star}
\newcommand{\perror}{\tau}

\renewcommand{\S}{\mathbb{S}}
\newcommand{\bH}{\overline{H}}
\newcommand{\bK}{\overline{K}}
\newcommand{\tX}{\widetilde{X}}
\newcommand{\tO}{\widetilde{O}}

\newcommand{\pois}{\mathrm{Pois}}

\newcommand{\e}{\epsilon}

\newcommand{\indicator}{\mathbb{I}}

\newcommand{\p}{p}
\renewcommand{\P}{\mathcal{P}}
\newcommand{\V}{\mathcal{V}}

\newcommand{\poly}{\mathrm{poly}}
\newcommand{\polylog}{\mathrm{polylog}}

\newcommand{\Halfplanes}{\mathds{H}}

\newcommand{\expkernel}{$\eps$-\textsc{exp-kernel}}
\newcommand{\exprkernel}{$(\e,r)$-\textsc{fpow-kernel}}
\newcommand{\probkernel}{$(\eps,\perror)$-\textsc{quant-kernel}}

\newcommand{\conedecomp}{\mathds{A}}
\newcommand{\pu}{\Pr^R}
\newcommand{\interior}{\mathrm{int}\,}
\newcommand{\segment}{\mathsf{seg}}
\newcommand{\goodseg}{\mathsf{Gs}}
\newcommand{\badseg}{\mathsf{Bs}}

\newcommand{\dw}{\omega}   %directional width
\newcommand{\innerprod}[2]{\langle #1,#2 \rangle}
\newcommand{\CH}{\mathsf{ConvH}}
\newcommand{\eps}{\varepsilon}
\renewcommand{\epsilon}{\varepsilon}
\newcommand{\grad}{\nabla}
\newcommand{\hypercube}{\overline{\mathbb{C}}}
\newcommand{\M}{M}

\newcommand{\topic}[1]{\vspace{0.2cm}\noindent {\bf #1}}
\newcommand{\jeff}[1]{{\color{blue} $\langle${\sffamily\small Jeff: }#1$\rangle$}}
\newcommand{\jian}[1]{{\color{red} $\langle${\sffamily\small Jian: }#1$\rangle$}}

\begin{document}
%\begin{spacing}{1.0}
%\setpagewiselinenumbers
%\modulolinenumbers[5]
%\linenumbers

\title{
$\varepsilon$-Kernel Coresets for Stochastic Points
}

%\author{
%Jian Li \\ IIIS\\ Tsinghua University \\ Beijing, China \and Jeff M. Phillips \\School of Computing \\ University of Utah \\ Salt Lake City, UT 84112  \and Haitao Wang\\ Department of Computer Science \\ Utah State University
%}

\author[1]{Lingxiao Huang}
\author[1]{Jian Li}
\author[2]{Jeff M. Phillips}
\author[3]{Haitao Wang}
\affil[1]{\small Tsinghua University, Beijing 100084, China, {\tt lijian83@mail.tsinghua.edu.cn}}
\affil[2]{\small University of Utah, Salt Lake City, UT 84112, USA, {\tt jeffp@cs.utah.edu}}
\affil[3]{\small Utah State University, Logan, UT 84322, USA, {\tt haitao.wang@usu.edu}}

\maketitle
\thispagestyle{empty}

\begin{abstract}
With the dramatic growth in the number of application domains that generate probabilistic, noisy and uncertain data,
there has been an increasing interest in designing algorithms for  geometric or combinatorial optimization problems
over such data.
In this paper, we initiate the study of constructing $\eps$-kernel coresets for uncertain points.
We consider uncertainty in the existential model where each point's location is fixed but only occurs with a certain probability,
and the locational model where each point has a probability distribution describing its location.
An $\eps$-kernel coreset approximates the width of a point set in any direction.
We consider approximating the expected width (an \expkernel), as well as
the probability distribution on the width (an \probkernel) for any direction. 
%\textcolor{red}{under the assumption that the dimension $d$ is a constant}.
We show that there exists a set of $O(\eps^{-(d-1)/2})$ deterministic points which approximate the expected width under the existential and locational models,
and we provide efficient algorithms for constructing such coresets.
We show, however, it is not always possible to find a subset of the original uncertain points which provides such an approximation.
However, if the existential probability of each point is lower bounded by a constant,
an \expkernel\ is still possible.
We also provide efficient algorithms for
construct an \probkernel\ coreset in nearly linear time.
Our techniques utilize or connect to
several important notions in probability and geometry,
such as Kolmogorov distances, VC uniform convergence and Tukey depth,
and may be useful in other geometric optimization problem in
stochastic settings.
Finally, combining with known techniques, we show a few applications to approximating the extent of uncertain functions,
maintaining extent measures for stochastic moving points and some shape fitting problems under uncertainty.
\eat{
In particular, we obtain the first PTAS for the minimum stochastic $j$-flat-center problem,
generalizing a previous result for stochastic minimum enclosing ball.
}
\end{abstract}

\section{Introduction}

\topic{Uncertain Data Models:}
The wide deployment of sensor monitoring infrastructure and
increasing prevalence of technologies such as data integration and
cleaning \cite{cheng2008cleaning,dong2007data} have resulted
in an abundance of uncertain, noisy and probabilistic data in many scientific and application domains.
Managing, analyzing, and solving optimization problems over such data have become
an increasingly important issue and have attracted significant attentions from several research communities
including theoretical computer science, databases, machine learning
and wireless networks.
\eat{
Consider a typical sensor monitoring application, where sensors are used to monitor the identities,
locations and other features of a number of objects (such as people, vehicles, or animals).
However, the sensor data can be very noisy and even conflicting with each other (e.g.,
the same object appears in two different locations at the same time).
A major way to capture such uncertain sensor readings
is to use probabilistic models (see e.g., \cite{conf/cidr/DeshpandeGM05,DBLP:conf/vldb/DeshpandeGMHH04,jeffery2006declarative}).
%People are usually interested in various statistics of the uncertain point set,
%such as the expectation/distribution of the directional width, the diameter,
%the minimum enclosing ball, the minimum spanning tree etc.
}

In this paper, we focus on two stochastic models, the existential model and locational model.
Both models have been studied extensively for a variety of computational geometry problems or combinatorial problems,
such as closest pairs \cite{KCS11b}, nearest neighbors \cite{agarwal2012nearest,KCS11b}, minimum spanning trees \cite{huang2015approximating,KCS11a},
convex hulls \cite{AHPSYZ14,FHKS16,SVY13,XLJ16},
maxima \cite{afshani2011approximate}, perfect matchings \cite{huang2015approximating}, clustering \cite{cormode2008approximation,guha2009exceeding} , minimum enclosing balls \cite{enclosingball14} and
range queries \cite{ADP13,agarwal2012range,li2016range}.

\begin{enumerate}
\item Existential uncertainty model:
In this model,
there are a set $\P$ of $n$ points
in $\R^d$. Throughout this paper, we assume that the dimension $d$ is a constant.
Each point $v \in \P$ is associated with a real number (called {\em existential probability})
$\p_{v}\in [0,1]$ which
indicates that $v$ is present independently with probability $p_v$.
\item Locational uncertainty model:
There are a set $\P$ of $n$ points and the existence of each point is certain.
However, the location of each point $v\in\P$ is a random location in  $\R^d$.
We assume the probability distribution is discrete and independent of other points.
For a point $v\in \P$ and a location $s\in \R^d$, we use $\p_{v,s}$ to denote the probability that the location of point $v$ is $s$.
Let $S$ be the set of all possible locations, and let $|S| = m$ be the number of all such locations.
\end{enumerate}

In the locational uncertainty model, we distinguish the use of the terms ``points'' and ``locations'';
a point refers to the object with uncertain locations and a location refers to a point (in the usual sense) in $\R^d$.
We will use capital letters (e.g., $P,S,\ldots$) to denote  sets of deterministic points and
calligraphy letters ($\calP, \calS,\ldots$) to denote sets of stochastic points.

\topic{Coresets:}
Given a large dataset $P$ and a class $C$ of queries, a coreset $S$
is a dataset of much smaller size such that for every
query $r\in C$, the answer $r(S)$ for the small dataset $S$ is close to the answer $r(P)$ for the original large dataset $P$.
Coresets~\cite{Phi16} have become more relevant in the era of big data as they can drastically reduce the size of a dataset while guaranteeing that answers for certain queries are provably close.
An early notion of a coreset concerns the directional width problem
(in which a coreset is called an $\epsilon$-kernel)
and several other geometric shape-fitting problems in the seminal paper \cite{agarwal2004approximating}.

We introduce some notations and review the definition of $\epsilon$-kernel. We assume the dimension $d$ is a constant.
For a set $P$ of deterministic points,
the {\em support function} $f(P,u)$ is defined to be
$
f(P,u)=\max_{p\in P} \innerprod{u}{p}$  for  $u\in \R^d,
$
where $\innerprod{.}{.}$ is the inner product.
The {\em directional width} of $P$ in direction $u\in \R^d$,
denoted by $\dw(P,u)$, is defined by
$
\dw(P,u)=f(P,u)+f(P,-u).
$
It is not hard to see that the support function and the directional width
only depend on the convex hull of $P$.
%(Note the definition is slightly different from \cite{agarwal2004approximating}.)
A subset $Q\subseteq P$ is called an \emph{$\epsilon$-kernel of $P$} if for
each direction $u\in \R^d$,
$
(1-\eps) \dw(P,u)
\leq \dw(Q,u)
\leq \dw(P,u).
$
For any set of $n$ points,
there is an $\epsilon$-kernel of size $O(\epsilon^{-(d-1)/2})$~\cite{agarwal2004approximating,agarwal2005geometric}, which can be constructed
in $O(n + \eps^{-(d-3/2)})$ time~\cite{Cha06,YAPV04}.

\subsection{Problem Formulations}
In this paper we focus on constructing $\eps$-kernel coresets when the input data is uncertain.  These results not only provide an understanding of how to compactly represent an approximate convex hull under uncertainty, but can lead to solutions to a variety of other shape-fitting problems.

\topic{$\eps$-Kernels for expectations.}
Suppose $\calP$ is a set of stochastic points
(in either the existential or locational uncertainty model).
Define the {\em expected} directional width of $\calP$ in direction $u$
to be
$
\dw(\calP,u)=\Exp_{P\sim\calP}[\dw(P,u)],
$
where $P\sim \calP$ means that $P$ is a (random) realization of $\calP$.

\begin{definition}
\label{def:expkernel}
For a constant $\epsilon>0$, a set $S$ of (deterministic or stochastic) points in $\R^d$
is called an \emph{\expkernel} of $\calP$, if for all directions $u\in \R^d$,
\[
(1-\epsilon)\dw(\calP,u)\leq \dw(S,u)\leq \dw(\calP,u).
\]
\end{definition}

Recall in the deterministic setting, we require that
the $\epsilon$-kernel $S$ be a subset of the original point set
(we call this {\em the subset constraint}).
It is important to consider the subset constraint since it can reveal how concisely arbitrary uncertain point sets can be represented with just a few uncertain points (size depending only on $\eps$).
%Alternatively a \emph{deviant} coreset is one that may violate the subset constraint.
For $\eps$-kernels on deterministic points, $\Omega(\eps^{-(d-1)/2})$ points may be required and can always be found under the subset constraint~\cite{agarwal2004approximating,agarwal2005geometric}.
However, in the stochastic setting, we will show this is no longer true.  Coresets without the subset constraint, in fact made of deterministic points, can sometimes be obtained when no coreset with the subset constraint is possible.

%However, in the stochastic setting (Definition~\ref{def:expkernel}),we do not require the full subset constraint (with the same probability distribution for each chosen point).
%In fact, under the subset constraint,
%there exists no $\epsilon$-kernel with small size.

\topic{$\eps$-Kernels for probability distributions.}
Sometimes it is useful to obtain more than just the expected value (say of the width) on a query; rather one may want (an approximation of) a representation of the full probability distribution that the query can take.

\begin{definition}
\label{def:expkernelCDF}
For a constant $\epsilon,\perror>0$, a set $\calS$ of stochastic points in $\R^d$
is called an \emph{\probkernel} of $\calP$, if for all directions $u$ and all $x\geq 0$,
\begin{align}
\label{eq:quantcore}
\Pr_{P\sim \calP}\Bigl[\dw(P,u)\leq (1-\epsilon)x\Bigr]-\perror\leq \Pr_{S\sim \calS}\Bigl[\dw(S,u)\leq x\Bigr] \leq
\Pr_{P\sim \calP}\Bigl[\dw(P,u)\leq (1+\epsilon)x\Bigr]+\perror.
\end{align}
\end{definition}

In the above definition, we do not require the points in $\calS$ are independent.
So when they are correlated, we will specify the distribution of $\calS$.
If all points in $\calP$ are deterministic and $\perror<0.5$, the above definition essentially boils down to
requiring $(1-\epsilon)\dw(\calP,u)\leq \dw(\calS,u)\leq (1+\epsilon)\dw(\calP,u)$.
Assuming the coordinates of the input points are bounded,
an \probkernel\ ensures that for any choice of $u$, the cumulative distribution
function of $\omega(\calS,u)$ is within a distance $\eps$ under the  L\'{e}vy metric, to that of $\omega(\P,u)$.
\footnote{
Assuming the coordinates of the input points are bounded,
the requirement for an \probkernel\ is in fact stronger than that of L\'{e}vy distance being no larger than $\epsilon$
as the former requires a multiplicative error on length, which gives better guarantee when the length is small.
}.
%For ease of notation,
%we sometimes write
%$\Pr\bigl[\dw(\calP,u)\leq t\bigr]$ to denote
%$\Pr_{P\sim \calP}\bigl[\dw(P,u)\leq t\bigr]$,
%and write \eqref{eq:quantcore}
%as
%$
%\Pr_\calS\Bigl[\dw(S,u)\leq x\Bigr] \in
%\Pr_\calP\Bigl[\dw(P,u)\leq (1\pm \epsilon)x\Bigr]\pm\epsilon.
%$

\topic{$\e$-Kernels for expected fractional powers.} Sometimes, the
notion \expkernel\ is not powerful enough for certain shape fitting problems
(e.g., the minimum enclosing cylinder problem and the minimum spherical shell problem)
in the stochastic setting.
The main reason is the appearance of the
$l_2$-norm in the objective function.
So we need to be able to handle the fractional powers in the objective function.
For a set $P$ of points in $\R^d$,
the polar set of $P$ is defined to be
$P^{\polar}=\{u\in \R^d\mid \langle u,v\rangle\geq 0, \forall v\in P\}$.
Let $r$ be a positive integer.
Given a set $P$ of points in $\R^d$ and $u\in P^\polar$, we define a function
$$
T_r(P,u)=\max_{v\in P}\innerprod{u}{v}^{1/r}-\min_{v\in P}\innerprod{u}{v}^{1/r}.
$$
We only care about the directions in $\calP^{\polar}$ (i.e., the polar of the points in $\calP$)
for which $T_r(P,u), \forall P\sim \calP$ is well defined.

\begin{definition}
\label{def:fun}
For a constant $\epsilon>0$, a positive integer $r$, a set $\calS$ of stochastic points in $\R^d$ is called an \exprkernel\ of $\calP$,
if for all directions $u\in \calP^{\polar}$,
$$
(1-\e)\Exp_{P\sim \calP}[T_r(P,u)]\leq \Exp_{P\sim \calS}[T_r(P,u)]\leq (1+\e)\Exp_{P\sim \calP}[T_r(P,u)].
$$
\end{definition}

\subsection{Our Results}

Now, we discuss the main technical results of the paper.

\topic{$\eps$-Kernels for expectations.}
First, we consider \expkernel s under various constraints.
Our first main result is that an \expkernel\ of size $O(\eps^{-(d-1)/2})$
exists for both existential and locational uncertainty model
and can be constructed in nearly linear time.

\begin{theorem}
\label{thm:expconstruction}
$\calP$ is a set of $n$ uncertain points in $\R^d$
(in either locational uncertainty model or existential uncertainty model).
There exists an \expkernel\ of size $O(\eps^{-(d-1)/2})$ for $\calP$.
For existential uncertainty model (locational uncertainty model resp.),
such an \expkernel\ can be constructed in $O(\epsilon^{-(d-1)} n\log n)$ time,
($O(\epsilon^{-(d-1)} m\log m)$ time resp.),
where $n$ is the number of points and $m$ is the total number of possible locations.
\end{theorem}

The existential result is a simple Minkowski sum argument.
We first show that there exists a convex polytope $M$
such that for any direction, the directional width of $M$ is exactly the same as the expected directional width of $\calP$ (Lemma~\ref{lm:existence}).
This immediately implies the existence of a \expkernel\ consisting $O(\eps^{-(d-1)/2})$ deterministic points
(using the result in \cite{agarwal2004approximating}), but without the subset constraint.
The Minkowski sum argument seems to
suggest that the complexity of $M$ is exponential.
However, we show that the complexity of $M$ is in fact polynomial
$O(n^{2d-2})$
and we can construct it explicitly in $O(n^{2d-1}\log n)$ time
(Theorem~\ref{thm:constructM}).

Although the complexity of $M$ is polynomial, we cannot afford to
construct it explicitly if we are to construct an \expkernel\ in
nearly linear time.
Thus we construct the \expkernel\ without explicitly constructing
$M$.
In particular, we show that it is possible to find the extreme vertex of $M$
in a given direction in nearly linear time,
by computing the gradient of the support function of $M$.
We also provide quadratic-size data structures that can calculate the exact width $\dw(\calP,\cdot)$ in logarithmic time under both models in $\R^2$ (Appendix~\ref{sec:computing}).

We also show that under subset constraint
(i.e., the \expkernel\ is required to be
a subset of the original point set,
with the same probability distribution for each chosen point),
there is no \expkernel\ of sublinear size
(Lemma~\ref{lm:negative}).
However, if there is a constant lower bound $\beta>0$ on the existential probabilities
(called $\beta$-assumption),
we can construct an \expkernel\ of constant size (Theorem~\ref{thm:beta} and  Appendix~\ref{sec:appendix}).

\topic{$\eps$-Kernels for probability distributions.}
Now, we describe our main results for \probkernel s.
We first propose a quite simple but general algorithm
for constructing  \probkernel s, which achieves the following guarantee.

\begin{theorem}
\label{thm:probconstruction}
An \probkernel\ of size
$\tO\left(\perror^{-2}\e^{-3(d-1)/2}\right)$
can be constructed in $\widetilde{O}\left(n\perror^{-2}\e^{-(d-1)}\right)$
time, under both existential and locational uncertainty models.
\end{theorem}

The algorithm is surprisingly simple.
Take a certain number of i.i.d. realizations,
compute an $\e$-kernel for each realization, and then
associate each kernel with probability $1/N$
(so the points are not independent).
The analysis requires the VC uniform convergence bound for
unions of halfspaces.
%We use the average width of these $\e$-approximations as our estimation.
%We show that such an \probkernel\ can be computed in linear time.
The details can be found in Section~\ref{sec:qkernel}.

For existential uncertainty model, we can improve the size bound as follows.

\begin{theorem}
\label{thm:probconstructionexsit}
$\calP$ is a set of uncertain points in $\R^d$ with existential uncertainty. Let $\lambda=\sum_{v\in \calP}(-\ln (1-p_v))$.
There exists an \probkernel\ for $\calP$,
which consists of a set of
independent uncertain points of cardinality $\min\{\tO(\perror^{-2}\max\{\lambda^2,\lambda^4\}),
\tO(\eps^{-(d-1)}\perror^{-2})\}$.
The algorithm for constructing such a coreset runs in  $\tO(n\log^{O(d)} n)$ time.
\end{theorem}

\eat{
If we do not require that the \probkernel\ be a collection of
independent stochastic points, there is a simpler construction, which works
for both existential and locational models.
}

%In Section~\ref{sec:qkernel}, we provide a randomized construction for an %\probkernel\ for any uncertainty model which gives an uniform sampler on points.

We note that another advantage of the improved construction is that the \probkernel\
is a set of independent stochastic points (rather than correlated points as in
Theorem~\ref{thm:probconstruction}).
We achieve the improvement by two algorithms.
The first algorithm transforms the Bernoulli distributed variables into
Poisson distributed random variables and creates a probability distribution
using the parameters of the Poissons, from which we take a number of i.i.d. samples as the coreset.
Our analysis leverages the additivity of Poisson distributions and the VC uniform convergence bound (for halfspaces).
However, the number of samples required depends on $\lambda(\calP)$,
so the first algorithm only works when $\lambda(\calP)$ is small.
The second algorithm complements the first one by identifying a convex set $K$ that lies in the convex hull of $\calP$ with high probability ($K$ exists when $\lambda(\calP)$ is large) and uses a small size deterministic $\epsilon$-kernel to approximate $K$.
The points in $\bK=\calP\setminus K$ can be approximated using the same sampling algorithm as in the first algorithm
and we can show that $\lambda(\bK)$ is small, thus requiring only a small number of samples.
Our algorithm can be easily extended to $\R^d$ for any constant $d$
and the size of the coreset is $\tO(\perror^{-2}\eps^{-(d-1)})$.
In Section~\ref{subsec:linearprobkernel}, we show such an \probkernel\ can be computed in $O(n\polylog n)$ time
using an iterative sampling algorithm.
Our technique has some interesting connections
to other important geometric problems (such as the Tukey depth problem) \cite{matousek1991computing},
may be interesting in its own right.

\topic{$\e$-Kernels for expected fractional powers.}
For \exprkernel s,
we provide a linear time algorithm for constructing
an \exprkernel\ of size $\tO(\epsilon^{-(rd-r+2)})$ in the existential uncertainty model under the $\beta$-assumption.
The algorithm is almost the same as the construction in Section~\ref{sec:qkernel}
except that some parameters are different.

\begin{theorem} (Section~\ref{sec:rfunction})
\label{thm:exprconstruction}
%$\calP$ is a set of $n$ uncertain points in $\R^d$ with existential uncertainty under the $\beta$-assumption.
An \exprkernel\ of size $\tO(\e^{-(rd-r+2)})$
can be constructed in $\widetilde{O}\left(n\e^{-(rd-r+4)/2}\right)$ time
in the existential uncertainty model under the $\beta$-assumption.
\eat{
In particular, the \exprkernel\ consists of $N=\tO(\e^{-(rd-r+4)/2})$ point sets, each occuring with probability $1/N$
and containing $O(\e^{-r(d-1)/2})$ deterministic points.
}
\end{theorem}

\topic{Applications to Uncertain Function Approximation and Shape Fitting.}
Finally, we show that the above results, combined with the duality
and linearization arguments \cite{agarwal2004approximating},
can be used to obtain constant size coresets for
the function extent problem in the stochastic setting,
and to maintain extent measures for stochastic moving points.

Using the above results, we also obtain efficient approximation schemes for
various shape-fitting problems in the stochastic setting, such as minimum enclosing ball, minimum spherical shell,
minimum enclosing cylinder and minimum cylindrical shell
in different stochastic settings.
We summarize our application results in
the following theorems.
The details can be found in Section~\ref{sec:app}.

\begin{theorem}
\label{thm:squarecenter}
Suppose $\calP$ is a set of $n$ independent stochastic points in $\R^d$ under either existential or locational
uncertainty model.
There are linear time approximation schemes
for the following problems:
(1)  finding a center point $c$ to minimize $\E[\max_{v\in \calP} \| v-c \|^2]$;
(2) finding a center point $c$ to minimize $\E[\obj(c)]=\E[\max_{v\in P} \| v-c \|^2-\min_{v\in P} \| v-c \|^2]$.
Note that when $d=2$ the above two problems correspond to minimizing the expected areas of the enclosing ball and the enclosing annulus, respectively.
\end{theorem}

Under $\beta$-assumption, we can obtain efficient
approximation schemes
for the following shape fitting problems.

\begin{theorem}
\label{thm:shapefittingbeta}
Suppose $\calP$ is a set of $n$ independent stochastic points in $\R^d$,
each appearing with probability at least $\beta$, for some fixed constant $\beta>0$.
There are linear time approximation schemes for minimizing
the expected radius (or width) for the minimum spherical shell, minimum enclosing cylinder, minimum cylindrical shell problems over $\calP$.
\end{theorem}

\subsection{Other Related Work}
Besides the stochastic models mentioned above,
geometric uncertain data has also been studied in the \emph{imprecise} model
~\cite{bs-ads-04,hm-ticpps-08,k-bmips-08,ls-dtip-08,nt-teb-00,obj-ue-05,kl-lbbsd-10}.
In this model, each point is provided with a region where it might be.  This originated with the study of imprecision in data representation~\cite{gss-cscah-93,gss-egbra-89}, and can be used to provide upper and lower bounds on several geometric constructs such as the diameter, convex hull, and flow on terrains~\cite{DHLS13,kl-lbbsd-10}.
%The other common models are existential~\cite{KCS11b,KCS11a,SVY13} and locational~\cite{agarwal2009indexing,agarwal2012nearest,ABSHNSW06,CG09,CLY09,JLP11}
%models described above.
%these roughly correspond to the tuple and attribute uncertainty, respectively, studied in the database community~\cite{CLY09,uncert-model}.

Convex hulls have been studied for uncertain points:
upper and lower bounds are provided under the imprecise model~\cite{ES11,LvK08,nt-teb-00,kl-lbbsd-10},
distributions of circumference and volume are calculated in the locational model~\cite{JLP11,loffler2009shape},
the most likely convex hull is found in the existential model in $\R^2$ and shown NP-hard for $\R^d$ for $d > 2$ and in the locational model~\cite{SVY13},
and the probability a query point is inside the convex hull~\cite{AHPSYZ14,FHKS16,XLJ16}.
As far as we know, the expected complexity of the convex hull under uncertain points has not been studied, although it has been studied~\cite{HP11} under other random data models.

There is a large body of literature~\cite{Phi16} on constructing coresets for various problems, such as shape fitting \cite{agarwal2004approximating,agarwal2005geometric},
shape fitting with outliers \cite{har2004shape},
clustering \cite{chen2009coresets,FL11,feldman2012data,har2004coresets,  langberg2010}, integrals \cite{langberg2010}, matrix approximation and regression \cite{deshpande2006matrix,FL11}
 and in different settings, such as geometric data streaming \cite{agarwal2005geometric,Cha06}
 and privacy setting \cite{feldman2009private}.
Coresets were constructed for imprecise points~\cite{LvK08} to help derive results for approximating convex hulls and a variety of other shape-fitting problems, but because of the difference in models, these approaches do not translate to existential or locational models.
In the locational model, coresets are created for range counting queries~\cite{ADP13} under the subset constraint, but again these techniques do not translate because $\eps$-kernel coresets in general cannot be constructed from a density-preserving subset of the data, as is preserved for the range counting coresets.
Also in the locational model (and directly translating to the existential model) L\"offler and Phillips~\cite{loffler2009shape} show how a large set of uncertain points can be approximated with a set of deterministic point sets, where each certain point set can be an $\eps$-kernel.  This can provide approximations similar to the \probkernel\ with space $O(\eps^{-(d+3)/2}\log(1/\delta))$ with probability at least $1-\delta$.  However it is not a coreset of the data, and answering width queries requires querying $O(\eps^{-2}\log(1/\delta))$ deterministic point sets.

Recently, Munteanu et al. \cite{enclosingball14} studied the minimum enclosing ball problem over stochastic points,
and obtained an efficient approximation scheme.
Their algorithm and analysis utilize the results
from the deterministic coreset literature~\cite{ackermann2010clustering}.
%They defined a near-metric distance measure
%between two non-empty points sets, which satisfies %non-negativity, symmetry and the triangle inequality.
However, they do not directly address the problem of
constructing coresets for stochastic points and
it is also unclear how to extend their technique
to other shape fitting problems, such as minimum spherical
shells.

Technically, our \probkernel\ construction bears some similarity
to the coreset by Har-Peled and Wang~\cite{har2004shape} for handling outliers.
From the dual (function extent) perspective,
they want to approximate the distance between two level sets in an arrangement of hyperplanes,
and (the dual of) $\calH$ in Section~\ref{subsec:linearprobkernel}
also needs to be (approximately) sandwiched by two fractional level sets
(our hyperplanes have weights).
However, we have an important requirement that the total weight outside $(1+\epsilon)\calH$ must be small, which cannot be addressed by their technique.

%\input{prel}

%\input{beta}

%\section{Extent Measures}
%\begin{lemma}
%\label{lm:specialpt}
%For any directional width query $u$, we can identify a point $v^\star$ such that, with probability at least $1-\beta^{n/2}$, the smallest point along direction $u$ that is instantiated is at least as small $v^\star$ and the largest point along direction $u$ that is instantiated is at least as large as $v^\star$.
%\end{lemma}
%\begin{proof}
%Reorder all points by index $i$ so that $\langle v_i, u \rangle \leq \langle v_{i+1}, u\rangle$ for all $i \in [1,n-1]$.  Set $v^\star = v_{\lfloor{n/2}}$.  It follows there are at least $n/2$ possible points smaller or equal to $v^\star$ in there $u$ ordering, and at least as many larger.  Each of these points independently has probably at least $\beta$ of existing, thus the probability that none exist smaller or equal to $v^\star$ is at most $\beta^{n/2}$.  And thus the probability that there is no smaller or equal point, or there is no larger or equal point is at most $2 \beta^{n/2}$.  Thus the probability the $v^\star$ is in the closed interval between the smallest and largest point in the $u$ order is at least $1-2 \beta^{n/2}$.
%\end{proof}

\section{$\eps$-Kernels for Expectations of Width}
\label{sec:expkernel}

We first state our results in this section for the existential uncertainty model.
All results can be extended to the locational uncertainty model, with slightly different bounds
(essentially replacing the number of points $n$ with the number of locations $m$) or assumptions.
We describe the difference for locational model in the appendix.

For simplicity of exposition, we assume in this section that all points in $\calP$ are in general positions
and all $p_v$s are strictly between $0$ and $1$.
%Our goal here is to extend the results in \cite{agarwal2004approximating} to uncertain points.
For any $u,v\in \R^d$, we use $\innerprod{u}{v}$
to denote the usual inner product $\sum_{i=1}^d u_iv_i$.
For ease of notation, we write $v\succ_u w$ as a shorthand notation for
$\innerprod{u}{v}> \innerprod{u}{w}$.
For any $u\in \R^d$,
the binary relation $\succ_u$ defines a total order
of all vertices in $\calP$.
(Ties should be broken in an arbitrary but consistent manner.)
We call this order
the {\em canonical order of $\calP$ with respect to $u$}.
For any two points $u$ and $v$, we use $\dist(u,v)$ or $\|v-u\|$
to denote their Euclidean distance.
For any two sets of points, $A$ and $B$,
the Minkowski sum of $A$ and $B$ is defined as
$A\oplus B:=\{a+b\mid a\in A, b\in B\}$.
Recall the definitions for a set $P$ of deterministic points and a direction $u\in \R^d$,
the support function is $f(P,u)=\max_{p\in P} \innerprod{u}{p}$ and
the {\em directional width} is
$
\dw(P,u)=f(P,u)-f(P,-u).
$
The support function and the directional width only depend on the convex hull of $P$.

%\begin{theorem}
%For any $\epsilon>0$ and input $\calP$ in $\R^d$, there exists an \expkernel { } $S$ of $O(1/\epsilon^{(d-1)/2})$ deterministic points (which may not be a subset of $\calP$).
%\end{theorem}

\begin{lemma}
\label{lm:existence}
Consider a set $\calP$ of uncertain points in $\R^d$ (in either locational uncertainty model or existential uncertainty model).
There exists a set $S$ of deterministic points in $\R^d$ (which may not be a subset of $\calP$) such that
$
\dw(u,\calP) = \dw(u,S)
$ for all $u\in \R^d$.
\end{lemma}

\begin{proof}
By the definition of the expected directional width of $\calP$, we have that
\[
\dw(\calP,u)=\Exp_{P\sim\calP}[\dw(P,u)]=\sum_{P\sim\calP}\Pr[P]\Bigl(f(P, u)+f(P,-u)\Bigr).
\]
Consider the Minkowski sum
$
\M=\M(\calP):=\sum_{P\sim \calP} \Pr[P] \CH(P),
$
where $\CH(P)$ is the convex hull of $P$ (including the interior).
It is well known that the Minkowski sum of a set of convex sets is also convex.
Moreover, it also holds that for all $u\in \R^d$ (see e.g., \cite{schneider1993convex})
$
f(\M,u)= \sum_{P\sim\calP}\Pr[P]f(P,u).
$
Hence, $\dw(\calP,u)=\dw(\M,u)$ for all $u\in \R^d$.
\end{proof}

By the result in \cite{agarwal2004approximating}, we know that for any convex body in $\R^d$,
there exists an $\epsilon$-kernel of size $O(\epsilon^{-(d-1)/2})$.
Combining with Lemma~\ref{lm:existence}, we can immediately obtain the following corollary, which is the first half of Theorem~\ref{thm:expconstruction}.

\begin{corollary}
\label{cor:existence}
For any $\epsilon>0$, there exists an \expkernel\ of size $O(\epsilon^{-(d-1)/2})$.
\end{corollary}

Recall that in Lemma~\ref{lm:existence}, the Minkowski sum
$
\M=\sum_{P\sim \calP} \Pr[P] \CH(P).
$
Since $\M$ is the Minkowski sum of exponential many convex polytopes, so $\M$ is also a convex polytope.
At first sight, the complexity of $\M$ (i.e., number of vertices) could be exponential.
However, as we will show shortly, the complexity of $\M$ is in fact polynomial.

We need some notations first.
For each pair $(r,w)$ of points in $\calP$ consider the hyperplane $H_{r,w}$ that passes
 through the origin and is orthogonal to the line connecting $r$ and $w$.
We call these ${n\choose 2}$ hyperplanes {\em the separating hyperplanes induced by $\calP$}
and use $\Gamma$ to denote the set.
Each such hyperplane divides $\R^{d}$ into 2 halfspaces.
For all vectors $u\in \R^d$ in each halfspace, the order of $\innerprod{r}{u}$ and $\innerprod{w}{u}$
is the same (i.e., we have $r \succ_u w$ in one halfspace and $w\succ_u r$ in the other).
Those hyperplanes in $\Gamma$ pass through the origin and thus
partition $\R^{d}$ into $d$-dimensional polyhedra cones.
\footnote{
We ignore the lower dimensional cells in the arrangement.
}
We denote this {\em arrangement} as $\conedecomp(\Gamma)$.

\eat{
\begin{figure}[t]
\centering
\includegraphics[width=0.7\linewidth]{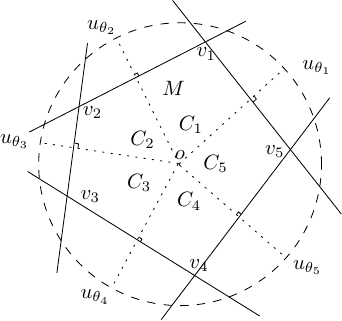}
\caption{
LHS: The figure depicts a pentagon $M$ in $\R^2$ to illustrate some intuitive facts in
convex geometry.
(1) The plane can be divided into 5 cones $C_1,\ldots, C_5$, by 5 angles $\theta_1,\ldots, \theta_5$.
$u_{\theta_i}$ is the unit vector corresponding to angle $\theta_i$.
Each cone $C_i$ corresponds to a vertex $v_i$ and
for any direction $u\in C_i$, $f(\M,u)=\innerprod{u}{v_i}$
and the vector $\grad f(\M,u)$ is $v_i$.
(2) Each direction $\theta_i$ is perpendicular to an edge of $\M$.
$\M=\cap_{i=1}^5 H_i$ where $H_i$ is the supporting halfplane with normal vector $u_{\theta_i}$.
RHS: $W=\{v_3=v_{\max},v_4,v_5,v_6,v_7=v_{\min}\}$ is the set of vertices surrounded by the dashed cycle.
}
\label{fig:partition}
\end{figure}
}

\eat{
\begin{figure}[t]
\centering
\includegraphics[width=0.3\linewidth]{searchmin}
\caption{
$W=\{v_3=v_{\max},v_4,v_5,v_6,v_7=v_{\min}\}$ is the set of vertices surrounded by the dashed cycle.
}
\label{fig:searchmin}
\end{figure}
}

Consider an arbitrary cone  $C\in \conedecomp(\Gamma)$.
Let $\interior C$ denote the interior of $C$.
We can see that for all vectors $u\in \interior C$, the canonical order of $\calP$
with respect to $u$ is the same
(since all vector $u\in \interior C$ lie in the same set of halfspaces).
We use $|\M|$ to denote the complexity of $M$, i.e., the number of vertices in $\CH(\M)$.

\begin{lemma}
\label{lm:complexity}
Assuming the existential model and $p_v\in (0,1)$ for all $v\in \calP$,
the complexity of $M$ is the same as the cardinality of $\conedecomp(\Gamma)$, i.e.,
$|M|=|\conedecomp(\Gamma)|.$
Moreover, each cone $C\in \conedecomp(\Gamma)$ corresponds to exactly one
vertex $v$ of $\CH(M)$ in the following sense:
the gradient $\grad f(\M,u) = v$ for all $u\in \interior C$
(note that here $v$ should be understood as a vector).
\end{lemma}

\begin{proof}
We have shown that $M$ is a convex polytope.
We first note that the support function
uniquely defines a convex body (see e.g., \cite{schneider1993convex}).
We need the following well known fact in convex geometry (see e.g., \cite{ghosh1998support}):
For any convex polytope $M$,
$\R^d$ can be divided into exactly $|\M|$ polyhedra cones (of dimension $d$, ignoring the boundaries),
such that each such cone $C_v$ corresponds to a vertex $v$ of $\M$,
and for each vector $u\in C_v$, it holds $f(\M,u)=\innerprod{u}{v}$
(i.e., the maximum of $f(\M,u)=\max_{v'\in \M}\innerprod{u}{v'}$ is achieved by $v$ for all $u\in C_v$).\footnote{
One intuitive way to see this is as follows:
The support function for a polytope is just the upper envelope of a finite set of linear functions, thus
a piecewise linear function, and the domain of each piece is a polyhedra cone.
}
See Figure~\ref{fig:partition} for an example in $\R^2$.
Hence, for each $u\in \interior C_v$
the gradient of of the support function (as a function of $u$) is exactly $v$:
\begin{align}
\label{eq:grad2}
\grad f(\M,u) =\Bigl\{\frac{\partial f(\M,u)}{\partial u_j}\Bigr\}_{j\in [d]}=
\Bigl\{\frac{\partial \innerprod{u}{v}}{\partial u_j}\Bigr\}_{j\in [d]} = \Big\{\frac{\partial \sum_{j\in [d]} v_ju_j}{\partial u_j}\Big\}_{j\in [d]}= v,
\end{align}
where $u_j$ is the $j$th coordinate of $u$.
%So, $\grad f(\M,u)$ is not continuous at the boundaries of these regions.
With a bit abuse of notation,
we denote the set of cones defined above by $\conedecomp(\M)$. %the {\em canonical cone decomposition} with respect to $M$.

Now, consider a cone  $C\in \conedecomp(\Gamma)$.
We show that for all $u\in \interior C$, $\grad f(\M,u)$ is a distinct constant vector independent of $u$.
In fact, we know that $f(\M,u)=f(\calP,u)=\sum_{v\in \calP}\pu(v,u)\innerprod{v}{u}$,
where $\pu(v,u) =\prod_{v' \succ_u v} (1-p_{v'}) p_v$.
For all $u\in \interior C$, the $\pu(v,u)$ value is the same since
the value only depends on the canonical order with respect to $u$, which is the same for all $u\in C$.
Hence, we can get that for all $u\in \interior C$,
\begin{align}
\label{eq:grad}
\grad f(\M,u)=\sum_{v\in \calP}\pu(v,u) v,
\end{align}
which is a constant independent of $u$.
We prove the lemma by showing that
the gradient $\grad f(\M,u)$ must be different
for two adjacent cones $C_1,C_2$ (separated by some hyperplane in $\Gamma$) in $\conedecomp(\Gamma)$.
Suppose $u_1\in \interior C_1$ and $u_2\in \interior C_2$.
Consider the canonical orders $O_1$ and $O_2$ of $\calP$ with respect to $u_1$ and $u_2$ respectively.
Since $C_1$ and $C_2$ are adjacently, $O_1$ and $O_2$ only differ by
one swap of adjacent vertices.
W.l.o.g., assume that $O_1=\{v_1,\ldots, v_i, v_{i+1},\ldots,v_n\}$
and $O_1=\{v_1,\ldots, v_{i+1}, v_{i},\ldots,v_n\}$.
Using \eqref{eq:grad}, we can get that
\begin{align*}
\grad f(\M,u_1)-\grad f(\M,u_2) & = \pu(v_i,u_1) v_i+\pu(v_{i+1},u_1) v_{i+1}  - \pu(v_i,u_2)  v_i - \pu(v_{i+1},u_2) v_{i+1}  \\
&=D \cdot (p_{v_i} v_i +(1-p_{v_i})p_{v_{i+1}}v_{i+1} -p_{v_{i+1}} v_{i+1} -(1-p_{v_{i+1}})p_{v_{i}}v_{i})\\
&=D \cdot p_{v_i} p_{v_{i+1}}(v_i-v_{i+1})\ne 0
\end{align*}
where $D=\prod_{j=1}^{i-1} (1-p_{v_j})\ne 0$.
%Since $\grad f(\M,u)$ is different in $C_1$ and $C_2$,
%$C_1$ and $C_2$ must correspond to different vertices of $M$ by \eqref{eq:grad2}.

In summary, we have shown in the first paragraph that $\grad f(\M,u)$
is piecewise constant, with a distinct constant in each cone in $\conedecomp(\M)$.
The same also holds for $\conedecomp(\Gamma)$.
This is only possible if $\conedecomp(\Gamma)$ (thinking as a partition of $\R^d$)
partitions $\R^d$ exactly the same way as $\conedecomp(\M)$ does.
Hence, we have $\conedecomp(\Gamma)=\conedecomp(\M)$
and the lemma follows immediately.
\end{proof}

\begin{figure}[t]
\centering
\includegraphics[width=0.35\linewidth]{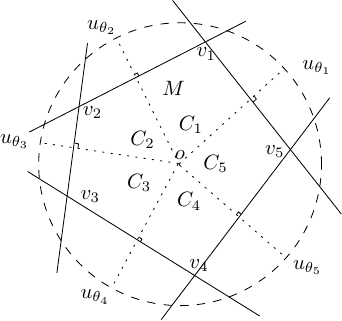}
\caption{
The figure depicts a pentagon $M$ in $\R^2$ to illustrate some intuitive facts in
convex geometry.
(1) The plane can be divided into 5 cones $C_1,\ldots, C_5$, by 5 angles $\theta_1,\ldots, \theta_5$.
$u_{\theta_i}$ is the unit vector corresponding to angle $\theta_i$.
Each cone $C_i$ corresponds to a vertex $v_i$ and
for any direction $u\in C_i$, $f(\M,u)=\innerprod{u}{v_i}$
and the vector $\grad f(\M,u)$ is $v_i$.
(2) Each direction $\theta_i$ is perpendicular to an edge of $\M$.
$\M=\cap_{i=1}^5 H_i$ where $H_i$ is the supporting halfplane with normal vector $u_{\theta_i}$.
%RHS: $W=\{v_3=v_{\max},v_4,v_5,v_6,v_7=v_{\min}\}$ is the set of vertices surrounded by the dashed cycle.
}
\label{fig:partition}
\end{figure}

Since $O(n^2)$ hyperplanes passing through the origin can
divide $\R^{d}$ into at most $O({n^2 \choose d-1})$ $d$-dimensional polyhedra cones
(see e.g., \cite{agarwal2000arrangements}), we immediately obtain the following corollary.
\begin{corollary} \label{cor:Mpolybound}
It holds that $|\M|\leq O({n^2\choose d-1})=O(n^{2d-2})$.
\end{corollary}

%\jeff{I don't follow this.  I see how to get $O(m^{2d})$ since each pair points divides $\S^{d-1}$ into 2 parts defined by a halfspace in $\R^d$, and the arrangement of ${m \choose 2} = O(m^2)$ halfspaces is of size $O(m^{2d})$.  Is it obvious how to reduce this to $O(m^d)$?}
%\jian{I don't know the right number. But there are ${m \choose d}=O(m^d)$ different hyperplanes (each goes through $d$ points). I just wrote an arbitrary number.}
%\jeff{I conjecture this can be reduced to $O(m^{2d-2})$.}
%Hence the expected directional width function $\dw(\calP, u)$ in $R$
%is a linear sum of cosine functions, hence has a continuous derivative.
%This implies that $\dw(H,u)$ should have a continuous derivative for all $u\in R$.
%It is easy to see that the directional width function (defined over $\S^{d-1}$)
%of a convex polytope with $X$ facets (face with dimension $d-1$) has $X$
%pieces, each having a continuous derivative. So the number of facets of $\M$ is at most $O(\M^{2d})$.
%\jian{In fact, I don't have a rigorous argument here. But I think the complexity of $\M$ is polynomial, not exponential.}

%We have shown the existence of a polynomial-size convex polytope that preserves the directional width of $\calP$
%exactly and a constant size \expkernel.
%Now, we focus on design efficient algorithms for finding such polytope and \expkernel.
%From now on, we restrict our attention in the two dimensional plane $\R^2$.

\color{black}

The proof of Lemma~\ref{lm:complexity} can be easily made constructive.
We only need to compute the set $\Gamma$ of all $O(n^2)$ hyperplanes and the arrangement $\conedecomp(\Gamma)$
in $O(n^{2d-2})$ time (see e.g., \cite{agarwal2000arrangements,edelsbrunner1986constructing}).
%and the vertices at each of their intersections (along directions $u$ where $d-1$ points have the same inner product with $u$).
%And as discussed there are $O(n^{2d-2})$ such points, and they can be found by enumeration in the same time.
Given each cone $C\in \conedecomp(\Gamma)$, we can calculate $\grad f(\M, u)$ for any $u\in C$, which gives exactly
one vertex of $M$ by \eqref{eq:grad2}, in $O(n \log n)$ time using the algorithm described in Lemma \ref{lem:10}.

\begin{theorem}
\label{thm:constructM}
In $\R^d$ for constant $d$, the polytope $M$ which defines $f(\calP,u)$ for any direction $u$ can be described with $O(n^{2d-2})$ vertices in $\R^d$, and can be computed in $O(n^{2d-1} \log n)$ time.
In $\R^2$, the runtime can be improved to $O(n^2 \log n)$.
\end{theorem}

 The improved running time in $\R^2$ is derived in Lemma \ref{lem:20} by carefully constructing each vertex of $M$ in $O(1)$ time using its neighboring vertex.  The extra $O(\log n)$ is needed to sort the vertices of $M$ to determine neighbors.

\subsection{A Nearly Linear Time Algorithm for Constructing \expkernel s}
\label{sec:linearalgo}

Now, we prove the main algorithmic result Theorem~\ref{thm:expconstruction} of this section:
we can find an \expkernel\ in nearly linear time.
If we already have the Minkowski sum $\M$,
we can directly use the algorithm in \cite{agarwal2004approximating} to find an $\epsilon$-kernel for $\M$.
However, constructing $\M$ explicitly takes $O(n^{2d-1}\log n)$ time according to Theorem~\ref{thm:constructM}
and this cannot be improved in general as the complexity of $\M$ is $O(n^{2d-2})$.
Therefore, in order to achieve a nearly linear time coreset construction, we can not compute $\M$ explicitly.
For ease of description, we first consider existential uncertainty model.
The details for locational uncertainty model can be found in Appendix~\ref{sec:appendix}.

\begin{theorem}\label{thm:expconstructionexistential}
(second half of Theorem~\ref{thm:expconstruction}, for existential model)
$\calP$ is a set of $n$ uncertain points in $\R^d$ with existential uncertainty.
An \expkernel\ of size $O(\eps^{-(d-1)/2})$ for $\calP$
can be constructed in $O(\epsilon^{-(d-1)} n\log n)$ time.
\end{theorem}

%We provide an outline of our algorithm and spell out some of the details of each step in the appendix.
%First, we need to find an affine transform $T$ such that
%the convex polygon $\M'=T(\M)$ is $\alpha$-fat for some constant $\alpha$.
%According to \cite{barequet2001efficiently}, in order to construct such $T$ (in $\R^2$),
%it suffices to identify two points in $\M$ such that their distance is a constant approximation of
%the diameter of $\M$.
%In Lemma~\ref{lm:transform}, we show this can be done without computing $\M$ explicitly.
%Next, we consider $K=O(1/\epsilon)$ directions
%$u_0, u_1,\ldots, u_{K-1}$ where $u_i=(\cos i\cdot \frac{2\pi}{K}, \sin i\cdot \frac{2\pi}{K})$.
%For each $u_i$, we compute the vertex $v_i = \arg\max_{v\in M'}\innerprod{u_i}{v}$.
%We show $v_i$ can be found in $O(n\log^2 n)$ time in Lemma~\ref{lm:extreme}.
%Finally, we argue that $\{v_i\}_{0\leq i<K}$ is an $\eps$-kernel of $\M$,
%thus also an \expkernel\ of $\calP$.
%We can see the total running time is $O(Kn\log^2n)=O(\frac{1}{\epsilon}n\log^2 n)$.
%The size can be further reduced to $O(1/\sqrt{\eps})$ using standard techniques~\cite{Cha06,YAPV04}.
%We can also slightly improve the running time to $O(\frac{1}{\sqrt{\epsilon}} n\log^2 n)$ time,
%based on the algorithm proposed in \cite{agarwal2004approximating}.

The following simple lemma provides an efficient procedure for finding the extreme vertex in $\M$ along any give direction, and is useful in several places later as well.
\begin{lemma}
\label{lm:extreme}
Given any direction $u\in \R^d$, we can find in $O(n \log n)$ time a vertex $v^\star\in \M$, at which
$\innerprod{v}{u}$ is maximized, over all $v\in \M$.
\end{lemma}
\begin{proof}
Fix an arbitrary direction $u\in \R^d$.
From the proof of Lemma~\ref{lm:complexity} (in particular \eqref{eq:grad2}), we know that
the vertex $v^\star\in \M$ that maximizes
$\innerprod{v}{u}$ can be computed by
$v^\star=\grad f(\M,u)$.
Using~\eqref{eq:grad}, $\grad f(\M,u)$ can be easily computed in $O(n\log n)$ time
(see Lemma~\ref{lem:10} for the details).
\end{proof}

Next, we need to find an affine transform $T$ such that
the convex polytope $\M'=T(\M)$ is $\alpha$-fat for some constant $\alpha$.
We recall that a set $P$ of points is {\em $\alpha$-fat}, for some constant $\alpha\leq 1$,
if there exists a point $x\in \R^d$, and a unit hypercube $\hypercube$ centered at $x$
such that $\alpha \hypercube \subset \CH(P)\subset \hypercube$.
According to Chapter 22 in \cite{har2011geometric}, in order to construct such $T$,
it suffices to identify two points in $\M$ such that their distance is a constant approximation of
the diameter of $\M$.
The following lemma (proven in Appendix~\ref{sec:appendix}) shows this can be done without computing $\M$ explicitly.

\begin{lemma}
\label{lm:transform}
We find an affine transform $T$ in $O(2^{O(d)}n\log n)$ time,
such that the convex polytope $\M'=T(\M)$ is $\alpha$-fat for some constant $\alpha$ ($\alpha$ may depend on $d$).
\end{lemma}

After obtaining $T$, we apply $T$ to $\calP$ in linear time.
Notice that $M'=T(\M(\calP))=\M(T(\calP))$.
Therefore, Lemma~\ref{lm:extreme} also holds for $M'$
(i.e., we can search over $M'$ the maximum vertex in any given direction in $O(n\log n)$ time).

%Let the vertex $v_i$ of $\M'$ be the one
%that maximizes $\innerprod{u_i}{v}$ over all $v\in \M'$
%where $u_i=u_{\theta_i}=(\cos i\cdot \frac{2\pi}{K}, \sin i\cdot \frac{2\pi}{K})$.

Let $\delta=O(\epsilon \alpha/d)$.
We compute a set $\calI$ of $O(\delta^{-(d-1)})=O(\epsilon^{-(d-1)})$ points on the unit sphere $\mathbb{S}^{d-1}$
such that for any point $v\in \mathbb{S}^{d-1}$, there is a point $u\in \calI$ such that $\|u-v\|\leq \delta$ (see e.g., \cite{agarwal1992farthest, chan2000approximating}).
For each $u$ in  $\calI$, we include $-u$ in $\calI$ as well.
For each direction $u\in \calI$, we compute $x(u)=\arg\max_{x\in M'} \innerprod{x}{u}$.
Based on the previous discussion, all $\{x(u)\}_{u\in \calI}$ can be computed in $O(\delta^{-(d-1)} n\log n)=O(\epsilon^{-(d-1)} n\log n)$ time.
%Now, we show $S=\{v_i\}_{0\leq i<k}$ is an $\eps$-kernel in the following lemma.

\begin{lemma}
\label{lm:directionkernel}
$S=\{x(u)\}_{u\in \calI}$ is an $\eps$-kernel for $M'$.
\footnote{
This is a folklore result.
A proof of the 2D case can be found in \cite{CM03}. The general case is a straightforward extension and
we provide a proof in Appendix~\ref{sec:appendix} for completeness.
}
\end{lemma}
%\jeff{There were many papers that came out around 2003 with this idea, for instance another one~\cite{SH03}
%that essentially already prove that this approach approximates the Hausdorff distance of the convex hull
%approximation within $\eps$ times the diameter of the point set.  Since it is fat,
%then the diameter is a constant approximation of the true width.}
Finally, we can then run existing $\eps$-kernel algorithms~\cite{Cha06,YAPV04} in $O(|S|)$ time to further reduce the size of $S$ to $O(\eps^{-(d-1)/2})$, which finishes the proof of Theorem~\ref{thm:expconstructionexistential}.
Lemma~\ref{lm:complexity} and Theorem~\ref{thm:expconstructionexistential} hold for locational uncertainty models as well. The details can be found in Appendix~\ref{sec:appendix}.
%For $\R^2$, we can also directly get an \expkernel\ of size $O(1/\sqrt{\epsilon})$
%in $O(\frac{1}{\sqrt{\epsilon}} n\log^2 n)$ time; these details can be found in the appendix.

\subsection{\expkernel\ Under the Subset Constraint}
\label{sec:expsubset}

First, we show that under the subset constraint
(i.e., the \expkernel\ is required to be
a subset of the original point set,
with the same probability distribution for each chosen point),
there exists no \expkernel\ with small size in general.
\footnote{If we require the $\eps$-\textsc{exp-kernel} to be
a subset of the original point set, but
with possibly different probability for each chosen point,
we do not know whether there always exists
an \expkernel\ with small size.}

\begin{lemma}
\label{lm:negative}
For some constant $\epsilon>0$,
there exist a set $\calP$ of stochastic points such that
no $o(n)$ size \expkernel\ exists for $\calP$ under the subset constraint
(for both locational model and
existential model).
\end{lemma}
\begin{proof}
To see this in the existential uncertainty model,
simply consider $n$ points, each with existence probability $1/n$.
$n/2$ of them co-locate at the origin
and the other $n/2$ of them co-locate at $x=1$.
It is not hard to see that the expected length of the diameter is $\Omega(1)$
but the expected length of the diameter of any $o(n)$ size subset is only $o(1)$
(with high probability, no point would even appear).
%\jeff{It seems this shows this is the wrong model for existential points.  Perhaps the total expected number of points should stay the same, so the points retained should have an increase in their $p_v$ value.}

The case for the locational uncertainty model is as simple.
Again, consider $n$ points.
For each point, with probability $1/n$, it appears at $x=1$.
Otherwise, its position is the origin (with probability $1-1/n$).
It is not hard to see that the expected length of the diameter of the original point set is $\Omega(1)$,
while that of any $o(n)$ size subset is only $o(1)$
(with high probability, no point would realize at $x=1$).
\end{proof}

In light of the above negative result, we make the following {\em $\beta$-assumption}:
we assume each possible location realizes a point with probability at least $\beta$, for a constant $\beta >0$.
%, and in fact for the locational model if each point can take one of $k$ locations, for a constant $k$,
%it is often assumed that each $p_{vs} = 1/k$.)
The proof of the following theorem can be found in Appendix~\ref{sec:appendix}.

%\input{structure}

%In this setting, we propose the following algorithm.
\begin{theorem}
\label{thm:beta}
Under the $\beta$-assumption, in the existential uncertainty model, there is an \expkernel\ in $\R^d$
of size $O(\beta^{-(d-1)}\epsilon^{-(d-1)/2} \log (1/\epsilon) )$ that satisfies the subset constraint.
\end{theorem}

\section{$\eps$-Kernels for Probability Distributions of Width}
\label{sec:probkernel}

Recall $\calS$ is an \probkernel\ if for all $x\geq 0$,
%\begin{align}
%\label{eq:probcore}
$
\Pr_{P\sim\calP}\Bigl[\dw(P,u)\leq (1-\epsilon)x\Bigr]-\perror\leq \Pr_{P\sim\calS}\Bigl[\dw(S,u)\leq x\Bigr] \leq
\Pr_{P\sim\calP}\Bigl[\dw(P,u)\leq (1+\epsilon)x\Bigr]+\perror.
$
%\end{align}
For ease of notation,
we sometimes write
$\Pr\bigl[\dw(\calP,u)\leq t\bigr]$ to denote
$\Pr_{P\sim \calP}\bigl[\dw(P,u)\leq t\bigr]$,
and abbreviate the above
as
$
\Pr\Bigl[\dw(S,u)\leq x\Bigr] \in
\Pr\Bigl[\dw(\calP,u)\leq (1\pm \epsilon)x\Bigr]\pm\perror.
$
We first provide a simple linear time algorithm for constructing
an \probkernel\ for both existential and locational models, in Section~\ref{sec:qkernel}.
The points in the constructed kernel are not independent.
Then, for existential models,
we provide a nearly linear time \probkernel\ construction where
all stochastic points in the kernel are independent in Section~\ref{subsec:probkernelexistential}.

\subsection{A Simple \probkernel\ Construction}
\label{sec:qkernel}

In this section, we show a linear time algorithm for
constructing an \probkernel\ for any stochastic model
if we can sample a realization from the model in linear time
(which is true for both locational and existential uncertainty models).
\eat{
Recall $\calS$ is an \probkernel\ if for all $x\geq 0$,
%\begin{align}
%\label{eq:probcore}
$
\Pr_{\calP}\Bigl[\dw(P,u)\leq (1-\epsilon)x\Bigr]-\perror\leq \Pr_\calS\Bigl[\dw(S,u)\leq x\Bigr] \leq
\Pr_\calP\Bigl[\dw(P,u)\leq (1+\epsilon)x\Bigr]+\perror.
$
%\end{align}
For ease of notation,
we sometimes write
$\Pr\bigl[\dw(\calP,u)\leq t\bigr]$ to denote
$\Pr_{P\sim \calP}\bigl[\dw(P,u)\leq t\bigr]$,
and abbreviate the above
as
$
\Pr_\calS\Bigl[\dw(S,u)\leq x\Bigr] \in
\Pr_\calP\Bigl[\dw(P,u)\leq (1\pm \epsilon)x\Bigr]\pm\perror.
$
}

%Now, we sketch an algorithm to achieve the above goal.
%For ease of exposition, we first describe the algorithm in $\R^2$.
%We also assume that all probability values are strictly between 0 and 1
%and $0<\epsilon,\perror\leq 1/2$ is a fixed constant.

\topic{Algorithm:}
Let $N=O\bigl(\perror^{-2}\e^{-(d-1)}\log \frac{1}{\e}\bigr)$.
We sample $N$ independent realizations from the stochastic model.
Let $\calH_i$ be the convex hull of the present points in the $i$th realization.
For $\calH_i$, we use the algorithm in \cite{agarwal2004approximating} to find a
deterministic $\epsilon$-kernel $\calE_i$ of size $O(\epsilon^{-(d-1)/2})$.
Our \probkernel\ $\calS$ is the following simple stochastic model:
with probability $1/N$, all points in $\calE_i$ are present.
Hence, $\calS$ consists of $O\bigl(\perror^{-2}\e^{-3(d-1)/2}\log \frac{1}{\e}\bigr)$ points (two such points either co-exist
or are mutually exclusive).
Hence, for any direction $u$,
$
\Prob[\dw(\calS,u)\leq t]=\frac{1}{N}\sum_{i=1}^N \indicator(\dw(\calE_i,u)\leq t),
$
where $\indicator(\cdot)$ is the indicator function.

\eat{
One useful property of the algorithm in \cite{agarwal2004approximating}
is that $\calE_i$ is a subset of the vertices of $\calH_i$.
Hence the convex polygon $\CH(\calE_i)$ is contained in $\calH_i$.
Since $\calE_i$ is an $\epsilon$-kernel, $(1+\epsilon)\CH(\calE_i))$ (properly shifted) contains $\calH_i$.
\eat{
\footnote{
In fact, most existing algorithms (e.g.,~\cite{agarwal2004approximating})
identify a point in the interior of $\calH$ as origin, and compute an $\epsilon$-kernel $\calE_i$
such that $f(\calE_i, u)\geq \frac{1}{1+\epsilon}f(\calH_i,u)$ for all directions $u$.
So, $\calH_i\subseteq (1+\epsilon)\CH(\calE_i)$ since $f(\calH_i,u)\leq f((1+\epsilon)\CH(\calE_i), u)$ for all directions $u$.
}
}
Let $\calK_i=\CH(\calE_i)$.
}

For a realization $P\sim \calP$, we use $\calE(P)$ to denote the deterministic $\epsilon$-kernel for $P$.
So, $\calE(P)$ is a random set of points, and we can think of $\calE_1,\ldots, \calE_N$ as samples from the random set.
Now, we show $\calS$ is indeed an \probkernel.
We start with the following simple observation.

%From sampler $\calP$, we construct a new sampler $\calP'$. While a realization is given by sampler $\calP$, sampler $\calP'$ outputs its %deterministic $\epsilon$-kernel. By the property of $\epsilon$-kernel, we can have the following observation.

\begin{observation}
\label{ob:quant}
For any $t\geq 0$ and any direction $u$, we have that
$$
\Pr[\dw(\calP,u)\leq t]\leq \Prob_{P\sim \calP}[\dw(\calE(P),u)\leq t]\leq \Pr[\dw(\calP,u)\leq (1+ \e)t].
$$
\end{observation}

\begin{proof}
For any realization $P$ of $\calP$, we have
$ \frac{1}{1+\e}\dw(P,u)\leq\dw(\calE(P),u)\leq \dw(P,u).
$
The observation follows by combining all realizations.
\end{proof}

We only need to show that $\calS$ is an \probkernel\ for $\calE(P)$.
We need the following two theorems.

\begin{theorem} (Theorem 5.22 in~\cite{har2011geometric}) (VC-dimension)
\label{thm:quant1}
Let $S_1=(X,\calR^1),\ldots,S_k=(X,\calR^k)$ be range spaces with VC-dimension $\delta_1,\ldots,\delta_k$, respectively. Next, let $f(r_1,\ldots,r_k)$ be a function that maps any $k$-tuple of sets $r_1\in \calR^1,\ldots, r_k\in \calR^k$ into a subset of $X$. Consider the range set
$$ \calR'=\{f(r_1,\ldots,r_k)\mid r_1\in \calR^1,\ldots, r_k\in \calR^k\}
$$
and the associated range space $(X,\calR')$. Then, the VC-dimension of $(X,\calR')$ is bounded by $O(k\delta \log k)$,
where $\delta=\max_i \delta_i$.
\end{theorem}

Suppose $(X,\calR)$ is a range space and $\mu$ is a probability measure over $X$.
We say a subset $C\subset X$ is an {\em $\e$-approximation} of the range space
if for any range $R\in \calR$, we have
$
\Bigl| \mu_C(R)-\mu(R) \Bigr|\leq \epsilon,
$
where $\mu_C(R)=|C\cap R|/|C|$.
We need the following celebrated uniform convergence result,
first established by Vapnik and Chervonenkis~\cite{vapnik1971uniform}.

\begin{theorem} (See Theorem 4.9 in~\cite{anthony2009neural})
\label{thm:quant2}
Suppose $(X,\calR)$ is any range space with VC-dimension at most $V$, where $|X|$ is finite and
$\mu$ is a probability measure defined over $X$.
For any $\e,\delta>0$, a random subset $C\subseteq X$ (according to $\mu$) of cardinality
$
s=O\left(\e^{-2}\left(V+\log (1/\delta)\right)\right)
$
is an $\e$-approximation for $X$ with probability $1-\delta$.
\end{theorem}

Now, we are ready to prove the main lemma in this section.

\begin{lemma}
\label{lm:quant1}
Let $N=O(\perror^{-2}\e^{-(d-1)}\log (1/\e))$.
For any $t\geq 0$ and any direction $u$, %we have that
\vspace{-.1in}
\[
\Prob[\dw(\calS,u)\leq t]\in \Prob_{P\sim \calP}[\dw(\calE(P),u)\leq t]\pm \perror.
\]
\end{lemma}

\begin{proof}
Let $L=O(\e^{-(d-1)/2})$.
We first note that $\calE(P)$ has at most $n^L$ possible realizations
since each $\e$-kernel is of size at most $L$,
We first build a mapping $g$ that maps each realization $\calE(P)$
to a point in $\R^{dL}$, as follows:
Consider a realization $P$ of $\calP$.
Suppose $\calE(P)=\{(x^1_1,\ldots, x^1_d),\ldots, (x^L_1,\ldots, x^L_d)\}$
(if $|\calE(P)|<L$, we pad it with $(0,\ldots, 0)$).
We let
$$
g(\calE(P))=(x^1_1,\ldots, x^1_d,\ldots, x^L_1,\ldots, x^L_d)\in \R^{dL}.
$$
For any $t\geq 0$ and any direction $u\in \R^d$,
note that $\dw(\calE(P),u)\geq t$ holds
if and only if there exists some $1\leq i,j\leq |\calE(P)|,i\neq j$ satisfies that
$\sum_{k=1}^d (x^i_k-x^{j}_k) u_k\geq t$,
which is equivalent to saying that
point $g(\calE(P))$ is in the union of those $O(|\calE(P)|^2)$ halfspaces
(for each $i,j$, we have one such halfspace).

Let $X$ be the image set of $g$. Let $(X,\calR^{i,j})$ ($1\leq i,j\leq L,i\neq j$) be a range space, where $\calR^{i,j}$ is the set of halfspaces $\{u=(u_1,\ldots,u_d)\in \R^d\mid \sum_{k=1}^d (x_k^i-x_k^{j})u_k\geq t\}$. Let $\calR'=\{\cup r_{i,j}\mid r_{i,j}\in \calR^{i,j},i,j\in [L]\}$. Note that each $(X,\calR^{i,j})$ has VC-dimension $d+1$. By Theorem~\ref{thm:quant1}, we can see that
the VC-dimension of $(X,\calR')$ is bounded by $O((d+1)L^2\lg L^2)=O(\e^{-(d-1)} \log (1/\e))$.
Notice that $\calS=\{\calE_1,\ldots,\calE_N\}$ is a collection of samples from $\calE(P)$.
Hence, by Theorem~\ref{thm:quant2}, for any $t$ and any direction $u$, we have that
$
\Prob[\dw(\calS,u)\leq t]\in \Prob_{P\sim \calP}[\dw(\calE(P),u)\leq t]\pm \perror.
$
\end{proof}

Combining Observation~\ref{ob:quant} and Lemma~\ref{lm:quant1}, we obtain the following theorem.

\begin{theorem}
\label{thm:quantmain}
Let $N=O(\perror^{-2}\e^{-(d-1)}\log (1/\e))$. For any $t\geq 0$ and any direction $u$, we have that
$$ \Prob[\dw(\calS,u)\leq t]\in \Prob[\dw(\calP,u)\leq (1\pm \e)t)]\pm \perror.
$$
\end{theorem}

\topic{Running time:}
In each sample, the size of an $\epsilon$-kernel $\calK_i$ is at most $O\bigl(\epsilon^{-(d-1)/2}\bigr)$.
Note that we can compute $\calK_i$ in $O(n+\e^{-(d-3/2)})$ time~\cite{Cha06,YAPV04}.
We take $O\bigl(\perror^{-2}\e^{-(d-1)}\log (1/\e)\bigr)$ samples in total. So the overall running time is $O\bigl(n\perror^{-2}\e^{-(d-1)}\log (1/\e)+\poly(\frac{1}{\e\perror})\bigr)
=\widetilde{O}\left(n\perror^{-2}\e^{-(d-1)}\right)$.
In summary, we obtain our main result for \probkernel\ in this subsection.

\begin{reptheorem}
{thm:probconstruction} (restated)
An \probkernel\ of size
$\tO\left(\perror^{-2}\e^{-3(d-1)/2}\right)$
can be constructed in $\widetilde{O}\left(n\perror^{-2}\e^{-(d-1)}\right)$
time, under both existential and locational uncertainty models.
\end{reptheorem}

\subsection{Improved \probkernel\ for Existential Models}
\label{subsec:probkernelexistential}

In this section, we show an \probkernel\ $\calS$ can be constructed in nearly linear time
for the existential model, and all points in $\calS$ are independent of each other.
The size bound $ \tO(\perror^{-2}\eps^{-(d-1)})$ (see Theorem~\ref{thm:probconstructionexsit})
is better than that in Theorem~\ref{thm:probconstruction} for the general case, and the independence property may be useful in certain applications.
Moreover, some of the insights developed in this section may be of independent interest
(e.g., the connection to Tukey depth).
Due to the independence requirement, the construction is somewhat more involved.
For ease of the description, we assume the Euclidean plane first.
All results can be easily extended to $\R^d$.
We also assume that all probability values are strictly between 0 and 1
and $0<\epsilon,\perror\leq 1/2$ is a fixed constant.

Let $\lambda(\calP)=\sum_{v\in \calP} (-\ln(1-p_v))$.
In the following, we present two algorithms.
The first algorithm works for any $\lambda(\calP)$ and
produces an \probkernel\ $\calS$ whose size depends on $\lambda(\calP)$.
In Section~\ref{subsec:2}, we present
the second algorithm that only works for $\lambda(\calP)\geq 3\ln(2/\perror)$
but produces an \probkernel\ $\calS$ with a constant size
(the constant only depends on $\epsilon$, $\perror$ and $\delta$).
Thus, we can get a constant size \probkernel\ by
running the first algorithm when $\lambda(\calP)\leq 3\ln(2/\perror)$
and running the second algorithm otherwise.

\subsubsection{Algorithm 1: For Any $\lambda(\calP)$}
\label{subsec:1}

In this section, we present the first algorithm
which works for any $\lambda(\calP)$.
%From now on, by Lemma~\ref{lm:probchange}, we can assume that $p_v\leq 1-\epsilon/4$.
We can think of each point $v$ associated with a Bernoulli random variable $X_v$
that takes value 1 with probability $p_v$ and 0 otherwise.
Now, we replace the Bernoulli random variable $X_v$ by a Poisson distributed random variable
$\tX_v$ with parameter $\lambda_v=-\ln(1-p_v)$ (denoted by $\pois(\lambda_v)$), i.e.,
$
\Pr[\tX_v=k] =\frac{1}{k!}\,\,\lambda_v^k\,\, e^{-\lambda_v}, \text{ for }k=0,1,2,\ldots.
$
Here, $\tX_v=k$ means that there are $k$ realized points located at the position of $v$.
We call the new instance {\em the Poissonized instance corresponding to $\calP$}.
We can check that $\Pr[\tX_v=0]=e^{-\lambda_v}=1-p_v=\Pr[X_v=0]$.
Also note that co-locating points do not affect any directional width,
so the Poissonized instance is essentially
equivalent to the original instance for our problem.

The construction of the \probkernel\ $\calS$ is as follows:
Let $\fP$ be the probability measure over all points in $\calP$ defined by
$\fP(\{v\})=\lambda_v/\lambda$ for every $v\in \calP$, where $\lambda:=\lambda(\calP)=\sum_{v\in \calP}\lambda_v$.
Let $\perror_1$ be a small positive constant
to be fixed later.
We take
$
N=\tilde{O}(\perror_1^{-2})
$
independent samples from $\fP$
(we allow more than one point to be co-located at the same position),
and let $\fQ$ be the empirical measure, i.e., each sample point
having probability $1/N$.
The coreset $\calS$ consists of the $N$ sample points in $\fQ$, each with the same
existential probability $1-\exp{(-\lambda/N)}$.
A useful alternative view of $\calS$ is to think of each point associated
with a random variable $Y_v$ following distribution $\pois(\lambda/N)$
(i.e., the Poissonized instance corresponding to $\calS$).
This finishes the description of the construction.

Now, we start the analysis.
Our goal is to show that $\calS$ is indeed an \probkernel.
The following theorem is a special case of Theorem~\ref{thm:quant2}
(specialized to the range space consisting of all halfplanes),
%, and restated with a tighter bound by Li, Long, and Srinivasan~\cite{LLS01}),
which shows that the empirical measure $\fQ$ is close to the original measure $\fP$
with respect to all half spaces.

\eat{
First, we provide a bound for $\lambda=\sum_v\lambda_v$, which will be useful for bounding the size of $\calS$.
\begin{lemma}
\label{lm:lambda}
Assuming $p_v\leq 1-\epsilon/4$ for all $v$, we have that
$\lambda =\sum_v \lambda_v \leq O(-\ln \epsilon \cdot \sum_v p_v)$.
\end{lemma}
\begin{proof}
For any $0<\alpha\leq x\leq 1$, it is easy to see that
$\ln (1-x) \geq x \frac{\ln \alpha}{1-\alpha}$.
\end{proof}
}

\begin{theorem}
\label{thm:vapnik}
{\em \cite{anthony2009neural,LLS01}}
We denote the set of all halfplanes by $\Halfplanes$.
With probability $1-\delta$, the empirical measure $\fQ$
(defined by $N=O(\perror_1^{-2}\log (1/\delta))$ independent samples)
satisfies the following:
$$
\sup_{H\in \Halfplanes}\,|\fP(H)-\fQ(H)| \leq \perror_1.
$$
\end{theorem}

From now on, we assume that $\fQ$ satisfies the statement of Theorem~\ref{thm:vapnik}.
We first observe a simple but useful lemma, which is a consequence of Theorem~\ref{thm:vapnik}.
For a halfplane $H$, we use $H\models 0$ to denote the event
that no point is realized in $H$.

\begin{lemma}
\label{lm:approxprob}
With probability $1-\delta$, for any halfplane $H\in \Halfplanes$, we have that
$$
\Pr_\calS\bigl[H\models 0\bigr]\in (1\pm O(\lambda\perror_1))\Pr_\calP\bigl[H\models 0\bigr].
$$
\end{lemma}
\begin{proof}
Fix an arbitrary halfplane $H\in \Halfplanes$.
Consider the Poissonized instance corresponding to $\calP$.
We first observe that
$
\Pr_{\calP}\bigl[H\models 0\bigr]
=\Pr_{\calP}\bigl[\sum_{v\in \calP\cap H} X_v=0\bigr].
$
Since $X_v$ follows distribution $\pois(\lambda_v)$,
$\sum_{v\in \calP\cap H} X_v$ follows Poisson distribution $\pois\bigl(\sum_{v\in \calP\cap H}\lambda_v\bigr)$.
Similarly, we have that $\Pr_{\calS}\bigl[H\models 0\bigr]
=\Pr_{\calS}\bigl[\sum_{v\in \calS\cap H} Y_v=0\bigr]$ since
$\sum_{v\in \calS\cap H} Y_v$ follows $\pois\bigl(\sum_{v\in \calS\cap H}\lambda/N\bigr)$.
Hence, we can see the following:
\begin{align*}
\Pr_{\calP}\bigl[H\models 0\bigr]
& =\exp{\bigl(-\sum_{v\in \calP\cap H}\lambda_v\bigr)} = \exp{\bigl(-\lambda\fP(H)\bigr)} \\
&\in \exp{\bigl(-\lambda(\fQ(H)\pm \perror_1)\bigr)}
= \exp{\bigl(-\sum_{v\in \calS\cap H}\frac{\lambda}{N}\pm \perror_1\lambda \bigr)} \\
& \in (1\pm O(\lambda\perror_1))\exp{\bigl(-\sum_{v\in \calS\cap H}\frac{\lambda}{N}\bigr)}
 = (1\pm O(\lambda\perror_1))\Pr_{\calS}\bigl[H\models 0\bigr].
\end{align*}
The first inequality follows from Theorem~\ref{thm:vapnik}
and the second is due to the fact that
$
e^{-\epsilon}\geq 1-\epsilon
$
and
$
e^{\epsilon}\leq 1+(e-1)\epsilon
$
for any $0<\epsilon<1$.
\end{proof}

For two real-valued random variables $X,Y$,
we define the Kolmogorov distance $\dist_K(X, Y)$
between $X$ and $Y$ to be
$
\dist_K(X, Y)
= \sup_{t\in \R} |\Pr[X\leq t]-\Pr[Y\leq t]|.
$
We also need the following simple lemma.

\begin{lemma}
\label{lm:kolsum}
Suppose we have four independent random variables $X$, $X'$, $Y$ and $Y'$
such that $\dist_K(X,X')\leq \varepsilon$ and $\dist_K(Y,Y')\leq \varepsilon$
for some $\varepsilon\geq 0$.
Then,
$\dist_K(X+Y,X'+Y')\leq 2\varepsilon$.
\end{lemma}
\begin{proof}
We need the following useful elementary fact about Kolmogorov distance:
Let $X, Y, Z$ be real-valued random variables such that $X$ is independent of $Y$ and independent of
$Z$. Then we have that $\dist_K(X + Y,X + Z)\leq \dist_K(Y,Z)$.
The rest of the proof is straightforward:
$
\dist_K(X+Y,X'+Y')\leq \dist_K(X+Y,X+Y')+\dist_K(X+Y',X'+Y')\leq 2\varepsilon.
$
The first inequality is the triangle inequality.
\end{proof}

Now, we are ready to show that $\calS$ is really an \probkernel.
We note that in this subsection our bound is stronger than \eqref{eq:quantcore} in that we do not need
to relax the length threshold.
We first prove the theorem under a simplified assumption:
we assume that there is a point $v^\star\in \R^2$ (not necessarily an input point), which we call the special point,
that lies in the convex hull of $\calP$ with probability at least $1-\delta/2$.
With the assumption, the proof is much simpler but still instructive as the analysis in Section~\ref{subsec:2} is an extension of this proof.
The general case is proved in Theorem~\ref{lm:boundprob2} and
the proof is more technical and the size bound is slightly worse.

\begin{theorem}
\label{lm:boundprob}
Assume that
there is a special point $v^\star\in \R^2$
that lies in the convex hull of $\calP$ with probability at least $1-\delta/2$.
%Given the $\beta$-assumption and a parameter $\delta \geq 6 (1-\beta)^{n/6}$,
The parameters of the algorithm are set as
\[
\perror_1=O\Bigl(\frac{\perror}{\lambda}\Bigr)
\quad\text{   and  }\quad
N=O\Bigl(\frac{1}{\perror_1^2}\log \frac{1}{\delta}\Bigr)
=O\Bigl(\frac{\lambda^2}{\perror^2}\log \frac{1}{\delta}\Bigr).
\]
With probability at least $1-\delta$, for any $t\geq 0$ and any direction $u$, we have that
\begin{align}
\label{eq:probcore1}
\Pr\Bigl[\dw(\calS,u)\leq t\Bigr] \in
\Pr\Bigl[\dw(\calP,u)\leq t\Bigr]\pm\perror.
\end{align}
\end{theorem}
\begin{proof}
We first condition on the event that
$v^\star$ is in the convex hull of all realized points
(which happens with probability at least $1-\delta/2$).   The remainder needs to hold with probability at least $1-\delta/2$.
Under the condition, we can pretend that $v^\star$ is a deterministic point in the original point set
(this does not affect any directional width as $v^\star$ is in the convex hull).

Fix an arbitrary direction $u$ (w.l.o.g., say it is the $x$-axis).
Rename all points as $v_1,v_2,\ldots,v_n$ according to the increasing
order of their projections to $u$.
Suppose $v^\star$ is renamed as $v_k$.
%Denote the $x$-coordinate of the special point $v^\star$ by $x^\star$.
%Suppose the $x$-coordinates of $\{v_1,\ldots,v_k\}$ are no larger than $x^\star$
%and the $x$-coordinates of $\{v_{k+1},\ldots,v_n\}$ are larger than $x^\star$.
Let the random variable $L$ be the directional width of $\{v_1,\ldots, v_k\}$ with respect to $u$
and $R$ be the directional width of $\{v_k,\ldots, v_n\}$ with respect to $u$.
Since $v^\star$ is assumed to be within the left and right extents, we can easily see that
$\dw(\calP, u)= L+R$.
Similarly, we define $L'$ ($R'$ resp.) to be the directional width of all points
in $\calS$ to the left (right resp.) of $v^\star$.
Since the convex hull of $\calS$ contains $v^\star$, we can also see that $\dw(\calS, u)= L'+R'$.
By Lemma~\ref{lm:approxprob}, we know that
$\dist_K(L,L')\leq O(\lambda\perror_1)$ and
$\dist_K(R,R')\leq O(\lambda\perror_1)$.
By Lemma~\ref{lm:kolsum}, we have that
$\dist_K(\dw(\calS, u), \dw(\calP, u))\leq O(\lambda\perror_1)$.
Let $\perror_1=O(\perror/\lambda)$,
the theorem follows.
\end{proof}

Now, we prove the theorem in the general case, where
the main difficulty comes from the fact that we can not separate the width into
two independent parts $L$ and $R$.
The proof is somewhat technical and can be found in Appendix~\ref{app:probkernel}.

\begin{theorem}
\label{lm:boundprob2}
Let $\perror_1=O(\frac{\perror}{\max\{\lambda,\lambda^2\}})$
  and
$N=O(\frac{1}{\perror_1^2}\log \frac{1}{\delta})
=O(\frac{\max\{\lambda^2,\lambda^4\}}{\perror^2}\log \frac{1}{\delta}).$
With probability at least $1-\delta$, for any $t\geq 0$ and any direction $u$, we have that
%\begin{align*}
$
\Pr\Bigl[\dw(\calS,u)\leq t\Bigr] \in
\Pr\Bigl[\dw(\calP,u)\leq t\Bigr]\pm \perror.
$
%\end{align*}
\end{theorem}

\subsubsection{Algorithm 2: For $\lambda(\calP)> 3\ln(2/\perror)$}
\label{subsec:2}

In the second algorithm, we assume that $\lambda(\calP)=\sum_{v\in \calP} \lambda_v>3\ln(2/\perror)$.
When $\lambda(\calP)$ is large, we cannot directly use the sampling technique
in the previous section since it requires a large number of samples.
However, the condition $\lambda(\calP)\geq 3\ln(2/\perror)$
implies there is a nonempty convex region $\calK$ inside the convex hull of $\calP$ with high probability.
Moreover, we can show the sum of $\lambda_v$ values in $\bcalK=\R^2\setminus \calK$ is small.
Hence, we can use the sampling technique just for $\bcalK$
and use the deterministic $\epsilon$-kernel construction for $\calK$.

Now, we describe the details of our algorithm.
Again consider the Poissonized instance of $\calP$.
Imagine the following process.
Fix a direction $u\in \S^1$.
\footnote{Here, $\S^1$ is the surface of the unit ball in $\R^d$.}
We move a sweep line $\ell_u$ orthogonal to $u$, along the direction $u$, to sweep through the points in $\calP$.
We use $H_u$ to denote the halfplane defined by $\ell_u$ (with normal vector $u$)
and $\bH_u$ denote its complement.
So $\calP(\bH_u)=\calP\cap \bH_u$ is the set of points that have been swept so far.
We stop the movement of $\ell_u$ at the first point such that $\sum_{v\in \bH_u} \lambda_v\geq \ln (2/\perror)$
(ties should be broken in an arbitrary but consistent manner).
%Let $H^-_u$ denote the open halfplane corresponding to $H_u$ (i.e., $H_u\setminus \ell_u$).
%Let $E_u$ denote the event that at least one point in $\bH_u$ shows up.
One important property about $\bH_u$ is that
$
\Pr[\bH_u\models 0]\leq  \perror/2.
$
We repeat the above process for all directions $u\in \S^1$ and let $\calH=\cap_u H_u$.
%and $\bH=\cup_u\bH_u=\R^2\setminus H$.
Since $\lambda(\calP)> 3\ln(2/\perror)$, by Helly's theorem, $\calH$ is nonempty.
A careful examination of the above process reveals
that $\calH$ is in fact a convex polytope and each edge of the polytope is defined by
two points in $\calP$.
\footnote{
This also implies that we only need to do the sweep for ${n\choose 2}$ directions.
In fact, by a careful rotational sweep, we only need $O(n)$ directional sweeps.
}
Moreover, $\calH$ is the region of points with Tukey depth at least $\ln(2/\perror)$.
\footnote{
The Tukey depth of a point $x\in\calP$ is defined
as the minimum total weight of points of $\calP$ contained in a closed halfspace
whose bounding hyperplane passes through $x$.
}

The construction of the \probkernel\ $\calS$ is as follows.
First, we use the algorithm in \cite{agarwal2004approximating} to find a
deterministic $\epsilon$-kernel $\calE_\calH$ of size $O(\epsilon^{-1/2})$ for $\calH$.
One useful property of the algorithm in \cite{agarwal2004approximating}
is that $\calE_\calH$ is a subset of the vertices of $\calH$.
Hence the convex polytope $\CH(\calE_\calH)$ is contained in $\calH$.
Since $\calE_\calH$ is an $\epsilon$-kernel, $(1+\epsilon)\CH(\calE_\calH))$ (properly shifted) contains $\calH$.
\footnote{
In fact, most existing algorithms (e.g.,~\cite{agarwal2004approximating})
identify a point in the interior of $\calH$ as origin, and compute an $\epsilon$-kernel $\calE_\calH$
such that $f(\calE_\calH, u)\geq \frac{1}{1+\epsilon}f(\calH,u)$ for all directions $u$.
So, $\calH\subseteq (1+\epsilon)\CH(\calE_\calH)$ since $f(\calH,u)\leq f((1+\epsilon)\CH(\calE_\calH), u)$ for all directions $u$.
}
Let $\calK=(1+\epsilon)\CH(\calE_\calH)$ and $\bcalK=\calP\setminus \calK$.
See Figure~\ref{fig:core}.

\begin{figure}[t]
\centering
\includegraphics[width=0.55\linewidth]{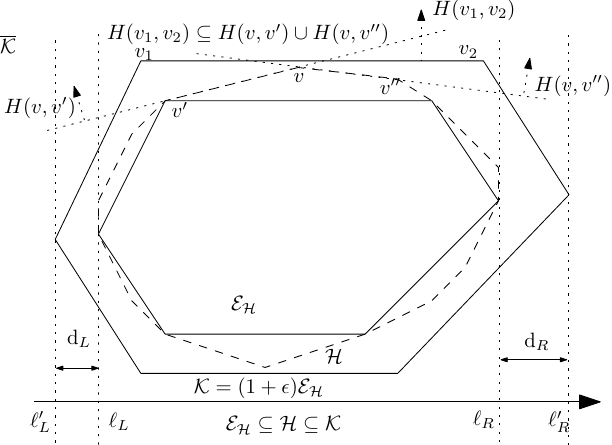}
\caption{The construction of the \probkernel\ $\calS$.
The dashed polygon is $\calH$. The inner solid polygon is $\CH(\calE_\calH)$
and the outer one is $K=(1+\epsilon)\CH(\calE_\calH)$.
$\bcalK$ is the set of points outside $\calK$.
}
\label{fig:core}
\end{figure}

Now, we apply the random sampling construction over $\bcalK$.
More specifically, let $\lambda:=\lambda(\bcalK)=\sum_{v\in \bcalK\cap\calP}\lambda_v$.
Let $\fP$ be the probability measure over $\calP\cap \bcalK$ defined by
$\fP(\{v\})=\lambda_v/\lambda$ for every $v\in \calP\cap \bcalK$.
Let $\perror_1=O(\perror/\lambda)$.
We take
$N=O(\perror_1^{-2}\log (1/\delta))$
independent samples from $\fP$
and let $\fQ$ be the empirical distribution with
each sample point
having probability $1/N$.
The \probkernel\ $\calS$ consists of the $N$ points in $\fQ$, each with the same
existential probability $1-\exp{(-\lambda/N)}$, as well as
all vertices of $\calK$, each with probability $1$.
This finishes the construction of $\calS$.

Now, we show that the size of $\calS$ is constant (only depending on $\epsilon$ and $\delta$),
which is an immediate corollary of the following lemma.
\begin{lemma}
\label{lm:lambdasum}
$
\lambda=\lambda(\bcalK)=\sum_{v\in \bcalK} \lambda_v =O\Bigl(\frac{\ln 1/\perror}{\sqrt{\epsilon}}\Bigr).
$
\end{lemma}

\begin{proof}
%By Lemma~\ref{lm:lambda}, it suffices to show
%that $\sum_{v\in \bcalK} p_v =O(1)$.

We can see that $\bcalK$ is the union of $O(\epsilon^{-1/2})$ half-planes,
each defined by a segment of $\calK$.
It suffices to show the sum of $\lambda_v$ values in each half-plane is $O(\ln (1/\perror))$.
Consider the half-plane $H(v_1,v_2)$ defined by segment $(v_1,v_2)$ of $\calK$.
Suppose $v$ is the vertex of $\calH$ that is closest to the line $(v_1,v_2)$.
Let $(v',v)$ and $(v,v'')$ be the two edges of $\calH$ incident on $v$.
Clearly, $H(v',v)\cup H(v,v'')$,
the union of the two half-planes defined by $(v',v)$ and $(v,v'')$,
strictly contains $H(v_1,v_2)$. See Figure~\ref{fig:core} for an illustration.
Hence,  $\sum_{v\in H(v_1,v_2)}\lambda_v$ is at most $2\ln (1/\perror)$.
\end{proof}

Now, we prove the main theorem in this section. The proof is an extension
of Theorem~\ref{lm:boundprob}. Here, the set $\calK$ plays
a similar role as the special point $v^\star$ in Theorem~\ref{lm:boundprob}.
Unlike Theorem~\ref{lm:boundprob}, we also need to relax the length threshold here,
which is necessary even for deterministic points.
\begin{theorem}
\label{lm:boundprob3}
Let $\lambda=\lambda(\bcalK)$ and $\perror_1=O(\perror/\lambda)$,
and
$
N=O\Bigl(\frac{1}{\perror_1^{2}}\log \frac{1}{\delta}\Bigr)
=O\Bigl( \frac{\ln^2 1/\perror}{\epsilon \perror^2} \log \frac{1}{\delta} \Bigr).
$
With probability at least $1-\delta$, for any $t\geq 0$ and any direction $u$, we have that
\begin{align}
\label{eq:probcore2}
\Pr\Bigl[\dw(\calS,u)\leq t\Bigr] \in
\Pr\Bigl[\dw(\calP,u)\leq (1\pm\epsilon)t\Bigr]\pm\perror.
\end{align}
\end{theorem}

\begin{proof}
The proof is similar to Theorem~\ref{lm:boundprob}.
Fix an arbitrary direction $u$ (w.l.o.g., say it is the $x$-axis).
Rename all points in $\calP$ as $v_1,v_2,\ldots,v_n$ according to the increasing
order of their $x$-coordinates.
We use $x(v_i)$ to denote the $x$-coordinate of $v_i$.
Let $\ell_L$ (or $\ell_R$) be the vertical line that passes the leftmost endpoint of $\calE_\calH$
(or the rightmost endpoint of $\calE_\calH$).
We use $x(\ell_L)$ (or $x(\ell_R)$) to denote the $x$ coordinate of $\ell_L$ (or $\ell_R$)
and let $\dist(\ell_L,\ell_R)=|x(\ell_L)-x(\ell_R)|$.
Suppose that $v_1,\ldots, v_k$ lie to the left of $\ell_L$
and $v_r,\ldots, v_n$ lie to the right of $\ell_R$.
Let the random variable $L=x(\ell_L)-f(\{v_1,\ldots, v_k\},-u)$
and $R=f(\{v_r,\ldots, v_n\},u)-x(\ell_R)$.
Let $W=L+R+\dist(\ell_R,\ell_L)$.
We can see that $W$ is close to $\dw(\calP, u)$ in the following sense.
Let $E$ denote the event that
at least one point in $\{v_1,\ldots, v_k\}$ is present
and at least one point in $\{v_r,\ldots, v_n\}$ is present.
Conditioning on $E$,
$W$ is exactly $\dw(\calP,u)$.
Moreover, we can easily see $\Pr[E] \geq (1-\perror/2)^2 \geq 1-\perror$.
Hence, we have
\begin{align*}
\Pr[W\leq t]-\perror & \leq (1-\perror) \Pr[W\leq t]
\leq \Pr[\dw(\calP, u)\leq t \mid E]\Pr[E] \\
& \leq \Pr[\dw(\calP, u)\leq t]= \Pr[\dw(\calP, u)\leq t \mid E]\Pr[E]
+\Pr[\dw(\calP, u)\leq t \mid \neg E]\Pr[\neg E] \\
& \leq  \Pr[W\leq t] +\perror.
\end{align*}

Similarly,
we let $\ell'_L$ (or $\ell'_R$) be the vertical line that passes the leftmost endpoint of $\calK$
(or the rightmost endpoint of $\calK$).
Suppose that $v'_1,\ldots, v'_j$ (points in $\calS$) lie to the left of $\ell_L$
and $v'_s,\ldots, v'_N$ lie to the right of $\ell_R$.
We define $L'= x(\ell'_L)-f(\{v'_1,\ldots,v'_j\}, -u)$
 and $R'=f(\{v'_s,\ldots, v'_N\}, u)-x(\ell'_R)$.
We can also see that $\dw(\calS, u)= L'+R'+\dist(\ell_R,\ell_L)$.

Let $\dist_L=x(\ell_L)-x(\ell'_L)$ and $\dist_R=x(\ell'_R)-x(\ell_R)$.
Let $H_t$ be the half-plane $\{(x,y)\mid x\leq x(\ell'_L)-t\}$.
We can see that for any $t\geq 0$,
\begin{align*}
\Pr[L\leq t+\dist_L] -O(\lambda\perror_1)
& = \Pr[X(H_t)= 0]-O(\lambda\perror_1) \\
& \leq \Pr[L'\leq t]=\Pr[Y(H_t)= 0] \leq \Pr[X(H_t)= 0]+O(\lambda\perror_1) \\
& =\Pr[L\leq t+\dist_L] +O(\lambda\perror_1),
\end{align*}
where the inequalities hold due to Lemma~\ref{lm:approxprob}.
Similarly, we can see that for any $t\geq 0$,
\begin{align*}
\Pr[R\leq t+\dist_R] -O(\lambda\perror_1)
 \leq \Pr[R'\leq t]
\leq \Pr[R\leq t+\dist_L] +O(\lambda\perror_1).
\end{align*}
Therefore, by Lemma~\ref{lm:kolsum}, we have that for any $t>0$,
\begin{align*}
\Pr[L+R\leq t+\dist_L+\dist_R] -O(\lambda\perror_1)
 \leq \Pr[L'+R'\leq t]
\leq \Pr[L+R\leq t+\dist_L+\dist_R] +O(\lambda\perror_1).
\end{align*}
Therefore, we can conclude that for any $t\geq \dist(\ell'_L,\ell'_R)$,
\begin{align*}
 \dw(\Pr[\calS, u)\leq t] & =
 \Pr[L'+R'+\dist(\ell'_L,\ell'_R)\leq t] \\
& \in \Pr[L+R+\dist(\ell'_L,\ell'_R)\leq t+\dist_L+\dist_R] \pm O(\lambda\perror_1) \\
& = \Pr[L+R+\dist(\ell_L,\ell_R)\leq t] \pm O(\lambda\perror_1) \\
& = \Pr[W\leq t] \pm O(\lambda\perror_1) \\
&= \Pr[\dw(\calP, u)\leq t]\pm O(\lambda\perror_1+\perror).
\end{align*}
Noticing that
$
\perror  \geq \Pr[\dw(\calP,u)\leq \dist(\ell_L,\ell_R)]
 \geq \Pr[\dw(\calP,u)\leq (1-\epsilon)\dist(\ell'_L,\ell'_R)],
$
we can obtain that, for any $t< \dist(\ell'_L,\ell'_R)$,
%\begin{align*}
$\Pr[\dw(\calS, u)\leq t] =0
\geq \Pr[\dw(\calP,u)\leq (1-\epsilon)t] - \perror.
$
%\end{align*}
Moreover, it is trivially true that
%\begin{align*}
 $\Pr[\dw(\calS, u)\leq t]=0
\leq \Pr[\dw(\calP,u)\leq (1-\epsilon)t] + \perror.$
%\end{align*}
The proof is completed.
\end{proof}

\topic{Higher Dimensions:}
Our constructions can be easily extended to $\R^d$ for any constant $d>2$.
The sampling bound (Theorem~\ref{thm:vapnik}) still holds if the number of samples is
$O(d\perror_1^{-2}\log (1/\delta))=O(\perror_1^{-2}\log (1/\delta))$.
Hence, Theorem~\ref{lm:boundprob} and Theorem~\ref{lm:boundprob2} hold
with the same parameters ($d$ is hidden in the constant).
In order for Algorithm 2 to work, we need $\lambda(\calP)> (d+1)\ln(2/\perror)$
to ensure $\calH$ is nonempty.
Instead of constructing an $\epsilon$-kernel $\calE_\calH$ with $O(\epsilon^{-(d-1)/2})$ vertices,
we construct a convex set $\calK$ which is the intersection of $O(\epsilon^{-(d-1)/2})$ halfspaces and satisfies
$(1-\epsilon)\calK \subseteq \calH \subseteq \calK$ (this can be done by either working with the dual,
or directly using
the construction implicit in \cite{dudley1974metric}).

Now, we briefly sketch how to compute such $\calK$ using the dual approach.
We first compute the dual $\dualH$ of $\calH$ in $\R^d$.
Recall the dual (also called the polar body) $\dualH$ of $\calH$ is
defined as the set $\{x\in \R^d \mid \innerprod{x}{y}\leq 1, y\in \calH\}$.
$\dualH$ has $O(n^d)$ vertices (each corresponding to a face of $\calH$).
Then, compute an $\epsilon$-kernel $\dualE_{\dualH}$ with $O(\epsilon^{-(d-1)/2})$ vertices for $\dualH$.
Taking the dual of $\dualE_{\dualH}$ gives the desired $\calK$, which is an intersection of $O(\epsilon^{-(d-1)/2})$ halfspaces
(each corresponding to a point in $\dualE_{\dualH}$).
The correctness can be easily seen by an argument through the gauge function
$g(\dualE_{\dualH}, x)=\min\{\lambda\geq 0\mid  x\in \lambda\dualE_{\dualH} \}$.
Since $\dualE_{\dualH} \subseteq \dualH \subseteq (1+\epsilon)\dualE_{\dualH}$, we can see that
$\frac{1}{1+\epsilon}g(\dualE_{\dualH}, x)=g((1+\epsilon)\dualE_{\dualH}, x) \leq g(\dualH,x) \leq g(\dualE_{\dualH}, x)$.
The correctness follows from the duality between the gauge function and the support function, which says
$g(\dualE_{\dualH}, x)=f(\calK, x)$
and  $g(\dualH,x)=f(\calH, x)$ for all $x\in \mathbb{S}^{d-1}$ (see e.g., \cite{schneider1993convex}).

%The number of vertices of $\calE_\calH$ can be bounded by $O(1/\epsilon^{(d-1)/2})$,
%so the the number of faces of $\calK$ can be bounded by $O(1/\epsilon^{d})$ (by McMullen's upper bound theorem).
We generalize Lemma~\ref{lm:lambdasum} for $\R^d$ by the following lemma.
\begin{lemma}
\label{lm:lambdasumhigh}
There is a convex set $\calK$, which is an intersection of $O(\epsilon^{-(d-1)/2})$ halfspaces and satisfies
$(1-\epsilon)\calK \subseteq \calH \subseteq \calK$. Moreover, we have that $\lambda(\bcalK)= O(\epsilon^{-(d-1)/2}\ln (1/\perror))$.
\end{lemma}
Plugging the new bound of $\lambda(\bcalK)$, we can see that it is enough to set
$$N=O(\perror^{-2}\epsilon^{-(d-1)}\log \frac{1}{\delta} \polylog \frac{1}{\perror})
=\tO(\epsilon^{-(d-1)}\perror^{-2}).$$
%which is the upper bound of the size of the \probkernel\ from Theorem \ref{lm:boundprob2}.

\topic{Running Time:}
The algorithm in Section~\ref{subsec:1} takes only $O(Nn)$ time (where $N$ is the size of the kernel, which is constant if $\epsilon$, $\perror$ and $\lambda$ are constant).
The algorithm in Section~\ref{subsec:2} is substantially slower. The most time consuming part is the construction of $\calH$, which is the intersection of all halfspaces.
%A straightforward implementation takes $O(n^3 \log n)$ time in $\R^2$:
%We only need to sweep ${n\choose 2}$ directions, and each sweeping takes $O(n \log n)$ time.
%Constructing a deterministic kernel for $\calH$ only takes linear time (linear in the size of $\calH$),
%using the algorithm in \cite{agarwal2004approximating}.
In $\R^d$, we need to sweep $O(n^d)$ directions (each determined by $d$ points). So the polytope $\calH$ may have
$O(n^d)$ faces.
Using the dual approach, we can compute $\calK$ in $O(n^d)$ time
(linear in the number of points in the dual space) as well.
Overall, the running time is $O(n^{d})$.

\subsubsection{A Nearly Linear Time Algorithm for Constructing \probkernel s}
\label{subsec:linearprobkernel}
We describe a nearly linear time algorithm for constructing an \probkernel\
in the existential uncertainty model.
As mentioned before,
the algorithm in Section~\ref{subsec:1} takes linear time.
So we only need a nearly linear time algorithm for constructing $\calH$ (and $\calK$).
Note that $\calH$ is the set of points in $\R^d$ with Tukey depth at least $\ln(2/\perror)$.
One tempting idea is to utilize the notion of $\epsilon$-approximation (which can be obtained by sampling) to compute
the approximate Tukey depth for the points, as done in \cite{matousek1991computing}.
However, a careful examination of this approach shows that the sample size needs to be as large
as $O(\lambda(\calP))$ (to ensure that for every halfspace, the difference between the real weight and the sample weight is less than, say
$0.1 \ln(2/\perror)$).
Another useful observation is that only points with small (around $\ln(2/\perror)$) Turkey depth are relevant in constructing $\calH$.
Hence, we can first sample an $\epsilon$-approximation of very small size (say $k=O(\log n)$), and use it to quickly identify the region $\calH_1$ in which all points have large (i.e., $\lambda(\calP)/k$) Turkey depth
(so $\calH_1\subseteq \calH$).
Then, we can delete all points inside $\calH_1$ and focus on the remaining points.
Ideally, the total weight of the remaining points can be reduced significantly and a random sample of the same size $k$
would give an $\epsilon'$-approximation of the remaining points for some $\epsilon'<\epsilon$.
We repeat the above until the total weight of the remaining points reduces to a constant,
and then a constant size sample suffices.
However, it is possible that all points have fairly small Tukey depth (consider the case where all points are in convex position),
and no point can be removed.
To resolve the issue, we use the idea in Lemma~\ref{lm:lambdasum}:
there is a convex set $\calK_1$ slightly larger than $\calH_1$ such that the weight of points outside $\calK_1$ is much smaller.
Hence, we can make progress by deleting all points inside $\calK_1$.
Since $\calK_1$ is only slightly larger than $\calH_1$, we do not lose too much in terms of the distance.
Our algorithm carefully implements the above iterative sampling idea.

For ease of exposition, we first focus on $\R^2$.
Consider the Poissonized instance of $\calP$.
We would like to  find two convex sets $\calH$ and $\calK$ satisfying the following properties.
\begin{enumerate}
\item [P1.]
Assume without loss of generality that the origin is in $\calH$.
We require that $\frac{1}{1+\epsilon}\calK\subseteq \calH\subseteq \calK$.
\item [P2.]
%We move a sweep line $\ell_u$ orthogonal to $u$, along the direction $u$, to sweep through the points in $\calP$,
%until we hit $\calH$ (i.e., $\ell_u$ is a supporting line of $\calH$).
%Again, let $H_u$ denote the open halfplane defined by $\ell_u$ (with normal vector $u$)
%and $\bH_u$ its complement.
For a direction $u\in \mathbb{S}^1$, we use $H(\calH, u)$ to denote the halfplane which does not contain $\calH$
and whose boundary is
the supporting line of $\calH$ with normal direction $u$.
We require that  $\lambda(H(\calH,u))=\sum_{v\in H(\calH,u)} \lambda_v\geq \ln (2/\perror)$
for all directions $u\in \mathbb{S}^1$.
\item [P3.]
$\lambda(\bcalK) = \tO(1/\sqrt{\e})$.
\end{enumerate}
By a careful examination of our analysis in Section~\ref{subsec:2}, we can see the above properties are all we need
for the analysis.

Let $\calH^\star$ denote the $\calH$ found using the exact algorithm in Section~\ref{subsec:2}.
We use the following set of parameters:
$$
z=O(\log n),\quad \e_1=O\left(\frac{\e}{\log n}\right), \quad \epsilon_2=O\left(\sqrt{\frac{\epsilon}{\log n}}\right).
$$
Our algorithm proceeds in rounds.
Initially, let $\calH_0=\CH(\{ v\in \calP \mid \lambda_v \geq \ln(2/\perror)\})$.
In round $i$ (for $1\leq i\leq z$),
we construct two convex sets $\calH_i$ and $\calK_i$
such that
\begin{enumerate}
\item
$\calH_0\subseteq \calK_0\subseteq \calH_1\subseteq \calK_1\subseteq   \ldots \subseteq \calH_z \subseteq \calK_z$;
\item
$\frac{1}{1+\e_1}\calK_i\subseteq \calH_i\subseteq \calK_i$
($\calK_i$ and $\calH_i$ are very close to each other);
\item
$\frac{1}{(1+\e_1)^i}\calH_i \subseteq \calH^\star$
($\calH_i$ is almost contained in $\calH^\star$);
\item
$\lambda(\calP\cap\bcalK_i)\leq \frac{1}{2} \lambda(\calP\cap\bcalK_{i-1})$
(the total weight outside $\calK_i$ reduces by a factor of at least one half).
\end{enumerate}
We repeat the above process
until
$\lambda(\bcalK_i)\leq\tO(1/\sqrt{\e})$.

Before spelling out the details of our algorithm, we need a few definitions.
\begin{definition}
For a set $P$ of weighted points in $\R^d$, we use $\TK(P, \gamma)$ to denote the set of points $x\in \R^d$
with Tukey depth at least $\gamma$.
It is known that $\TK(P, \gamma)$ is convex (see e.g., \cite{matousek1991computing}).
By this definition, $\calH^\star=\TK(\calP, \ln(2/\perror))$.
\end{definition}

%\begin{definition}
%For a set $P$ of weighted points (with weight $\lambda$) in $\R^d$,
%a set $\calE$ of weighted points (with weight $\lambda'$) in $\R^d$ is an $\epsilon$-approximation of $P$ if
%$$
%|\lambda(H)-\lambda'(H)| \leq \epsilon,\,\,  \text{ for every halfspace } H.
%$$
%\end{definition}
\noindent
Recall the definition of $\epsilon$-approximation from Theorem~\ref{thm:quant2}.
By Theorem~\ref{thm:quant2} (or Theorem~\ref{thm:vapnik}),
we can see that a set of $O(\epsilon^{-2}\log(1/\delta))$ sampled points is an $\epsilon \lambda(P)$-approximation with probability $1-\delta$.

We are ready to describe the details of our algorithm.
Initially $\calH_0=\CH(\{ v\in \calP \mid \lambda_v \geq \ln(2/\perror)\})$ (obviously $\calH_0\subseteq \calH^\star$).
Compute a deterministic $\e_1$-kernel $C_{\calH_0}$ of $\calH_0$
and let $\calK_0=(1+\e_1)\CH(C_{\calH_0})$.
Delete all point in $\calP\cap \calK_0$ and let $\calP_1$ be the remaining points in $\calP$ (i.e., $\calP_1=\calP\cap\bcalK_0$).
Let $\ver(\calK_0)$ denote all vertices of $\calK_0$ (notice that some of them may not be original points in $\calP$).
\eat{
In the first round, we compute an $\perror_1\lambda(\calP)$-approximation $\calE_1$ (of size $L=O(1/\perror_1^2)$) of $\calP$ (w.r.t. the weight $\lambda_v$):
$\calE$ is a set of points with uniform weight $\lambda_1$, such that with probability $1-\delta$,
$
|\lambda(H)-\lambda_1(H)| \leq \perror_1 \lambda(\calP),\, \text{ for every halfspace } H.
$
Compute $H_1=\TK(\calE_1, \epsilon_1 \lambda(\calP))$ for $\epsilon_1=...$.
Using the brute-force algorithm mentioned in Section~\ref{subsec:2}, this can be done in $O(\log^{O(d)} n)$ time.
We can see that
$$
H_1\subseteq \TK(\calP, (\epsilon_1-\perror_1)\lambda(\calP)) \subseteq \TK(\calP, \ln(2/\epsilon))= \calH^\star.
$$
Compute a deterministic $\epsilon_3$-kernel $C_{H_1}$ of $H_1$
and let $K_1=(1+\epsilon)\CH(C_{H_1})$.
Using the same proof as Lemma~\ref{xxx}, we can see that
$\lambda(\bcalK_1)\leq O((\epsilon_1+\perror_1)\lambda(\calP)/\sqrt{\epsilon_1}) \leq \lambda(\calP)/2$.
Then, we delete all points in $\calP\cap K_1$ and add all vertices of $K_1$ (denoted as $\ver(K_1)$).
}

Now, suppose we describe the $i$th round for general $i>1$.
We have the remaining vertices in $\calP_{i-1}$ and $\ver(\calK_{i-1})$.
Let each point  $v\in \calP_i$ has the same old weight $\lambda_v$
and each point in $\ver(\calK_{i-1})$ has weight $+\infty$ (to make sure every point in $\calK_{i-1}$ has Turkey depth $+\infty$).
Using random sampling on $\calP_i$, obtain an $\epsilon_2 \lambda(\calP_i)$ -approximation $\calE_i$
(of size $L=O(\epsilon_2^{-2}$) for $\calP_i$.
Then compute (using the brute-force algorithm described in Section~\ref{subsec:2})
$$
\calH_i=\TK(\calE_i\cup \ver(\calK_{i-1}), \max\{4\epsilon_2 \lambda(\calP_i),2\ln(2/\perror)\}).
$$
Note that $\calK_{i-1}\subseteq \calH_i$.
Compute a deterministic $\epsilon_2$-kernel $C_{\calH_i}$ of $\calH_i$
and let $\calK_i=(1+\epsilon_1)\CH(C_{\calH_i})$ (hence $\CH(C_{\calH_i})\subseteq \calH_i\subseteq \calK_i$).
Then, we delete all points in $\calP\cap \calK_i$ and add all vertices of $\calK_i$ (denoted as $\ver(\calK_i)$).
Let $\calP_{i+1}$ be the remaining points in $\calP\cap \bcalK_{i}$.

Our algorithm terminates when $\lambda(\calP_i)\leq\tO(1/\sqrt{\e})$.
Suppose the last round is $z$.
Finally, we let $\calH=\frac{1}{(1+\e_1)^z} \calH_z$ and $\calK=\calK_z$.

First we show the algorithm terminates after at most a logarithmic number of rounds.

\begin{lemma}
\label{lm:boundz}
$z= O(\log n)$.
\end{lemma}
\begin{proof}
If $
\calH_i=\TK(\calE_i\cup \ver(\calK_{i-1}), 2\ln(2/\perror))
$, then we stop since $\lambda(\calP_{i+1})\leq \tilde{O}(1/\sqrt{\e})$ by Lemma~\ref{lm:lambdasum}. Thus, we only need to bound the number of iterations where $
\calH_i=\TK(\calE_i\cup \ver(\calK_{i-1},4\epsilon_2 \lambda(\calP_i)).
$

Initially, it is not hard to see that $\lambda(\calP_1)\leq n\ln(2/\perror)$. Using Lemma~\ref{lm:lambdasum}, we can see that
$\lambda(\calP_{i+1})=\lambda(\calP\cap \bcalK_i)\leq O(5\epsilon_2\lambda(\calP_i)/\sqrt{\e_1}) \leq \lambda(\calP_i)/2$ for the constant defining $\e_1$ sufficiently large.
Hence, $\lambda(\calP_i) \leq \lambda(\calP)/2^i$.
\end{proof}

We need to show $\calH$ and $\calK$ satisfy P1, P2 and P3.
P3 is quite obvious by our algorithm.
It is also not hard to see P1 since $(1+\e_1)^{z+1}\leq 1+\e$ and
$$
\frac{1}{(1+\e)} \calK_z\subseteq \frac{1}{(1+\e_1)^{z+1}} \calK_z\subseteq \frac{1}{(1+\e_1)^z} \calH_z= \calH \subseteq \calH_z\subseteq \calK_z = \calK.
$$
The most difficult part is to show P2 holds:
For every direction $u\in \mathbb{S}^1$,
$\lambda(H(\calH, u))=\sum_{v\in H(\calH, u)} \lambda_v\geq \ln (2/\perror)$.
In fact, we show that $\calH \subseteq \calH^\star$, from which P2 follows trivially, which suffices to prove the following lemma.

\begin{lemma}
$\frac{1}{(1+\e_1)^i}\calH_i \subseteq \calH^\star$ for all $0\leq i\leq z$.
In particular, $\calH\subseteq \calH^\star$.
\end{lemma}
\begin{proof}
We prove the lemma by induction.
$\calH_0\subseteq \calH^\star$ clearly satisfies the lemma.
For ease of notation, we let $\eta=1/(1+\e_1)$.
Suppose the lemma is true for $\calH_{i-1}$, from which we can see that
$$
\eta^i \calK_{i-1}\subseteq \eta^{i-1}\calH_{i-1}\subseteq \calH^\star.
$$
Now we show the lemma holds for $\calH_i$.
Consider the $i$th round.
Let $\calE_i$ be
an $\epsilon_2 \lambda(\calP_i)$ -approximation for $\calP_i=\calP \cap\bcalK_{i-1}$
and
$\calH_i=\TK(\calE_i\cup \ver(\calK_{i-1}), 4\epsilon_2 \lambda(\calP_i))$.
Fix an arbitrary direction $u\in \mathbb{S}^1$ (w.l.o.g., assume that $u=(0,-1)$, i.e., the downward direction),
let $H(\eta^i\calH_i,u)$ be a halfplane whose boundary is tangent to $\eta^i \calH_i$.
It suffices to show that
$\lambda(H(\eta^i\calH_i,u))=\sum_{v\in H(\eta^i\calH_i,u)} \lambda_v\geq \ln (2/\perror)$.
We move a sweep line $\ell_u$ orthogonal to $u$, along the direction $u$ (i.e., from top to bottom), to sweep through the points in $\calP_i\cap \ver(\calK_{i-1})$
until the total weight we have swept is at least $\ln(2/\perror)$.
We distinguish two cases:
\begin{enumerate}
\item
$\ell_u$ hits a point $v$ in $\ver(\calK_{i-1})$ (recall the weight for such point is $+\infty$).
We can see that $v$ is the topmost point of $\calK_{i-1}$ and $\calH_i$ (or equivalently, $\ell_u$ is also a supporting line for $\calH_i$).
Since $\eta^i\calK_{i-1}\subseteq  \calH^\star$ by the induction hypothesis,
the topmost point of $\eta^i\calK_{i-1}$ is lower than that for $\calH^\star$.
The topmost point of $\eta^i\calK_{i-1}$ is also the highest point of $\eta^i\calH_i$,
from which we can see
$H(\eta^i\calH_i,u)$ is lower than $H(\calH^\star, u)$, which implies that
$\lambda(H(\eta^i\calH_i,u))\geq \ln (2/\perror)$.
\item
$\ell_u$ stops moving when it hits an original point in $\calP_i$.
Since $\max\{3\epsilon_2 \lambda(\calP_i),2\ln (2/\perror)-\epsilon_2 \lambda(\calP_i)\}> \ln (2/\perror)$,
by definition of $\calH_i$, $H(\calH_i, u)$ can not be higher than $\ell_u$.
The boundary of  $H(\eta^i\calH_i, u)$ is even lower, from which we can see $\lambda(H(\eta^i\calH_i, u))\geq \ln (2/\perror)$.
\end{enumerate}
Hence, every point in $\eta^i\calH_i$ has Tukey depth at least $\ln (2/\perror)$, which implies the lemma.
\end{proof}

\topic{Running time:}
In each round, we compute in linear time an $\epsilon_2$-approximation $\calE_i$ of size $O(\epsilon^{-2}_2\log (1/\delta)) = \polylog(n)$ (with
$\delta=\poly(n)$ to ensure each probabilistic event succeeds with high probability).
$\calK_i$ is a dilation of an $\e_1$-kernel. So the size of $\ver(\calK_i)$ is at most $1/\sqrt{\e_1}=O(\log^{1/2} n)$.
Deciding whether a point is inside $\calK_i$ can be solved in $\polylog(n)$ time, by a linear program with $|\ver(\calK_i)|$ variables.
To compute $\calH_i$, we can use the brute-force algorithm described in Section~\ref{subsec:2}, which takes
$\poly(|\calE_i\cap \ver(\calK_{i-1})|)=\polylog(n)$ time.
There are logarithmic number of rounds. So the overall running time is $O(n\polylog n)$.

\topic{Higher Dimension:}
Our algorithm can be easily extended to $\R^d$ for any constant $d>2$.
In $\R^d$, we let
$\e_1=O(\e/\log n)$ and $\epsilon_2=O(\epsilon/\log n)^{(d-1)/2}$.
With the new parameters, we can easily check that Lemma~\ref{lm:boundz} still holds.
We can construct an \probkernel\ of size
$\min\{O\Bigl(\perror^{-2}\max\{\lambda^2,\lambda^4\}\log (1/\delta)\Bigr),
O(\perror^{-2}\epsilon^{-(d-1)}\log (1/\delta) \polylog (1/\perror)\}$.
The first term is from Theorem~\ref{lm:boundprob2}
and the second from the higher-dimensional extension to Theorem~\ref{lm:boundprob3}.
Now, let us examine the running time.
In $\R^d$, $|\ver(\calK_i)|$ is at most $\e_1^{-(d-1)/2}=O(\log^{(d-1)/2}n)$.
So deciding whether a point is inside $\calK_i$ can be solved in $\log^{O(d)}(n)$ time.
Computing $\calH_i$ takes $\log^{O(d)}(n)$ time using the brute-force algorithm.
So the overall running time is $O(n\log^{O(d)} n)$.

In summary, we obtain the following theorem for \probkernel.

\begin{reptheorem}
{thm:probconstructionexsit} (restated)
	$\calP$ is a set of uncertain points in $\R^d$ with existential uncertainty. Let $\lambda=\sum_{v\in \calP}(-\ln (1-p_v))$.
	There exists an \probkernel\ for $\calP$,
	which consists of a set of
	independent uncertain points of cardinality $\min\{\tO(\perror^{-2}\max\{\lambda^2,\lambda^4\}),
	\tO(\eps^{-(d-1)}\perror^{-2})\}$.
	The algorithm for constructing such a coreset runs in  $\tO(n\log^{O(d)} n)$ time.
\end{reptheorem}

\subsection{\probkernel\ Under the Subset Constraint}
\label{sec:expsubsetprobkernel}

We show it is possible to construct an \probkernel\ in the existential model under the {\em $\beta$-assumption}:
each possible location realizes a point with a probability at least $\beta$, where $\beta >0$ is some fixed constant.
%, and in fact for the locational model if each point can take one of $k$ locations, for a constant $k$,
%it is often assumed that each $p_{vs} = 1/k$.)
\begin{theorem}
\label{thm:betaprob}
Under the $\beta$-assumption, there is an \probkernel\ in $\R^d$, which is
of size $O(\mu^{-(d-1)/2} \log (1/\mu) )$ and satisfies the subset constraint,
in the existential uncertainty model, where $\mu=\min\{\e,\perror\}$.
\end{theorem}
In fact, the algorithm is exactly the same as constructing an \expkernel\
and the proof of the above theorem is implicit in the proof of Theorem~\ref{thm:beta}.

\section{\exprkernel\ Under the $\beta$-Assumption}
\label{sec:rfunction}

In this section, we show an \exprkernel\ exists in the existential uncertainty model
under the $\beta$-assumption.
%We assume that all points in $\calP$ are present in $\calR^d$. For any $u,v\in \calR_+^d$, we use $\langle u,v\rangle$ to denote the usual inner product $\sum_{i=1}^d u_i v_i$. For ease of notation, we write $v\succ_u w$ as a shorthand notation for $\langle u,v\rangle>\langle u,w\rangle$. For any $u\in \calR^d$, the binary relation $\succ_u$ defines a total order of all vertices in $\calP$. We call this order the canonical order of $\calP$ with respect to $u$. For any two points $u$ and $v$, we use $\dist(u,v)$ or $\parallel v-u\parallel$ to denote their Euclidean distance.
Recall that the function $T_r(P,u)=\max_{v\in P}\innerprod{u}{v}^{1/r}-\min_{v\in P}\innerprod{u}{v}^{1/r}$.
For ease of notation,
we write
$\Exp[T_r(\calP,u)]$ to denote
$\Exp_{P\sim \calP}[T_r(P,u)]$.
Our goal is to find a set $\calS$ of stochastic points such that for all directions $u\in \calP^{\polar}$, we have that
$\Exp[T_r(\calS,u)]\in (1\pm\e)\Exp[T_r(\calP,u)]$.

Our construction of $\calS$ is almost the same as that
in Section~\ref{sec:qkernel}.
Suppose we sample $N$ (fixed later) independent realizations
and take the $\e_0$-kernel for each of them.
Suppose they are $\{\calE_1,\ldots,\calE_N\}$
and we associate each a probability $1/N$.
We denote the resulting \exprkernel\ by $\calS$.
Hence, for any direction $u\in \calP^{\polar}$,
$\Exp[T_r(\calS,u)]=\frac{1}{N}\sum_{i=1}^N T_r(\calE_i,u)$
and we use this value as the estimation of
$\Exp[T_r(\calP,u)]$.
Now, we show $\calS$ is indeed an \exprkernel.

Recall that we use $\calE(P)$ to denote the deterministic $\epsilon$-kernel for any realization $P\sim \calP$.
%Similarly, we can think of $\calE_1,\ldots, \calE_N$ as samples from the random set $\calE(P)$.
We first compare $\calP$ with the random set $\calE(P)$.

\begin{lemma}
\label{lm:exprkernel}
For any $t\geq 0$ and any direction $u\in \calP^{\polar}$, we have that
$$ (1-\e/2)\Exp[T_r(\calP,u)]\leq \Exp_{P\sim \calP}[T(\calE(P),u,r))]\leq \Exp[T_r(\calP,u)].
$$
\end{lemma}

\begin{proof}
%For any realization $P$ of $\calP$, we suppose $P'$ is the deterministic $\epsilon_0$-kernel of $P$.
By Lemma 4.6 in~\cite{agarwal2004approximating}, we have that
$ (1-\e/2)T_r(P,u)\leq T_r(\calE(P),u)\leq T_r(P,u).
$
The lemma follows by combining all realizations.
\end{proof}

Now we show that $\calS$ is an \exprkernel\ of $\calE(P)$. We first prove the following lemma. The proof is almost the same as that of Lemma~\ref{lm:quant1}, and can be found in Appendix~\ref{app:exprkernel}.

\begin{lemma}
\label{lm:expr1}
Let $N=O\bigl(\e_1^{-2}\e_0^{-(d-1)/2}\log (1/\e_0)\bigr)$, where $\e_0=(\e/4(r-1))^r$, $\e_1=\eps\beta^2$. For any $t\geq 0$ and any direction $u\in \calP^{\polar}$, we have that
$$  \Prob_{P\sim \calS}[\max_{v\in P}\innerprod{u}{v}^{1/r}\geq t]\in \Prob_{P\sim \calP}[\max_{v\in \calE(P)}\innerprod{u}{v}^{1/r}\geq t)]\pm \e_1/4, \text{ and }
$$
$$ \Prob_{P\sim \calS}[\min_{v\in P}\innerprod{u}{v}^{1/r}\geq t]\in \Prob_{P\sim \calP}[\min_{v\in \calE(P)}\innerprod{u}{v}^{1/r}\geq t)]\pm \e_1/4. \quad\quad
$$
\end{lemma}

\begin{lemma}
\label{thm:betaexpr}
Let $N=O\bigl(\beta^{-4}\e^{-(rd-r+4)/2}\log (1/\e)\bigr)$ and $\e_0=(\e/4(r-1))^r$.
$\calS$ constructed above is an \exprkernel\ in $\R^d$.
%, where $\e_0=(\e/6(r-1))^r$, $\e_1=\eps\beta^2/2$.
\end{lemma}

\begin{proof}
Fix a direction $u\in \calP^{\polar}$.
Let $A=\max_{v\in \calP}\innerprod{u}{v}^{1/r}$, $B=\min_{v\in \calP}\innerprod{u}{v}^{1/r}$.
We observe that $B \leq \max_{v\in P}\innerprod{u}{v}^{1/r}\leq A$ for any realization $P\sim\calP$.
We also need the following basic fact about the expectation:
For a random variable $X$, if $\Pr[X\geq a]=1$, then
$\Exp[X]=\int_{b}^{\infty} \Pr[X\geq x]\d x + b$ for any $b\leq a$. Thus, we have that
\begin{align*}
\Exp_{P\sim \calP}[\max_{v\in \calE(P)}\innerprod{u}{v}^{1/r}] &
=\int_{B}^{A} \Pr_{P\sim \calP}[\max_{v\in \calE(P)}\innerprod{u}{v}^{1/r} \geq x]\d x +B \\
&\leq \int_{B}^{A} \Pr_{P\sim \calS}[\max_{v\in P}\innerprod{u}{v}^{1/r} \geq x] \d x + B+\epsilon_1 (A-B)/4\\
&= \Exp_{P\sim \calS}[\max_{v\in P}\innerprod{u}{v}^{1/r}]+\epsilon_1 (A-B)/4,
\end{align*}
where the first inequality is due to Lemma~\ref{lm:expr1}. Similarly, we can show the following two inequalities:
$$
\Exp_{P\sim \calS}[\max_{v\in P}\innerprod{u}{v}^{1/r}]\in \Exp_{P\sim \calP}[\max_{v\in \calE(P)}\innerprod{u}{v}^{1/r}]\pm \epsilon_1 (A-B)/4,
$$
$$
\Exp_{P\sim \calS}[\min_{v\in P}\innerprod{u}{v}^{1/r}]\in \Exp_{P\sim \calP}[\min_{v\in \calE(P)}\innerprod{u}{v}^{1/r}]\pm \epsilon_1 (A-B)/4.
$$
Recall that $T_r(P,u)=\max_{v\in P}\innerprod{u}{v}^{1/r}-\min_{v\in P}\innerprod{u}{v}^{1/r}$.
By the linearity of expectation, we conclude that
$$
\Exp[T_r(\calP,u)]\in \Exp_{P\sim \calP}[T_r(\calE(P),u)]\pm \epsilon_1 (A-B)/2.
$$
Combining Lemma~\ref{lm:exprkernel}, we have that $\Exp[T_r(\calS,u)]\in (1\pm \e/2)\Exp[T_r(\calP,u)]\pm \epsilon_1 (A-B)/2$.
By the $\beta$-assumption, we know that $\Exp[T_r(\calP,u)]\geq \beta^2(A-B)$.
Thus, $\epsilon_1 (A-B)/2\leq \frac{\e}{2}\Exp[T_r(\calP,u)]$,
and $\Exp[T_r(\calS,u)]\in (1\pm \e)\Exp[T_r(\calP,u)]$.
\end{proof}

\topic{Running time:}
In each sample, the size of a deterministic $\epsilon_0$-kernel $\calE_i$ is at most $O\bigl(\epsilon_0^{-(d-1)/2}\bigr)$. Note that constructing an $\e_0$-kernel can be solved in linear time.
We take $O\bigl(\e_1^{-2}\e_0^{-(d-1)/2}\log (1/\e_0)\bigr)$ samples in total. So the overall running time is $O\bigl(n\beta^{-4}\e^{-(rd-r+4)/2}\log  (1/\e)+\poly(1/\e)\bigr)
=\widetilde{O}\left(n\e^{-(rd-r+4)/2}\right)$.

Note that each $\e_0$-kernel contains $O(\e^{-r(d-1)/2})$ points.
We take $N=O\bigl(\e_1^{-2}\e_0^{-(d-1)/2}\log (1/\e_0)\bigr)$ independent samples.
So the total size of \exprkernel\ is $O(\beta^{-4}\e^{-(rd-r+2)}\log(1/\e))$.
In summary, we obtain the following theorem.

\begin{reptheorem}
{thm:exprconstruction}
(restated)
%$\calP$ is a set of $n$ uncertain points in $\R^d$ with existential uncertainty under the $\beta$-assumption.
An \exprkernel\ of size $\tO(\e^{-(rd-r+2)})$
can be constructed in $\widetilde{O}\left(n\e^{-(rd-r+4)/2}\right)$ time
in the existential uncertainty model under the $\beta$-assumption.
In particular, the \exprkernel\ consists of $N=\tO(\e^{-(rd-r+4)/2})$ point sets, each occuring with probability $1/N$
and containing $O(\e^{-r(d-1)/2})$ deterministic points.
\end{reptheorem}

\section{Applications}
\label{sec:app}

In this section, we show that our coreset results for the directional width problem readily imply
several coreset results for other stochastic problems, just as in the deterministic setting.
We introduce these stochastic problems and briefly summarize our results below.

\subsection{Approximating the Extent of Uncertain Functions}
We first consider the problem of approximating the extent of a set $\calH$ of uncertain functions.
As before,
we consider both the existential model and the locational model of uncertain functions.
\begin{enumerate}
\item
In the existential model, each uncertain function $h$ is a function in $\mathbb{R}^d$ associated with a existential probability $p_f$,
which indicates the probability that $h$ presents in a random realization.
\item
In the locational model, each uncertain function $h$ is associated with a finite set $\{h_1,h_2,\ldots\}$ of deterministic functions in $\mathbb{R}^d$.
Each $h_i$ is associated with a probability value $p(h_i)$, such that $\sum_i p(h_i)=1$.
In a random realization, $h$ is independently  realized to some $h_i$, with probability $p(h_i)$.
\end{enumerate}
We use $\calH$ to denote the random instance, that is a random set of functions.
We use $h\in \calH$ to denote the event that the deterministic function $h$ is present in the instance.
For each point $x\in \mathbb{R}^d$, we let the random variable
$
\extent_\calH(x) = \max_{h\in \calH} h(x) - \min_{h\in \calH} h(x)
$
be the extent of $\calH$ at point $x$.
Suppose $\calS$ is another set of uncertain functions.
We say $\calS$ is the \expkernel\ for $\calH$ if
$
(1-\epsilon)\extent_\calH(x) \leq
\extent_\calS(x) \leq
\extent_\calH(x)
$ for any $x\in \R^d$.
We say $\calS$ is the \probkernel\ for $\calH$ if
$
\Pr_{S\sim \calS}\Bigl[\extent_S(x)\leq t\Bigr] \in
\Pr_{H\sim \calH}\Bigl[\extent_H(x)\leq (1\pm \epsilon)t\Bigr]\pm \phi.
$
for any $t\geq 0$ and any $x\in \R^d$.

Let us first focus on linear functions in $\mathbb{R}^d$.
Using the {\em duality transformation} that maps linear function
$y=a_1x_1+\ldots+a_dx_d+a_{d+1}$ to the point
$(a_1,\ldots,a_{d+1})\in \mathbb{R}^{d+1}$, we can reduce the extent problem
to the directional width problem in $\mathbb{R}^{d+1}$.
Let $\calH$ be a set of uncertain linear functions (under either existential or locational model) in $\R^d$ for constant $d$.
From Theorem~\ref{thm:constructM}
and Corollary~\ref{cor:existence},
we can construct a set $S$ of $O(n^{2d})$ deterministic linear functions in $\R^d$, such that
$\extent_S(x) =\E[\extent_\calH(x)]$ for any $x\in \R^d$.
Moreover, for any $\epsilon>0$, there exists an \expkernel\ of size $O(\epsilon^{-d/2})$
and an \probkernel\ of size $\tO(\perror^{-2}\eps^{-d})$.
Using the standard linearization technique \cite{agarwal2004approximating}, we can obtain the following generalization for uncertain polynomials.

\begin{theorem}
\label{thm:uncertainpoly}
Let $\calH$ be a family of uncertain polynomials in $\R^d$ (under either existential or locational model) that
admits linearization of dimension $k$.
We can construct a set $M$ of $O(n^{2k})$ deterministic polynomials, such that
$\extent_M(x) =\E[\extent_\calH(x)]$ for any $x\in \R^d$.
Moreover, for any $\epsilon>0$, there exists an \expkernel\ of size $O(\epsilon^{-k/2})$
and an \probkernel\ of size $\min\{\tO(\perror^{-2}\max\{\lambda^2,\lambda^4\}), \tO(\eps^{-k}\perror^{-2})\}$.  Here $\lambda=\sum_{h\in \calH}(-\ln (1-p_h))$.
\end{theorem}

Now, we consider functions of the form $u(x)=p(x)^{1/r}$ where $p(x)$ is a polynomial and $r$ is a positive integer.
We call such a function a {\em fractional polynomial}.
We still use $\calH$ to denote the random set of fractional polynomials.
Let $\calH^{\polar}\subseteq \R^d$ be the set of points such that for any points $x\in \calH^{\polar}$ and
any function $u\in \calH$, we have $u(x)\geq 0$.
For each point $x\in \calH^{\polar}$, we let the random variable
$
\extent_{r,\calH}(x) = \max_{h\in \calH} h(x)^{1/r} - \min_{h\in \calH} h(x)^{1/r}.
$
We say another random set $\calS$ of functions  is the \exprkernel\ for $\calH$ if
$
(1-\epsilon)\extent_{r,\calH}(x) \leq
\extent_{r,\calS}(x) \leq
\extent_{r,\calH}(x)
$
for any $x\in \calH^{\polar}$.
By the duality transformation and Theorem~\ref{thm:exprconstruction}, we can obtain the following result.

\begin{theorem}
\label{thm:fpoweruncertainpoly}
Let $\calH$ be a family of uncertain fractional polynomials in $\R^d$
in the existential uncertainty model under the $\beta$-assumption.
Further assume that each polynomial admits a linearization of dimension $k$.
For any $\epsilon>0$, there exists an \exprkernel\ of size $\tO(\epsilon^{-(rk-r+2)})$.
Furthermore, the \exprkernel\ consists of
$N=O\bigl(\e^{-(rk-r+4)/2}\bigr)$ sets,
each occurring with probability $1/N$ and containing $O\bigl(\e^{-r(k-1)/2}\bigr)$ deterministic
fractional polynomials.
\end{theorem}

\subsection{Stochastic Moving Points}

We can extend our stochastic models to moving points.
In the existential model, each point $v$ is present with probability $p_v$ and
follows a trajectory $v(t)$ in $\R^d$ when present ($v(t)$ is the position of $v$ at time $t$).
In the locational model, each point $v$ is associated with a distribution of  trajectories (the support size is finite)
and the actual trajectory of $v$ is a random sample for the distribution.
Such uncertain trajectory models have been used in several applications in spatial databases \cite{zheng2011probabilistic}.
For ease of exposition, we assume the existential model in the following.
Suppose each trajectory is a polynomial of $t$ with degree at most $r$.
For each point $v$, any direction $u$ and time $t$, define the polynomial
$
f_v(u, t) =\innerprod{v(t)}{u}
$
and let $\calH$ include $f_v$ with probability $p_v$.
For a set $\calP$ of points, the directional width at time $t$ is
$
\extent_\calH (u,t) = \max_{v\in \calP} f_v(u, t) -\min_{v\in \calP} f_v(u, t).
$
Each polynomial $f_v$ admits a linearization of dimension $k=(r+1)d-1$.
Using Theorem~\ref{thm:uncertainpoly},
we can see that there is a set $M$ of $O(n^{2k})$ deterministic moving points, such that
the directional width of $M$ in any direction $u$
is the same as the expected directional width of $\calP$ in direction $u$.
Moreover, for any $\epsilon>0$, there exists an \expkernel\
(which consists of only deterministic moving points) of size $O(\epsilon^{-(k-1)/2})$
and an \probkernel\
(which consists of both deterministic and stochastic moving points) of size $\tO(\eps^{-k}\perror^{-2})$.

\subsection{Shape Fitting Problems}

Theorem~\ref{thm:uncertainpoly} can be also applied to some stochastic variants of certain shape fitting problems.
%First let us consider the minimum enclosing ball problem.
%We are given a set $P$ of points in $\R^d$. Find the center point $c$ such that
%$\max_{v\in P} \| v-c \|^2$ is minimized.
We first consider the following variant of the minimum enclosing ball problem over stochastic points.
We are given a set $\calP$ of stochastic points (under either existential or locational model),
find the center point $c$ such that
$\E[\max_{v\in \calP} \| v-c \|^2]$ is minimized.
It is not hard to see that the problem is equivalent to minimizing the expected area of the
enclosing ball in $\R^2$.
For ease of exposition, we assume the existential model where $v$ is present with probability $p_v$.
For each point $v\in P$, define the polynomial
$
h_v(x) = \|x\|^2 - 2\innerprod{x}{v} + \|v\|^2,
$
which admits a linearization of dimension $d+1$ \cite{agarwal2004approximating}.
Let $\calH$ be the family of uncertain polynomials $\{h_v\}_{v\in \calP}$
($h_v$ exists with probability $p_v$).
We can see that for any $x\in \R^d$,
$\max_{v\in \calP} \|x-v\|^2 = \max_{h_v\in \calH} h_v(x)$.
Using Theorem~\ref{thm:uncertainpoly},
\footnote{
We can see from the proof that all results that hold for width/extent also hold for support function/maximum.
} we can see
that there is a set $M$ of $O(n^{2d+2})$ deterministic polynomials such that
$\max_{h\in M} h(x) = \E[\max_{v\in \calP} \|x-v\|^2]$ for any $x\in \R^d$
and
a set $S$ of $O(\epsilon^{-(d+1)/2})$ deterministic polynomials such that
$(1-\epsilon) \E[\max_{v\in \calP} \|x-v\|^2]\leq \max_{h\in S} h(x) \leq \E[\max_{v\in \calP} \|x-v\|^2]$ for any $x\in \R^d$.
We can store the set $S$ instead of the original point set in order to answer the following queries:
given a point $v$, return the expected length of the furthest point from $v$.
The problem of finding the optimal center $c$ can be also carried out over $S$,
which can be done in $O(\epsilon^{-O(d^2)})$ time:
We can decompose the arrangement of $n$ semialgebraic surfaces in $\R^d$
into $O(n^{O(d+k)})$ cells of constant description complexity, where $k$ is the linearization dimension
(see e.g., \cite{agarwal2000arrangements}). By enumerating all those cells in the arrangement of $S$,
we know which polynomials lie in the upper envelopes, and we can compute the minimum value in
each such cell in constant time when $d$ is constant.

The above argument can also be applied to the following variant of the spherical shell
%and cylindrical shell problems
for stochastic points.
We are given a set $\calP$ of stochastic points (under either existential or locational model).
%In the spherical shell problem,
Our objective is to find the center point $c$ such that
$\E[\obj(c)]=\E[\max_{v\in P} \| v-c \|^2-\min_{v\in P} \| v-c \|^2]$ is minimized.
The problem is equivalent to minimizing the expected area of the
enclosing annulus in $\R^2$.
%In the cylindrical shell problem, our objective is to find a line $\ell$ such that
%$\E[\obj(c)]=\E[\max_{v\in P} \dist(\ell, v)^2-\min_{v\in P} \dist(\ell,v)^2]$ is minimized.
The objective can be represented as a polynomial of linearization dimension $k=d+1$.
%while the second $k=O(d^2)$.
Proceeding as for the enclosing balls, we can show
there is a set $S$ of $O(\epsilon^{-(k-1)/2})$ deterministic polynomials such that
$(1-\epsilon) \E[\obj(c)]\leq \extent_S(x) \leq \E[\obj(c)]$ for any $x\in \R^d$.
We would like to make a few remarks here.
\begin{enumerate}
\item
Let us take the minimum enclosing ball for example.
If we examine the construction of set $S$,
each polynomial $h\in S$ may {\em not} be of the form
$h(x) = \|x\|^2 - 2\innerprod{x}{v} + \|v\|^2$, therefore does not
translate back to a minimum enclosing ball problem over deterministic points.
\item
Another natural objective function for the minimum enclosing ball and
the spherical shell problem
%or the cylindrical shell problem
would be
the expected radius $\E[\max_{v\in P} \dist(v,c) ]$ and
the expected shell width $\E[\max_{v\in P} \dist(v,c)-\min_{v\in P} \dist(v,c) ]$.
However, due to the fractional powers (square roots) in the objectives,
simply using an \expkernel\ does not work.
This is unlike the deterministic setting.
\footnote{
In particular, there is no stochastic analogue of Lemma 4.6 in \cite{agarwal2004approximating}.
}
We leave the problem of finding small coresets for the spherical shell problem as an interesting open problem.
%In a personal communication, Timothy Chan showed that the minimum-width spherical shell problem
%can be solved in polynomial time for constant $d$ through a
%argument involving coresets.
%However, their algorithm does not provide a coreset of any sort.
However, under the $\beta$-assumption,
we can use \exprkernel s to handle such fractional powers, as in the next subsection.
\end{enumerate}

\begin{reptheorem}{thm:squarecenter} (restated)
Suppose $\calP$ is a set of $n$ independent stochastic points in $\R^d$ under either existential or locational
uncertainty model.
There are linear time approximation schemes
for the following problems:
(1)  finding a center point $c$ to minimize $\E[\max_{v\in \calP} \| v-c \|^2]$;
(2) finding a center point $c$ to minimize $\E[\obj(c)]=\E[\max_{v\in P} \| v-c \|^2-\min_{v\in P} \| v-c \|^2]$.
Note that when $d=2$ the above two problems correspond to minimizing the expected areas of the enclosing ball and the enclosing annulus, respectively.
\end{reptheorem}

\subsection{Shape Fitting Problems (Under the $\beta$-assumption)}
\label{subsec:shapefitting}
In this subsection, we consider
several shape fitting problems in the existential model {\em under the $\beta$-assumption}.
We show how to use Theorem~\ref{thm:fpoweruncertainpoly}
to obtain linear time approximation schemes for those problems.

\begin{enumerate}
\item
(Minimum spherical shell)
We first consider the minimum spherical shell problem.
Given a set $\calP$ of stochastic points (under the $\beta$-assumption),
our goal is to find the center point $c$ such that $ \E[\max_{v\in P} \| v-c \|-\min_{v\in P} \| v-c \|]$ is minimized.
For each point $v\in P$, let
$
h_v(x) = \|x\|^2 - 2\innerprod{x}{v} + \|v\|^2,
$
which admits a linearization of dimension $d+1$.
It is not hard to see that $\E[\max_{v\in P} \| v-c \|]=\E[\max_{v\in P} \sqrt{h_v(c)}]$ and $\E[\min_{v\in P} \| v-c \|]=\E[\min_{v\in P} \sqrt{h_v(c)}]$.
Using Theorem~\ref{thm:fpoweruncertainpoly},
we can see that there are
$N=\tO\bigl(\e^{-(d+3)}\bigr)$ sets $S_i$, each containing
$O\bigl(\e^{-(d+1)}\bigr)$ fractional polynomial $\sqrt{h_v}$s such that for all $x\in \R^d$,
\begin{align}
\label{eq:minenclosing}
\frac{1}{N}\sum_{i\in [N]}(\max_{S_i}\sqrt{h_v(x)}-\min_{S_i}\sqrt{h_v(x)})\in (1\pm \e)(\E[\max_{v\in P} \| v-x \|]-\E[\min_{v\in P} \| v-x \|]).
\end{align}
Note that our \exprkernel\ satisfies the subset constraint.
Hence, each function $\sqrt{h_v}$ corresponds to an original point in $\calP$.
So, we can store $N$ point sets $P_i\subseteq \calP$, with $|P_i|=O\bigl(\e^{-d}\bigr)$ as the coreset for the original point set.
By \eqref{eq:minenclosing}, an optimal solution for the coreset is an $(1+\epsilon)$-approximation for the original problem.

Now, we briefly sketch how to compute the optimal solution for the coreset.
Consider all points in $\cup_i P_i$.
Consider the arrangement of $O\bigl(\e^{-O(d)}\bigr)$ hyperplanes, each bisecting a pair of points in $\cup_i P_i$.
For each cell $C$ of the arrangement,
for any point $v\in C$, the ordering of all points in $\cup_i P_i$ is fixed.
We then enumerate all those cells in the arrangement
and try to find the optimal center in each cell.
Fix a cell $C$.
For any point set $P_i$, we know which point is the furthest one and which point is the closest one from points in $C_0$.
Say they are $v_i=\arg\max_{v\in P_i}\|v-x\|$ and $v'_i=\arg\min_{v\in P_i}\|v-x\|$.
Hence, our problem can be formulated as the following optimization problem:
$$
\min_x \frac{1}{N}\sum_i (d_i-d'_i), \ \ \text{s.t.}\quad d_i^2=\|v_i-x\|^2,d'^2_i=\|v'_i-x\|^2,d_i,d'_i\geq 0,\forall i\in[N]; x\in C_0.
$$
The polynomial system has a constant number of variables and constraints, hence can be solved in constant time.
More specifically, we can introduce a new variable $t$ and let $t=\frac{1}{N}\sum_i (d_i-d'_i)$.
All polynomial constraints define a semi-algebraic set.
By using constructive version of Tarski-Seidenberg theorem, we can project out all variables except $t$ and the resulting set is still a semi-algebraic set (which would be a finite collection of points and intervals in $\R^1$) (See e.g.,\cite{basu2011algorithms}).

\item (Minimum enclosing cylinder, Minimum cylindrical shell)
Let $\calP$ be a set of stochastic points in the existential uncertainty model under the $\beta$-assumption.
Let $\dist(\ell,v)$ denote the distance between a point $v\in \R^d$ and a line $\ell\subset \R^d$.
The goal for the minimum enclosing cylinder problem is to find a line $\ell$ such that
$\E[\max_{v\in \calP}\dist(\ell,v)]$ is minimized,
while that for the minimum cylindrical shell problem is to minimize $\E[\max_{v\in \calP}\dist(\ell,v)-\min_{v\in \calP}\dist(\ell,v)]$.
The algorithms for both problems are almost the same
and we only sketch the one for the minimum enclosing cylinder problem.

We follow the approach in \cite{agarwal2004approximating}.
We represent a line $\ell\in \R^d$ by a $(2d-1)$-tuple
$(x_1,\ldots,x_{2d-1})\in \R^{2d-1}$:
$
\ell= \{p+tq\mid t\in \R\},
$
where $p=(x_1,\cdots,x_{d-1},0)$ is the intersection point of $\ell$ with the hyperplane $x_d=0$
and $q=(x_d,\ldots,x_{2d-1}),\|q\|^2=1$ is the orientation of $\ell$.
Then for any point $v\in \R^d$, we have that
$$
\dist(\ell,v)=\|(p-v)-\innerprod{p-v}{q}q\|,
$$
where the polynomial $\dist^2(\ell,v)$ admits a linearization of dimension $O(d^2)$.
Now, proceeding as for the minimum enclosing ball problem
and using Theorem~\ref{thm:fpoweruncertainpoly},
we can obtain a coreset $\calS$ consisting
$N=O\bigl(\e^{-O(d^2)}\bigr)$ deterministic point sets $P_i\subseteq \calP$.

We briefly sketch how to obtain the optimal solution for the coreset.
We can also decompose $\R^{2d-1}$ (a point $x$ in the space with $\|(x_d,\ldots,x_{2d-1})\|=1$ represents a line in $\R^d$)
into $O\bigl(\e^{-O(d^2)}\bigr)$ semi-algebraic cells such that
for each cell, the ordering of the points in $\calS$ (by their distances to a line in the cell) is fixed.
Note that such a cell is a semi-algebraic cell.
For a cell $C$, assume that $v_i=\arg\max_{v\in P_i}\dist(\ell,v_i)$ for all $i\in [N]$,
where $\ell$ is an arbitrary line in $C$.
We can formulate the problem as the following polynomial system:
$$
\min_l \frac{1}{N}\sum_i d_i, \ \ \text{s.t.} \quad  d_i^2=\dist^2(\ell,v_i),d_i\geq 0,\forall i\in[N]; \ell=(p,q)\in C_0, \|q\|^2=1.
$$
Again the polynomial system has a constant number of variables and constraints.
Thus, we can compute the optimum in constant time.
\end{enumerate}
%However, the above remarks does not rule out the possibility that
%there exists a small coreset consisting of only deterministic points (may even be in a high dimensional space).
%We leave the problem as an interesting open problem.

\begin{reptheorem}{thm:shapefittingbeta} (restated)
Suppose $\calP$ is a set of $n$ independent stochastic points in $\R^d$,
each appearing with probability at least $\beta$, for some fixed constant $\beta>0$.
There are linear time approximation schemes for minimizing
the expected radius (or width) for the minimum spherical shell, minimum enclosing cylinder, minimum cylindrical shell problems over $\calP$.
\end{reptheorem}

\section{Concluding Remarks}
We initiate the study of constructing coresets for various stochastic geometric extent problems.
Our work opens up several avenues for further research.
%There are several interesting open questions left.
%One obvious open question is to construct an \probkernel\
%for the locational uncertainty model.
%Our algorithm for constructing an \probkernel\ under the existential model does not seem to extend
%directly to  the locational model and new ideas may be needed.
One obvious further direction is to construct constant sized \exprkernel s efficiently without the $\beta$-assumption.
A promising approach is the importance sampling, based on the sensitivity of functions, developed in \cite{langberg2010}.
Note that such a construction would lead to efficient approximation schemes for several shape fitting problems
without the $\beta$-assumption.
While the size bounds we obtained for \expkernel s are tight (match the deterministic setting),
the bounds for \probkernel s may not.
Thus obtaining better bounds for \probkernel\
is an interesting open problem.
Another interesting and important further direction is to extend the concept of coresets
to other problems (e.g., clustering) over stochastic datasets.

\vspace{0.2cm}
\noindent
{\bf Acknowledgement:}
Jian Li would like to thank
the Simons Institute for the Theory of Computing for providing a wonderful
research enviorment,
where part of this research was carried out.
Jian Li would also like to thank Leonard Schulman and Timothy Chan for stimulating discussions.

\bibliographystyle{plain}
\bibliography{core}

\begin{thebibliography}{10}

\bibitem{ADP13}
A.~Abdullah, S.~Daruki, and J.M. Phillips.
\newblock Range counting coresets for uncertain data.
\newblock In {\em Proceedings 29th ACM Syposium on Computational Geometry},
  pages 223--232, 2013.

\bibitem{ackermann2010clustering}
Marcel~R Ackermann, Johannes Bl{\"o}mer, and Christian Sohler.
\newblock Clustering for metric and nonmetric distance measures.
\newblock {\em ACM Transactions on Algorithms (TALG)}, 6(4):59, 2010.

\bibitem{afshani2011approximate}
P.~Afshani, P.K. Agarwal, L.~Arge, K.G. Larsen, and J.M. Phillips.
\newblock {(Approximate)} uncertain skylines.
\newblock In {\em Proceedings of the 14th International Conference on Database
  Theory}, pages 186--196, 2011.

\bibitem{agarwal2012range}
P.K. Agarwal, S.-W. Cheng, and K.~Yi.
\newblock Range searching on uncertain data.
\newblock {\em ACM Transactions on Algorithms (TALG)}, 8(4):43, 2012.

\bibitem{agarwal2012nearest}
P.K. Agarwal, A.~Efrat, S.~Sankararaman, and W.~Zhang.
\newblock Nearest-neighbor searching under uncertainty.
\newblock In {\em Proceedings of the 31st Symposium on Principles of Database
  Systems}, pages 225--236, 2012.

\bibitem{AHPSYZ14}
P.K. Agarwal, S.~Har-Peled, S.~Suri, H.~Y{\i}ld{\i}z, and W.~Zhang.
\newblock Convex hulls under uncertainty.
\newblock In {\em Proceedings of the 22nd Annual European Symposium on
  Algorithms}, pages 37--48, 2014.

\bibitem{agarwal2004approximating}
P.K. Agarwal, S.~Har-Peled, and K.R. Varadarajan.
\newblock Approximating extent measures of points.
\newblock {\em Journal of the ACM}, 51(4):606--635, 2004.

\bibitem{agarwal2005geometric}
P.K. Agarwal, S.~Har-Peled, and K.R. Varadarajan.
\newblock Geometric approximation via coresets.
\newblock {\em Combinatorial and Computational Geometry}, 52:1--30, 2005.

\bibitem{agarwal2008robust}
P.K. Agarwal, S.~Har-Peled, and H.~Yu.
\newblock Robust shape fitting via peeling and grating coresets.
\newblock {\em Discrete \& Computational Geometry}, 39(1-3):38--58, 2008.

\bibitem{agarwal1992farthest}
P.K. Agarwal, J.~Matou{\v{s}}ek, and S.~Suri.
\newblock Farthest neighbors, maximum spanning trees and related problems in
  higher dimensions.
\newblock {\em Computational Geometry - Theory and Applications},
  1(4):189--201, 1992.

\bibitem{agarwal2000arrangements}
P.K. Agarwal and M.~Sharir.
\newblock {\em {\em Arrangements and their applications.} Handbook of
  Computational Geometry, {\em J. Sack and J. Urrutia (eds.)}}, pages 49--119.
\newblock Elsevier, Amsterdam, The Netherlands, 2000.

\bibitem{anthony2009neural}
Martin Anthony and Peter~L Bartlett.
\newblock {\em Neural network learning: Theoretical foundations}.
\newblock cambridge university press, 2009.

\bibitem{bs-ads-04}
D.~Bandyopadhyay and J.~Snoeyink.
\newblock Almost-{D}elaunay simplices: Nearest neighbor relations for imprecise
  points.
\newblock In {\em Proceedings of the 15th ACM-SIAM Symposium on Discrete
  Algorithms}, pages 410--419, 2004.

\bibitem{barequet2001efficiently}
G.~Barequet and S.~Har-Peled.
\newblock Efficiently approximating the minimum-volume bounding box of a point
  set in three dimensions.
\newblock {\em Journal of Algorithms}, 38(1):91--109, 2001.

\bibitem{basu2011algorithms}
Saugata Basu, Richard Pollack, and M~Roy.
\newblock Algorithms in real algebraic geometry.
\newblock {\em AMC}, 10:12, 2011.

\bibitem{chan2000approximating}
T.M. Chan.
\newblock Approximating the diameter, width, smallest enclosing cylinder, and
  minimum-width annulus.
\newblock In {\em Proceedings of the 16th Annual Son Computational Geometry},
  pages 300--309, 2000.

\bibitem{Cha06}
T.M. Chan.
\newblock Faster core-set constructions and data-stream algorithms in fixed
  dimensions.
\newblock {\em Computational Geometry: Theory and Applications}, 35:20--35,
  2006.

\bibitem{chen2009coresets}
K.~Chen.
\newblock On coresets for k-median and k-means clustering in metric and
  euclidean spaces and their applications.
\newblock {\em SIAM Journal on Computing}, 39(3):923--947, 2009.

\bibitem{cheng2008cleaning}
R.~Cheng, J.~Chen, and X.~Xie.
\newblock Cleaning uncertain data with quality guarantees.
\newblock {\em Proceedings of the VLDB Endowment}, 1(1):722--735, 2008.

\bibitem{cormode2008approximation}
G.~Cormode and A.~McGregor.
\newblock Approximation algorithms for clustering uncertain data.
\newblock In {\em Proceedings of the 27th Symposium on Principles of Database
  Systems}, pages 191--200, 2008.

\bibitem{CM03}
G.~Cormode and S.~Muthukrishnan.
\newblock Radial histograms for spatial streams.
\newblock Technical Report 2003-11, Center for Discrete Mathematics and
  Computer Science (DIMACS), 2003.

\bibitem{conf/cidr/DeshpandeGM05}
A.~Deshpande, C.~Guestrin, and S.~Madden.
\newblock Using probabilistic models for data management in acquisitional
  environments.
\newblock In {\em CIDR}, pages 317--328, 2005.

\bibitem{DBLP:conf/vldb/DeshpandeGMHH04}
A.~Deshpande, C.~Guestrin, S.~Madden, J.M. Hellerstein, and W.~Hong.
\newblock Model-driven data acquisition in sensor networks.
\newblock In {\em VLDB}, pages 588--599, 2004.

\bibitem{deshpande2006matrix}
A.~Deshpande, L.~Rademacher, S.~Vempala, and G.~Wang.
\newblock Matrix approximation and projective clustering via volume sampling.
\newblock In {\em Proceedings of the 17th ACM-SIAM symposium on Discrete
  algorithm}, pages 1117--1126, 2006.

\bibitem{dong2007data}
X.~Dong, A.Y. Halevy, and C.~Yu.
\newblock Data integration with uncertainty.
\newblock In {\em Proceedings of the 33rd International Conference on Very
  Large Data Bases}, pages 687--698, 2007.

\bibitem{DHLS13}
A.~Driemel, H.~HAverkort, M.~L\"offler, and R.I. Silveira.
\newblock Flow computations on imprecise terrains.
\newblock {\em Journal of Computational Geometry}, 4:38--78, 2013.

\bibitem{dudley1974metric}
R.M. Dudley.
\newblock Metric entropy of some classes of sets with differentiable
  boundaries.
\newblock {\em Journal of Approximation Theory}, 10(3):227--236, 1974.

\bibitem{edelsbrunner1986constructing}
H.~Edelsbrunner, J.~O'Rourke, and R.~Seidel.
\newblock Constructing arrangements of lines and hyperplanes with applications.
\newblock {\em SIAM Journal on Computing}, 15(2):341--363, 1986.

\bibitem{ES11}
W.~Evans and J.~Sember.
\newblock The possible hull of imprecise points.
\newblock In {\em Proceedings of the 23rd Canadian Conference on Computational
  Geometry}, 2011.

\bibitem{feldman2009private}
D.~Feldman, A.~Fiat, H.~Kaplan, and K.~Nissim.
\newblock Private coresets.
\newblock In {\em Proceedings of the 41st Annual ACM Symposium on Theory of
  Computing}, pages 361--370, 2009.

\bibitem{FL11}
D.~Feldman and M.~Langberg.
\newblock A unified framework for approximating and clustering data.
\newblock In {\em Proceedings of the 43rd ACM Symposium on Theory of
  Computing}, pages 569--578, 2011.

\bibitem{feldman2012data}
Dan Feldman and Leonard~J Schulman.
\newblock Data reduction for weighted and outlier-resistant clustering.
\newblock In {\em Proceedings of the twenty-third annual ACM-SIAM symposium on
  Discrete Algorithms}, pages 1343--1354. SIAM, 2012.

\bibitem{FHKS16}
Martin Fink, John Hershberger, Nirman Kumar, and Subhash Suri.
\newblock Hyperplane seperability and convexity of probabilistic point sets.
\newblock In {\em Proceedings Symposium on Computational Geometry}, 2016.

\bibitem{ghosh1998support}
P.K. Ghosh and K.V. Kumar.
\newblock Support function representation of convex bodies, its application in
  geometric computing, and some related representations.
\newblock {\em Computer Vision and Image Understanding}, 72(3):379--403, 1998.

\bibitem{guha2009exceeding}
S.~Guha and K.~Munagala.
\newblock Exceeding expectations and clustering uncertain data.
\newblock In {\em Proceedings of the 28th Symposium on Principles of Database
  Systems}, pages 269--278, 2009.

\bibitem{gss-cscah-93}
L.J. Guibas, D.~Salesin, and J.~Stolfi.
\newblock Constructing strongly convex approximate hulls with inaccurate
  primitives.
\newblock {\em Algorithmica}, 9:534--560, 1993.

\bibitem{HP11}
S.~Har-Peled.
\newblock On the expected complexity of random convex hulls.
\newblock {\em arXiv:1111.5340}, 2011.

\bibitem{har2004coresets}
S.~Har-Peled and S.~Mazumdar.
\newblock On coresets for k-means and k-median clustering.
\newblock In {\em Proceedings of the 36th Annual ACM Symposium on Theory of
  Computing}, pages 291--300, 2004.

\bibitem{har2011geometric}
Sariel Har-Peled.
\newblock {\em Geometric approximation algorithms}, volume 173.
\newblock American mathematical society Providence, 2011.

\bibitem{har2004shape}
Sariel Har-Peled and Yusu Wang.
\newblock Shape fitting with outliers.
\newblock {\em SIAM Journal on Computing}, 33(2):269--285, 2004.

\bibitem{hm-ticpps-08}
M.~Held and J.S.B. Mitchell.
\newblock Triangulating input-constrained planar point sets.
\newblock {\em Information Processing Letters}, 109(1):54--56, 2008.

\bibitem{huang2015approximating}
Lingxiao Huang and Jian Li.
\newblock Approximating the expected values for combinatorial optimization
  problems over stochastic points.
\newblock In {\em The 42nd International Colloquium on Automata, Languages, and
  Programming}, pages 910--921. Springer, 2015.

\bibitem{jeffery2006declarative}
S.R. Jeffery, G.~Alonso, M.J. Franklin, W.~Hong, and J.~Widom.
\newblock Declarative support for sensor data cleaning.
\newblock In {\em Pervasive Computing}, pages 83--100. 2006.

\bibitem{JLP11}
A.G. J{\o}rgensen, M.~L\"offler, and J.M. Phillips.
\newblock Geometric computation on indecisive points.
\newblock In {\em Proceedings of the 12th Algorithms and Data Structure
  Symposium}, pages 536--547, 2011.

\bibitem{KCS11b}
P.~Kamousi, T.M. Chan, and S.~Suri.
\newblock The stochastic closest pair problem and nearest neighbor search.
\newblock In {\em Proceedings of the 12th Algorithms and Data Structure
  Symposium}, pages 548--559, 2011.

\bibitem{KCS11a}
P.~Kamousi, T.M. Chan, and S.~Suri.
\newblock Stochastic minimum spanning trees in euclidean spaces.
\newblock In {\em Proceedings of the 27th Symposium on Computational Geometry},
  pages 65--74, 2011.

\bibitem{k-bmips-08}
H.~Kruger.
\newblock Basic measures for imprecise point sets in {$\mathbb{R}^d$}.
\newblock Master's thesis, Utrecht University, 2008.

\bibitem{langberg2010}
M.~Langberg and L.J. Schulman.
\newblock Universal $\epsilon$-approximators for integrals.
\newblock In {\em Proceedings of the 21st Annual ACM-SIAM Symposium on Discrete
  Algorithms}, 2010.

\bibitem{li2016range}
J.~Li and H.~Wang.
\newblock Range queries on uncertain data.
\newblock {\em Theoretical Computer Science}, 609(1):32--48, 2016.

\bibitem{LLS01}
Y.~Li, P.M. Long, and A.~Srinivasan.
\newblock Improved bounds on the samples complexity of learning.
\newblock {\em Journal of Computer and System Sciences}, 62:516--527, 2001.

\bibitem{loffler2009shape}
M.~L{\"o}ffler and J.~Phillips.
\newblock Shape fitting on point sets with probability distributions.
\newblock In {\em Proceedings of the 17th European Symposium on Algorithms},
  pages 313--324, 2009.

\bibitem{ls-dtip-08}
M.~L{\"o}ffler and J.~Snoeyink.
\newblock {Delaunay} triangulations of imprecise points in linear time after
  preprocessing.
\newblock In {\em Proceedings of the 24th Sympoium on Computational Geometry},
  pages 298--304, 2008.

\bibitem{LvK08}
M.~L\"offler and M.~van Kreveld.
\newblock Approximating largest convex hulls for imprecise points.
\newblock {\em Journal of Discrete Algorithms}, 6:583--594, 2008.

\bibitem{matousek1991computing}
J.~Matou{\v{s}}ek.
\newblock Computing the center of planar point sets.
\newblock {\em Discrete and Computational Geometry}, 6:221, 1991.

\bibitem{enclosingball14}
A.~Munteanu, C.~Sohler, and D.~Feldman.
\newblock Smallest enclosing ball for probabilistic data.
\newblock In {\em Proceedings of the 30th Annual Symposium on Computational
  Geometry}, 2014.

\bibitem{nt-teb-00}
T.~Nagai and N.~Tokura.
\newblock Tight error bounds of geometric problems on convex objects with
  imprecise coordinates.
\newblock In {\em Jap.\ Conf.\ on Discrete and Comput.\ Geom.}, LNCS 2098,
  pages 252--263, 2000.

\bibitem{obj-ue-05}
Y.~Ostrovsky-Berman and L.~Joskowicz.
\newblock Uncertainty envelopes.
\newblock In {\em Abstracts of the 21st European Workshop on Comput.\ Geom.},
  pages 175--178, 2005.

\bibitem{Phi16}
Jeff~M. Phillips.
\newblock Coresets and sketches.
\newblock In {\em Handbook of Discrete and Computational Geometry}, number
  Chapter 49. CRC Press, 3rd edition, 2016.

\bibitem{gss-egbra-89}
D.~Salesin, J.~Stolfi, and L.J. Guibas.
\newblock Epsilon geometry: building robust algorithms from imprecise
  computations.
\newblock In {\em Proceedings of the 5th Symposium on Computational Geometry},
  pages 208--217, 1989.

\bibitem{schneider1993convex}
R.~Schneider.
\newblock {\em Convex bodies: the Brunn-Minkowski theory}, volume~44.
\newblock Cambridge University Press, 1993.

\bibitem{SVY13}
S.~Suri, K.~Verbeek, and H.~Y{\i}ld{\i}z.
\newblock On the most likely convex hull of uncertain points.
\newblock In {\em Proceedings of the 21st European Symposium on Algorithms},
  pages 791--802, 2013.

\bibitem{kl-lbbsd-10}
M.~van Kreveld and M.~L{\"o}ffler.
\newblock Largest bounding box, smallest diameter, and related problems on
  imprecise points.
\newblock {\em Computational Geometry: Theory and Applications}, 43:419--433,
  2010.

\bibitem{vapnik1971uniform}
V.N. Vapnik and A.Y. Chervonenkis.
\newblock On the uniform convergence of relative frequencies of events to their
  probabilities.
\newblock {\em Theory of Probability \& Its Applications}, 16(2):264--280,
  1971.

\bibitem{XLJ16}
Jie Xue, Yuan Li, and Ravi Janardan.
\newblock On the separability of stochasitic geometric objects, with
  applications.
\newblock In {\em Proceedings Symposium on Computational Geometry}, 2016.

\bibitem{YAPV04}
H.~Yu, P.K. Agarwal, R.~Poreddy, and K.~Varadarajan.
\newblock Practical methods for shape fitting and kinetic data structures using
  coresets.
\newblock {\em Algorithmica}, 52(378-402), 2008.

\bibitem{zheng2011probabilistic}
K.~Zheng, G.~Trajcevski, X.~Zhou, and P.~Scheuermann.
\newblock Probabilistic range queries for uncertain trajectories on road
  networks.
\newblock In {\em Proceedings of the 14th International Conference on Extending
  Database Technology}, pages 283--294, 2011.

\end{thebibliography}

\appendix
\section{Missing Details in Section~\ref{sec:expkernel}}
\label{sec:appendix}

\subsection{Details for Section~\ref{sec:linearalgo}}

\vspace{0.2cm}
\noindent
\textbf{Lemma~\ref{lm:transform}.}
We find an affine transform $T$ in $O(2^{O(d)} n\log n)$ time,
such that the convex polytope $\M'=T(\M)$ is $\alpha$-fat for some constant $\alpha$.
\begin{proof}
%Let $\theta_1,\ldots, \theta_k$ and $v_1,\ldots, v_k$
%be defined as in Lemma~\ref{lm:extreme}.
By the results in \cite{barequet2001efficiently},
we only need to construct an approximate bounding box,
which can be done as follows:
We first identify two points $y_1$ and $y_2$ in $\M$ such that their distance is a constant approximation of
the diameter of $\M$. Then we project the points in $\M$ to a hyperplane $H\in \R^{d-1}$ perpendicular to the line through $y_1$ and $y_2$, and recursively identify two points among the projected points as the approximate diameter. 
Hence, it suffices to show how to identify such two points $y_1$ and $y_2$.
Let $\delta=\arccos(1/2)$.
Suppose we are working on $\R^d$.
We compute a set $\calI$ of $O(\delta^{-(d-1)})$ points on the unit sphere $\mathbb{S}^{d-1}$
such that for any point $v\in \mathbb{S}^{d-1}$, there is a point $u\in \calI$ such that $\angle(u,v)\leq \delta$ (see e.g., \cite{agarwal1992farthest, chan2000approximating}).
From Lemma~\ref{lm:extreme},
we know that we can compute for each direction $u\in \mathbb{S}^{d-1}$, the point $x(u)\in \M$ that maximizes $\innerprod{u}{x(u)}$
in $O(n\log n)$ time.
For each $u\in \calI$, compute both $x(u)$ and $x(-u)$, and pick the pair that maximizes
$\|x(u)-x(-u)\|$. Now, we argue this is a constant approximation of the diameter.
Suppose the diameter of $\M$ is $(y_1, y_2)$ where $y_1,y_2\in \M$.
Consider the direction $v=(y_1-y_2)/\|y_1-y_2\|$.
Without loss of generality, assume $y_1=\arg\max_{y} \innerprod{y}{v}$
and $y_2=\arg\max_{y} \innerprod{y}{-v}$.
Moreover, there is a direction $u\in \calI$ such that $\angle(u,v)\leq \delta$.
Therefore, we can get  that
\begin{align*}
\dw(\M, u) & = f(\M, u)+f(\M,-u) \geq  \innerprod{y_1}{u}+ \innerprod{y_2}{-u} \\
&=\innerprod{u}{y_1-y_2} =\|y_1-y_2\| \cos \angle(u, v)
\geq \|y_1-y_2\| /2.
\end{align*}
In the third equation, we use the simple fact that $\cos \angle(u, v)=\innerprod{u}{v}/\|u\| \|v\|$.
\end{proof}

\vspace{0.2cm}
\noindent
\textbf{Lemma~\ref{lm:directionkernel}.}
$S=\{x(u)\}_{u\in \calI}$ is an $\eps$-kernel for $M'$.
\begin{proof}
Consider an arbitrary direction $v\in \mathbb{S}^{d-1}$ with $\|v\|=1$.
Suppose the point $a\in \M'$ maximizes $\innerprod{v}{a}$
$b\in \M'$ maximizes $\innerprod{-v}{b}$.
Hence, $\dw(\M',v)=\innerprod{v}{a}-\innerprod{v}{b}=\innerprod{v}{a-b}$.
By the construction of $\calI$, there is a direction $u\in \calI$ (with  $\|u\|=1$)
such that $\|u-v\|\leq \delta$. Then, we can see that
\begin{align*}
\dw(S,v) & \geq \innerprod{v}{x(u)}-\innerprod{v}{x(-u)}=\innerprod{v}{x(u)-x(-u)} \\
& = \innerprod{u}{x(u)-x(-u)} + \innerprod{v-u}{x(u)-x(-u)} \\
& \geq \innerprod{u}{a-b}- \|v-u\| \|x(u)-x(-u)\|  \\
& = \innerprod{v}{a-b} + \innerprod{u-v}{a-b}- \|v-u\| \|x(u)-x(-u)\|   \\
& \geq  \innerprod{v}{a-b} - \|v-u\| \|x(u)-x(-u)\| - \|u-v\| \|a-b\| \\
& \geq  \dw(\M',v) - O(\delta d) \geq (1-\epsilon) \dw(\M',v)
\end{align*}
In the last and 2nd to last inequalities, we use the fact that $\M'$ is $\alpha$-fat
(i.e., $\alpha \hypercube \subset \M'\subset \hypercube$).
\end{proof}

\subsection{Details for Section~\ref{sec:expsubset}}

\vspace{0.2cm}
\noindent
\textbf{Theorem~\ref{thm:beta}.}
Under the $\beta$-assumption, there is an \expkernel\ in $\R^d$ (for $d=O(1)$), which is
of size $O(\beta^{-(d-1)}\epsilon^{-(d-1)/2} \log (1/\epsilon) )$ and satisfies the subset constraint,
in the existential uncertainty model.

\begin{proof}
Our algorithm is inspired by the peeling idea in \cite{agarwal2008robust}. %\cite{SH03}
Let $\epsilon_1= \epsilon \alpha\beta^2/4\sqrt{d}$, where $\alpha$ is a constant defined later.
%One can construct a set $\calI$ of $K=O(1/\epsilon_1^{d-1})$ points on $\S^{d-1}$ so that for any point $x\in \S^{d-1}$, there exists a point $y\in \calI$
%such that $||x-y||\leq \epsilon_1$ (see e.g., \cite{agarwal2004approximating}).
%For each $u_i\in \calI$, let $b(u_i)$ be the set of $L$ vertices
%that maximizes $\innerprod{u_i}{v}$ over all $v\in \calP$
%(i.e., the first $L$ vertices in the canonical order w.r.t. $u_i$)
%where $L = \log_{1-\beta} \epsilon \alpha \beta^2 =O(\log(1/\epsilon))$.
%Let $\calS=\cup_{0\leq i<K} b(u_i)$.
We repeat the following for $L=O(\log_{1-\beta} \epsilon_1) = O(\log (1/\epsilon))$ rounds:
In round $i$, we first compute an $(\epsilon_1/\sqrt{d})$-kernel $\calS_i$
(of size $O((\sqrt{d}/\epsilon_1)^{(d-1)/2})=O(\beta^{-(d-1)}\epsilon^{-(d-1)/2})$)
for the remaining points (in the deterministic sense)
and then delete all points of $\calS_i$.
Let $\calS=\cup_i \calS_i$.
Now, we show that $\calS$ is an \expkernel\ for $\calP$.

We first establish a lower bound of $\dw(\calP,u)$ for any unit vector $u\in\mathbb{S}^{d-1}$.
Assume without loss of generality that $\alpha \hypercube \subset \CH(\calP)\subset \hypercube$
where $\hypercube=[-1,1]^d$ and $\alpha$ is a constant only depending on $d$.
Since $\alpha\hypercube\subset \CH(P)$,
we know there is a point $v\in \CH(P)$
such that $\innerprod{u}{v}\geq \alpha$
and a different point
$w\in \CH(P)$
such that $\innerprod{u}{w}\leq -\alpha$.
Hence, we have that
$$\dw(\calP,u)\geq \beta^2(\innerprod{u}{v}-\innerprod{u}{w})\geq 2\alpha\beta^2.$$

Fix an arbitrary direction $u\in \S^{d-1}$.
%Suppose $\theta_i\leq \theta\leq \theta_{i+1} $ for some $i$.
Now, we bound the difference between $f(\calP,u)$ and $f(\calS,u)$.
We show that for any real value $x\in [-\sqrt{d},\sqrt{d}]$,
\begin{align}
\label{eq:probbound}
\Pr_{P\sim \calP}[f(P, u) \geq x] \leq  \Pr_{S\sim \calS}[f(S, u) \geq x-\epsilon_1]+\epsilon_1.
\end{align}
In fact, a proof of the above statement provides a proof for Theorem~\ref{thm:betaprob}
(i.e., $\calS$ is an \probkernel\ as well).

Let $\calL_\calP=\{v_1, v_2,\ldots, v_L\}$ be the set of $L$ points $v\in \calP$ that maximize $\innerprod{v}{u}$
(i.e., the first $L$ vertices in the canonical order w.r.t. $u$).
Similarly, let $\calL_\calS=\{w_1, w_2,\ldots, w_L\}$ be the set of $L$ points $v\in \calS$ that maximize $\innerprod{w}{u}$.
We distinguish two cases:
\begin{enumerate}
\item
$\calL_\calP=\calL_\calS$: If $x\geq \innerprod{u}{v_L}$,
we can see that $\Pr_{P\sim \calP}[f(P, u) \geq x] =  \Pr_{S\sim \calS}[f(S, u) \geq x]$.
If $x< \innerprod{u}{v_L}$,
both $\Pr_{P\sim \calP}[f(P, u) \geq x]$ and $\Pr_{S\sim \calS}[f(S, u) \geq x]$ are
at least $1-\prod_{v\in \calL_\calP}(1-p_v)\geq 1-(1-\beta)^L\geq 1-\epsilon_1$.
\item
Suppose $j$ is the smallest index such that $v_j\ne w_j$.
For $x> \innerprod{u}{v_j}$,
we can see that
$\Pr_{P\sim \calP}[f(P, u) \geq x] =  \Pr_{S\sim \calS}[f(S, u) \geq x]$.
Now, we focus on the case where $x\leq \innerprod{u}{v_j} $.
From the construction of $\calS$, we can see that
$\innerprod{w_{j'}}{u} \geq \innerprod{v_j}{u} - \epsilon_1$ for all $j'\geq j$.\footnote{
To see this, consider the round in which $w_{j'}$ is chosen.
Let $t$ be the vertex minimizing $\innerprod{t}{u}$.
As $v_j$ is not chosen, we must have
$\innerprod{w_{j'}}{u}-\innerprod{t}{u} \geq (1-\epsilon_1/\sqrt{d}) (\innerprod{v_{j}}{u}-\innerprod{t}{u})$.
}
%Let $I_\calP(x)=\{v\in \calP \mid  \innerprod{u}{v}\geq x\}$
%and
%$I_\calS(x)=\{v\in \calS \mid  \innerprod{u}{v}\geq x\}$.
Hence, for $x\leq \innerprod{u}{v_j} $, we can see that
\begin{align*}
\Pr_{S\sim \calS}[f(S, u) \geq x-\epsilon_1]
\geq  1-\prod_{v\in \calL_\calS}(1-p_v)
\geq 1-\epsilon_1.
\end{align*}

\end{enumerate}
So, in either case, \eqref{eq:probbound} is satisfied.
%We have that for each $v'_i\in S'$,
%$$
%\prod_{j: \innerprod{v_j}{u} \geq x}(1-p_{v'_j}) < \prod_{j: \innerprod{v_j}{u} \geq x-\epsilon_1 } (1-p_{v_j})
%$$
%(since the set of vertices satisfy the subscript condition in the LHS is a superset of that of the RHS).
We also need the following basic fact about the expectation:
For a random variable $X$, if $\Pr[X\geq a]=1$, then
$\Exp[X]=\int_{b}^{\infty} \Pr[X\geq x]\d x + b$ for any $b\leq a$.
Since $-\sqrt{d}\leq f(P, u)\leq \sqrt{d}$ for any realization $P$, we have
\begin{align*}
f(\calP,u) & =\int_{-\sqrt{d}}^{\infty} \Pr_{P\sim \calP}[f(P, u) \geq x] \d x -\sqrt{d} \\
&\leq \int_{-\sqrt{d}}^{\infty} \Pr_{S\sim \calS}[f(S, u) \geq x-\epsilon_1] \d x + 2\sqrt{d}\epsilon_1 -\sqrt{d}\\
&\leq \int_{-\sqrt{d}-\epsilon_1}^{\infty} \Pr_{S\sim \calS}[f(S, u) \geq x] \d x -\sqrt{d}-\epsilon_1 + 3\sqrt{d}\epsilon_1\\
& = f(\calS,u)+3\sqrt{d}\epsilon_1,
\end{align*}
where the only inequality is due to \eqref{eq:probbound} and the fact that
$\Pr_{P\sim \calP}[f(P, u) \geq x]=\Pr_{S\sim \calS}[f(S, u) \geq x]=0$ for $x>1$.
Similarly, we can get that $f(\calS, -u)\geq
f(\calP,-u) -3\epsilon_1\sqrt{d}$.
By the choice of $\epsilon_1$,
we have that $6\sqrt{d}\epsilon_1\leq \epsilon\cdot 2\alpha\beta^2\leq \epsilon\dw(\calP,u)$. Hence,
$\dw(\calS,u) \geq \dw(\calP,u) -6\sqrt{d}\epsilon_1 \geq (1-\epsilon)\dw(\calP,u)$.
\end{proof}

\subsection{Locational uncertainty}
Similar results are possible for uncertain points with locational uncertainty.
Now there are $m$ possible locations, and thus ${m \choose 2}$ hyperplanes $\Gamma$ that partition $\R^d$.
 We can replicate all bounds in this setting, except that $m$ replaces $n$ in each bound.
The main difficulty is in replicating Lemma \ref{lem:10} that given a direction $u$ calculates the vertex of $M$; for locational uncertain points this is described in Lemma \ref{lem:10-loc}.  Moreover, the $O(m^2 \log m)$ bound for $\R^2$ is also carefully described in Lemma \ref{lem:20loc}.

In the locational uncertainty model, Lemma~\ref{lm:complexity} also holds
with a stronger general position assumption.
With the new general position assumption, it is straightforward to show
that the gradient vector is different for two adjacent cones in $\conedecomp(\Gamma)$.
Other parts of the proof is essentially the same as Lemma~\ref{lm:complexity}.
The details can be found below.
Theorem~\ref{thm:expconstruction} also holds for the locational model without any change in the proof (the running time becomes $O(m\log^2 m)$).

Now, we prove that Lemma~\ref{lm:complexity} also holds for the locational model.
For this purpose, we need a stronger general position assumption:
(1) For any $v\in \calP$, $\sum_{s\in S}p_{vs}\in (0,1)$. This suggests that we need to consider
the model with both existential and locational uncertainty.
We can make this assumption hold by subtracting an infinitesimal value from each probability value
without affecting the directional width in any essential way.
(2) For any two points $v_1,v_2\in \calP$, two positions $s_1,s_2$ and
two subsets of positions $S_1, S_2$, $p_{v_1,s_1} (\sum_{s\in S_2} p_{v_2,s}) s_1
\ne p_{v_2,s_2} (\sum_{s\in S_1} p_{v_1,s}) s_2$ (this is indeed a general position assumption since we only have a finite number of equations
to exclude but uncountable number of choices of the positions).

\vspace{0.2cm}
\noindent
\textbf{Lemma~\ref{lm:complexity}.} (for the locational model)\textbf{.}
Assuming the locational model and the above general position assumption,
the complexity of $\M$ is the same as the cardinality of $\conedecomp(\Gamma)$, i.e.,
$|M|=|\conedecomp(\calP)|.$
Moreover, each cone $C\in \conedecomp(\calP)$ corresponds to exactly one
vertex $v$ of $\CH(M)$ in the following sense:
$\grad f(\M,u) = v$ for all $u\in \interior C$.

\vspace{-0.2cm}
\begin{proof}
The proof is almost the same as that for Lemma~\ref{lm:complexity} except
that we need to show $f(\M,u)$ is different for two adjacent cones in $\conedecomp(\Gamma)$.
Again, let $C_1,C_2$ be two adjacent cones separated by some hyperplane in $\Gamma$.
Suppose $u_1\in \interior C_1$ and $u_2\in \interior C_2$.
Consider the canonical orders $O_1$ and $O_2$ of $\calP$ with respect to $u_1$ and $u_2$ respectively.
W.l.o.g., assume that $O_1=\{s_1,\ldots, s_i, s_{i+1},\ldots,s_n\}$
and $O_1=\{s_1,\ldots, s_{i+1}, s_{i},\ldots,s_n\}$.

Using the notations from Lemma~\ref{lem:10-loc},
$f(\calP,u)$ can computed by
$\sum_{s \in S,v \in \P} \Pr^R(v,s,u) \innerprod{s}{u}$.
Hence,
$\grad f(\calP,u) = \sum_{s \in S,v \in \P} \Pr^R(v,s,u)  s$.

Suppose $s_i$ is a possible location for $v_1$ and
$s_{i+1}$ is a possible location for $v_2$.
If $v_1\ne v_2$, we have
$
\grad f(\calP,u_1)-\grad f(\calP,u_2) =
\Pr^R_\emptyset(s_i,u_1) \cdot \Bigl(s_i p_{v_1s_i}  \sum_{s' \in S^R(s_i,u_1)} p_{v_{2}s'} -
 s_{i+1} p_{v_2s_{i+1}}  \sum_{s' \in S^R(s_i,u_2)} p_{v_{1}s'} \Bigr)\ne 0.
$
If $v_1= v_2=v$, we have
$
\grad f(\calP,u_1)-\grad f(\calP,u_2) =
\Pr^R_\emptyset(s_i,u_1) \cdot \Bigl(s_i p_{v s_i}  -
 s_{i+1} p_{v s_{i+1}}  \Bigr)\ne 0.
$
\end{proof}

\section{Missing Details in Section~\ref{sec:probkernel}}
\label{app:probkernel}

\textbf{Theorem~\ref{lm:boundprob2}.}
Let $\perror_1=O(\frac{\perror}{\max\{\lambda,\lambda^2\}})$
  and
$N=O(\frac{1}{\perror_1^2}\log \frac{1}{\delta})
=O(\frac{\max\{\lambda^2,\lambda^4\}}{\perror^2}\log \frac{1}{\delta}).$
With probability at least $1-\delta$, for any $t\geq 0$ and any direction $u$, we have that
%\begin{align*}
$
\Pr\Bigl[\dw(\calS,u)\leq t\Bigr] \in
\Pr\Bigl[\dw(\calP,u)\leq t\Bigr]\pm \perror.
$
%\end{align*}

\begin{proof}
Fix an arbitrary direction $u$ (w.l.o.g., say it is the x-axis)
and rename all points in $\calP$ as $v_1,v_2,\ldots,v_n$ as before.
Consider the Poissonized instance of $\calP$.
%Imagine each point $v_i$ contains $\lambda_{v_i}$ units of mass.
Let $v'_1,\ldots, v'_N$ be the $N$ points in $\calS$
(also sorted in nondecreasing order of their x-coordinates).
%Imagine each point $v'_i$ contains $\lambda/N$ units of mass.
Now, we create a coupling between all mass in
$\fP$ and that in $\fQ$, as follows.
We process all points in $\fP$ from left to right, starting with $v_1$.
The process has $N$ rounds. In each round, we assign exactly $1/N$ units of mass in $\fP$
to a point in $\fQ$.
In the first round, if $v_1$ contains less than $1/N$ units of mass,
we proceed to $v_2, v_3, \ldots v_i$ until we reach $1/N$ units collectively.
We split the last node $v_i$ into two node $v_{i1}$ and $v_{i2}$
so that the mass contained in $v_1,\ldots, v_{i-1},v_{i1}$ is exactly $1/N$,
and we assign those nodes to $v'_1$. We start the next round with $v_{i2}$.
If $v_1$ contains more than $1/N$ units of mass,
we split $v_1$ into $v_{11}$ ($v_{11}$ contains $1/N$ units) and $v_{12}$
and we start the second round with $v_{12}$. We repeat this process
until all mass in $\fP$ is assigned.

The above coupling can be viewed as a mass transportation from $\fP$ to $\fQ$.
We will need one simple but useful property about this transportation:
for any vertical line $x=t$, at most $\perror_1$ units of mass
are transported across the vertical line
(by Theorem~\ref{thm:vapnik}).

In the construction of the coupling, many nodes in $\fP$ may be split.
We rename them to be $v_1,\ldots, v_m$ (according to the order in which they are processed).
The sequence $v_1,\ldots, v_m$ can be divided into $N$ segments, each assigned to a point in $\calS$.
For a point $v'_i$ in $\calS$, let $\segment(i)$ be the segment (the set of points) assigned to $v'_i$.
%For any node $v$, we use $H_R(v)$ to denote the open halfplane to the right of
%the vertical line that passes though $v$.
For any node $v$ and real $t>0$, we use $H(v,t)$ to denote the right open halfplane defined by
the vertical line $x=x(v)+t$, where $x(v)$ is the $x$-coordinate of $v$ (see Figure~\ref{fig:interval}).

Let $X_i$ ($Y_i$ resp.) be the Poisson distributed random variable corresponding to $v_i$  ($v'_i$ resp.)
(i.e., $X_i\sim \pois(\lambda_{v_i})$ and $Y_i\sim \pois(\lambda/N)$ ) for all $i$.
For any $H\subset \R^2$, we write $X(H) =\sum_{v_i\in H\cap \calP} X_i$ and
$Y(H) =\sum_{v'_i\in H\cap \calS} Y_i$.
We can rewrite $\Pr[\dw(\calS,u)\leq t]$ as follows:
\begin{align}
\label{eq:S}
\Pr[\dw(\calS,u)\leq t] & =\sum_{i=1}^N \Pr[v'_i \text{ is the leftmost point and }\dw(\calS,u)\leq t] +\Pr[\text{no point in }\calS\text{ appears} ] \notag\\
& =\sum_{i=1}^N \Pr[Y_i\ne 0] \, \Pr\Bigl[\sum_{j=1}^{i-1}Y_j=0\Bigr]\,\Pr[Y(H(v'_i, t))=0]+\Pr[\sum_{v'_i\in \calS}Y_i=0]
\end{align}
%\jeff{The above rewrite is a bit opaque to me.}
%Each summand is the probability that $v'_i$ is the leftmost point
%and there is no point with $x$-coordinate larger than $x(v'_i)+t$.
Similarly, we can write that
\footnote{
Note that splitting nodes does not change the distribution of $\dw(\calP,u)$:
Suppose a node $v$ (corresponding to r.v. $X$) was spit to two nodes $v_1$ and $v_2$
(corresponding to $X_1$ and $X_2$ resp.). We can see that
$\Pr[X\ne 0]=\Pr[X_1\ne 0\text{ and }X_2\ne 0]=\Pr[X_1+X_2\ne 0]$.
}
\begin{align}
\label{eq:P}
\Pr[\dw(\calP,u)\leq t]&=\sum_{i=1}^m \Pr[X_i\ne 0]\, \Pr\Bigl[\sum_{j=1}^{i-1}X_j=0\Bigr]\, \Pr[X(H(v_i, t))=0]+\Pr[\sum_{v_i\in \calP}X_i=0]\notag\\
&=\sum_{i=1}^N\sum_{k\in \segment(i)} \Pr[X_k\ne 0]\, \Pr\Bigl[\sum_{j=1}^{k-1}X_j=0\Bigr]\, \Pr[X(H(v_k, t))=0]+\Pr[\sum_{v_i\in \calP}X_i=0]
\end{align}
We proceed by attempting to show each
each summand of \eqref{eq:P} is close to the corresponding one in \eqref{eq:S}.
First, we can see that
$
\Pr[\sum_{v'_i\in \calS}Y_i=0]=\Pr[\sum_{v_i\in \calP}X_i=0]
$
since both $\sum_{v'_i\in \calS}Y_i$ and $\sum_{v_i\in \calP}X_i$ follow
the Poisson distribution $\pois(\lambda)$.

For any segment $i$, we can see that
$\sum_{k\in \segment(i)}\lambda_{v_k}=\lambda/N$.
Moreover, we have $\lambda_{v_k}\leq \lambda/N \leq \perror/32$, thus
$\exp(-\lambda_{v_k})\in (1-\lambda_{v_k}, (1+\perror/16)(1-\lambda_{v_k}))$.
\begin{align}
\label{eq:ynot0}
\sum_{k\in \segment(i)} \Pr[X_k\ne 0] &
=\sum_{k\in \segment(i)}(1-\exp(-\lambda_{v_k}))
\in (1\pm \frac{\perror}{16}) \sum_{k\in \segment(i)}\lambda_{v_k}\notag \\
& \subset (1\pm \frac{\perror}{8}) (1-\exp (\frac{\lambda}{N}))
  =(1\pm \frac{\perror}{8})\Pr[Y_i\ne 0].
\end{align}
Then,
we notice that for any $k\in \segment(i)$ (i.e., $v_k$ is in the segment assigned to $v'_i$), it holds that
\begin{align}
\label{eq:yiis0}
 \Pr\Bigl[\,\sum_{j=1}^k X_j=0\,\Bigr] \in
 [e^{-i \lambda /N}, e^{-\lambda (i-1)/N}]
\subset  (1\pm \frac{\perror}{8}) e^{-\lambda (i-1)/N}
= (1\pm \frac{\perror}{8})\Pr\Bigl[\,\sum_{j=1}^{i-1}Y_j=0\,\Bigr].
\end{align}
%\Pr\Bigl[\,\sum_{j=1}^{i-1}Y_j=0\,\Bigr](1-\perror/8)& \leq \Pr\Bigl[\,\sum_{j=1}^{i-1}Y_j=0\,\Bigr]\,e^{-\lambda /N}=
%e^{-i \lambda /N} \notag\\
%&\leq\leq e^{-\lambda (i-1)/N}
%=\Pr\Bigl[\,\sum_{j=1}^{i-1}Y_j=0\,\Bigr]
%\end{align}
The first inequality holds because $\sum_{j=1}^k X_j \sim \pois\bigl(\sum_{j=1}^k \lambda_{v_j}\bigr)$
and $\lambda (i-1)/N\leq \sum_{j=1}^k \lambda_{v_j}\leq \lambda i/N$.

If we can show that $\Pr[X(H(v_k, t))=0]$ is close to $\Pr[Y(H(v'_i, t))=0]$ for $k\in \segment(i)$, we can finish
the proof easily since each summand of $\eqref{eq:P}$ would be close to the corresponding one in $\eqref{eq:S}$.
However, this is in general not true and we have to be more careful.

\begin{figure}[t]
\centering
\includegraphics[width=0.6\linewidth]{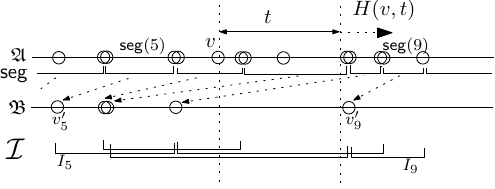}
\caption{Illustration of the interval graph $\calI$.
For illustration purpose, co-located points
(e.g., points that are split in $\fP$)
are shown
as overlapping points.
The arrows indicate the assignment of the segments
to the points in $\fQ$.
Theorem~\ref{thm:vapnik} ensures that any vertical line can not
stab many intervals.
}
\label{fig:interval}
\end{figure}

Recall that the sequence $v_1,\ldots, v_m$ is divided into $N$ segments.
Let $K=\lambda/\perror$.
We say that the $i$th segment (say $\segment(i)=\{v_j, v_{j+1}, \ldots, v_k\}$) is a {\em good segment}
if
$$
E_i=\max\Bigl\{\bigl|\fQ(H(v'_i,t))-\fP(H(v_j,t))\bigr|\, ,\, \bigl|\fQ(H(v'_i,t))-\fP(H(v_k,t))\bigr|\Bigr\}\leq \frac{1}{K}.
$$
Otherwise, the segment is {\em bad}.
For a good segment $\segment(i)$ and any $k\in \segment(i)$,
\begin{align}
\label{eq:his0}
\Pr[X(H(v_k, t))=0] &
%\exp\Bigl(\sum_{v\in H(v_k,t)\cap \calP} \lambda_v\Bigr)
=\exp\bigl(-\lambda \fP(H(v_k,t))\bigr)
\in \exp\bigl(-\lambda \fQ(H(v'_i,t))\pm \lambda/K \bigr)\notag\\
&\subset
\Pr[Y(H(v'_i, t))=0]e^{\pm\lambda /K}
\subset
\Pr[Y(H(v'_i, t))=0](1\pm \perror/8).
\end{align}
We use $\goodseg$ to denote the set of good segments
and $\badseg$ the set of bad segments.
Now, we consider the summations in both \eqref{eq:S} and \eqref{eq:P} with only good segments.
We have that
\begin{align*}
&\sum_{i\in \goodseg} \sum_{k\in \segment(i)} \Pr[X_k\ne 0]\, \Pr\Bigl[\sum_{j=1}^{k-1}X_j=0\Bigr]\, \Pr[X(H(v_k, t))=0]\\
\in \,\, &
\sum_{i\in \goodseg}  \Pr\Bigl[\sum_{j=1}^{i-1}Y_j=0\Bigr](1\pm \perror/8)\, \Pr[Y(H(v'_i, t))=0](1\pm \perror/8) \Pr[Y_i\ne 0]  (1\pm \perror/8)\\
\subset \,\, &
\sum_{i\in \goodseg} \Pr[Y_i\ne 0]\, \Pr\Bigl[\sum_{j=1}^{i-1}Y_j=0\Bigr]\,\Pr[Y(H(v'_i, t))=0]\pm \perror/2,
\end{align*}
where the first inequality is due to \eqref{eq:yiis0} and \eqref{eq:his0}
and the second holds because \eqref{eq:ynot0}

Now, we show the total contributions of bad segments to both \eqref{eq:S} and \eqref{eq:P}
are small. We partition all of the bad segments into $\log (1/\perror)+1$ different sets $B_0,\ldots,B_{\log \frac{1}{\perror}-1}$. Let $B_i=\{i\mid \frac{2^i}{K}< E_i \leq \frac{2^{i+1}}{K}\}$ for $0\leq i\leq \log (1/\perror)-1$. Let $B_{\log \frac{1}{\perror}}=\{i\mid E_i>\frac{1}{\lambda}\}$.

With the above notations, we prove the following crucial inequality:
\begin{equation} \label{eq:1}
\sum_{i=0}^{\log \frac{1}{\perror}} |B_i|\cdot 2^i = O(\perror_1 N K).
\end{equation}

Now, we prove \eqref{eq:1}.
Consider all points $v_1,\ldots,v_m$ and $v'_1,\ldots,v'_n$ lying on the same $x$-axis.
For each $i$ (with $\segment(i)=\{v_j, v_{j+1}, \ldots, v_k\}$),
we draw the minimal interval $I_i$ that contains $v'_i, v_j$ and $v_k$.
If the $i$th segment is bad and belongs to $B_j$, we also say $I_i$ is a {\em bad interval} of label $j$.
All intervals $\{I_i\}_i$ define an interval graph $\calI$.
We can see that any vertical line can stab at most $\perror_1 N+1$ intervals,
because at most $\perror_1$ unit of mass can be transported across the vertical line,
and each interval is responsible for a transportation of exactly $1/N$ units of mass
(except the one that intersects the vertical line).
Hence, the interval graph $\calI$ can be colored with at most $\perror_1 N+1$ colors
(this is because the clique number of $\calI$ is at most $\perror_1 N+1$ and
the chromatic number of an interval graph is the same as its clique number).
Consider a color class $C$ (which consists of a set of non-overlapping intervals).
Imagine we move an interval $I$ of length $t$ along the $x$-axis from left to right.
When the left endpoint of $I$ passes through an bad interval of label $j$ in $C$,
by the definition of bad segments, the right endpoint of $I$ passes through $\Omega(2^j N/K)$ segments. Suppose the color class $C$ contains $b_j$ bad segments in $B_j$. Since the right endpoint of $I$ can pass through at most $N$ segments, we have the following inequalities by summing over all labels:
$$
\sum_{j=0}^{\log \frac{1}{\perror}-1} b_j\cdot 2^j N/K\leq N.
$$
Summing up all color classes, we obtain (\ref{eq:1}).

For $B_j$ ($0\leq j\leq \log \frac{1}{\perror}$), we can bound the total contribution
as follows.
By the definition of $B_j$, we can see that
\begin{align*}
&\sum_{i\in B_j} \sum_{k\in \segment(i)} \Pr[X_k\ne 0]\, \Pr\Bigl[\sum_{j=1}^{k-1}X_j=0\Bigr]\, \Pr[X(H(v_k, t))=0]\\
\subset \,\, &
(1\pm 5\cdot 2^j\perror)\sum_{i\in B_j} \Pr[Y_i\ne 0]\, \Pr\Bigl[\sum_{j=1}^{i-1}Y_j=0\Bigr]\,\Pr[Y(H(v'_i, t))=0],
\end{align*}
Thus, the total contribution of bad segments in $B_j$ ($0\leq j\leq \log \frac{1}{\perror}-1$) to the corresponding summands in (\eqref{eq:S}-\eqref{eq:P})
is at most
$$
5\cdot 2^j\perror \sum_{i\in B_j} \Pr[Y_i\ne 0]= 5 |B_j|\cdot 2^j\perror \times (1-\exp{(-\frac{\lambda}{N})})=O(|B_i| 2^j\perror \lambda/N),
$$
where $\Pr[Y_i\ne 0]=1-\exp{(-\frac{\lambda}{N})}$ (since $Y_i\sim \pois(\frac{\lambda}{N})$).

For $B_{\log \frac{1}{\perror}}$, the total contribution is bounded by the following. \begin{align*}
&\Biggl|\sum_{i\in B_{\log \frac{1}{\perror}}}\Bigl( \sum_{k\in \segment(i)} \Pr[X_k\ne 0]\, \Pr\Bigl[\sum_{j=1}^{k-1}X_j=0\Bigr]\, \Pr[X(H(v_k, t))=0] \\
& -\Pr[Y_i\ne 0]\, \Pr\Bigl[\sum_{j=1}^{i-1}Y_j=0\Bigr]\,\Pr[Y(H(v'_i, t))=0]\Bigr)\Biggr|\\
\leq \,\, &
\sum_{i\in B_{\log \frac{1}{\perror}}} \sum_{k\in \segment(i)} \Pr[X_k\ne 0]\,+\sum_{i\in B_{\log \frac{1}{\perror}}} \Pr[Y_i\ne 0]\,
\leq \,\,  3 \sum_{i\in B_{\log \frac{1}{\perror}}} \Pr[Y_i\ne 0] \\
&\leq \,\, 3 |B_{\log \frac{1}{\perror}}| \times (1-\exp{(-\frac{\lambda}{N})} =O(|B_{\log \frac{1}{\perror}}| \lambda/N)
\end{align*}

Summing up all $j$ and using (\ref{eq:1}), we obtain the following inequality .
$$
\sum_{j=0}^{\log \frac{1}{\perror}} O(|B_j| 2^j \perror \lambda/N) = O(\perror_1 \perror \lambda K )\leq \frac{\perror}{4}.
$$
This finishes the proof.
\end{proof}

\section{Missing Details in Section~\ref{sec:rfunction}}
\label{app:exprkernel}

\textbf{Lemma \ref{lm:expr1}.}
Let $N=O\bigl(\e_1^{-2}\e_0^{-(d-1)/2}\log (1/\e_0)\bigr)$, where $\e_0=(\e/4(r-1))^r$, $\e_1=\eps\beta^2$. For any $t\geq 0$ and any direction $u\in \calP^{\polar}$, we have that
$$  \Prob_{P\sim \calS}\Bigl[\max_{v\in P}\innerprod{u}{v}^{1/r}\geq t\Bigr]\in \Prob_{P\sim \calP}\Bigl[\max_{v\in \calE(P)}\innerprod{u}{v}^{1/r}\geq t)\Bigr]\pm \e_1/4, \text{ and }
$$
$$ \Prob_{P\sim \calS}\Bigl[\min_{v\in P}\innerprod{u}{v}^{1/r}\geq t\Bigr]\in \Prob_{P\sim \calP}\Bigl[\min_{v\in \calE(P)}\innerprod{u}{v}^{1/r}\geq t)\Bigr]\pm \e_1/4. \quad\quad
$$

\begin{proof}
The argument is almost the same as that in Lemma~\ref{lm:quant1}. Let $L=O(\e_0^{-(d-1)/2})$.
We still build a mapping $g$ that maps each realization $\calE(P)$
to a point in $\R^{dL}$, as follows:
Consider a realization $P$ of $\calP$.
Suppose $\calE(P)=\{(x^1_1,\ldots, x^1_d),\ldots, (x^L_1,\ldots, x^L_d)\}$
(if $|\calE(P)|<L$, we pad it with $(0,\ldots, 0)$).
We let
$
g(\calE(P))=(x^1_1,\ldots, x^1_d,\ldots, x^L_1,\ldots, x^L_d)\in \R^{dL}.
$
For any $t\geq 0$ and any direction $u\in \calP^{\polar}$,
note that $\max_{v\in \calE(P)}\innerprod{u}{v}^{1/r}\geq t$ holds
if and only if there exists some $1\leq i\leq |\calE(P)|$ satisfies that
$\sum_{j=1}^d x^i_j u_j\geq t^r$,
which is equivalent to saying that
point $g(\calE(P))$ is in the union of the those $|\calE(P)|$ half-spaces.

Let $X$ be the image set of $g$. Let $(X,\calR^i)$ ($1\leq i\leq L)$) be a range space, where $\calR^i$ is the set of half spaces $\{\sum_{j=1}^d x_j^iu_j\geq t\mid u=(u_1,\ldots,u_d)\in \R^d, t\geq0\}$. Let $\calR'=\{\cup r_i\mid r_i\in \calR^i,i\in [L]\}$. Note that each $(X,\calR^i)$ has
VC-dimension $d+1$. By Theorem~\ref{thm:quant1}, we have that the VC-dimension of $(X,\calR')$ is bounded by
$O((d+1)L\lg L)=O(\e_0^{-(d-1)/2} \log (1/\e_0))$.
Then by Theorem~\ref{thm:quant2}, for any $t$ and any direction $u$, we have that
$
 \Prob_{P\sim \calS}[\max_{v\in P}\innerprod{u}{v}^{1/r}\geq t]\in \Prob_{P\sim \calP}[\max_{v\in \calE(P)}\innerprod{u}{v}^{1/r}\geq t)]\pm \e_1/4.
$
The proof for the second statement is the same.
%Similarly, we have that
%$$ \Prob_{P\sim \calS}[\min_{v\in P}\innerprod{u}{v}^{1/r}\geq t]\in \Prob_{P\sim \calP}[\min_{v\in \calE(P)}\innerprod{u}{v}^{1/r}\geq t)]\pm %\e_1/4.
%$$
\end{proof}

\section{Computing the Expected Direction Width}
\label{sec:computing}

We handle both the existential and location model of uncertain points in this section.
For any direction $u$, denote by
$\omega(\calP,u)$ the expected width of $\calP$ along the direction
$u$, and $f(\calP,u) = \E_{P \sim \calP}[\max_{p \in P} \innerprod{u}{p}]$ is the support function.  Recall $\omega(\calP,u) = f(\calP,u) - f(\calP,-u)$ by linearity of expectation.

\subsection{Computing Expected Width for Existential Uncertainty}
 The existential model is a bit simpler and we handle that first.
 Recall in this model we
let $\calP$ be a set of $n$ uncertain points, and each point $v\in
\calP$ has a probability $p_v$.
We have the following two lemmas.

\begin{lemma}\label{lem:10}
For any direction $u$, we can compute $\omega(\calP,u)$, $f(\calP,u)$, and $\grad f(\calP,u)$ in $O(n\log
n)$ time; if the points of $\calP$ are already sorted along the
direction $u$, then we can compute them in $O(n)$ time.
\end{lemma}
\begin{proof}
Consider any direction $u$.
Without loss of generality, assume $\|u\|=1$.
In the following, we first show how to compute $f(\calP,u)$. The value $f(\calP,-u)$ can be computed in a similar manner and we ignore the discussion. After having $f(\calP,u)$ and $f(\calP,-u)$, $\omega(\calP,u)$ can be computed immediately by $\omega(\calP,u) = f(\calP,u) - f(\calP,-u)$. Finally, we will discuss how to compute $\grad f(\calP,u)$.
%Note that $u$ is always considered as a vector.
Let $\rho(u)$ be the ray of direction $u$ in the plane passing through the origin.

Consider a point $v\in \calP$.
%Denote by $v(u)$ the perpendicular projection of $v$
%on $\rho(u)$ and thus
Note that $\innerprod{v}{u}$ is
the coordinate  of the perpendicular projection of $v$ on $\rho(u)$.
Denote by $\calP^R(v,u)$ the subset of
points $v'\in \calP$ such that $v'>_u v$ (i.e., $\innerprod{v'}{u}>\innerprod{v}{u}$).
Denote by $\Pr^R(v,u)$ the probability that $v$ appears in a realization but all points of
$\calP^R(v,u)$ do not appear (i.e., $\innerprod{v}{u}$ is the largest among all
points of $\calP$ that appear in the realization). Hence, we have
\begin{equation}\label{eq:PrR}
\Pr^R(v,u)=p_v\cdot \prod_{v'>_u v}(1-p_{v'}).
\end{equation}

Now $f(\calP,u)$ can be seen as the expected largest coordinate of the
projections of the points in $\calP$ on $\rho(u)$.
%An easy observation is that $\omega(\calP,u)=\omega_R(\calP,u)-\omega_L(\calP,u)$.
According to the definition of $\Pr^R(v,u)$, we have
$f(\calP,u)= \sum_{v\in \calP}\Pr^R(v,u)\innerprod{v}{u}$.

Based on the above discussion, we can compute $f(\calP,u)$ in the
following way. First, we project all points of $\calP$ on $\rho(u)$ and
obtain the coordinate $\innerprod{u}{v}$ for each $v\in \calP$. Second,
we sort all points of $\calP$ by the
coordinates of their projections on $\rho(u)$. Then, the values
$\Pr^R(v,u)$ for all points $v\in \calP$ can be obtained in $O(n)$
time by considering the projection points on $\rho(u)$ from right to
left. Finally, $f(\calP,u)$ can be computed in additional $O(n)$ time.
Therefore, the total time for computing $f(\calP,u)$ is $O(n\log
n)$, which is dominated by the sorting.
If the points of $\calP$ are
given sorted along the direction $u$, then we can avoid the sorting
step and compute $f(\calP,u)$ in overall $O(n)$ time.

It remains to compute $\grad f(\calP,u)$. Recall that $\grad f(\calP,u)=\sum_{v\in \calP}\Pr^R(v,u)v$ by the proof of Lemma \ref{lm:complexity}. Note that the above has already computed $\Pr^R(v,u)$ for all points $v\in \calP$. Therefore, $\grad f(\calP,u)$ can be computed in additional $O(n)$ time.
The lemma thus follows.
\end{proof}

\begin{lemma}\label{lem:20}
We can build a data structure of $O(n^2)$ size in $O(n^2\log n)$ time
that can compute $\omega(\calP,u)$, $f(\calP,u)$, and $\grad f(\calP,u)$ in $O(\log n)$ time for any query direction $u$. Further, we can construct $M$ explicitly in $O(n^2\log n)$ time.
\end{lemma}
\begin{proof}
Consider any direction $u$ with $||u||=1$. We follow the definitions and notations in the proof of Lemma \ref{lem:10}.
We first show how to build a data structure to compute $f(\calP,u)$. Computing $f(\calP,-u)$ can be done similarly. Again, after having $f(\calP,u)$ and $f(\calP,-u)$, $\omega(\calP,u)$ can be computed immediately by $\omega(\calP,u) = f(\calP,u) - f(\calP,-u)$.

Denote by $o$ the origin. For any {\em ray} $\rho$ through $o$ in the plane, we refer to the {\em
angle} of $\rho$ as the angle $\alpha$ in $[0,2\pi)$ such that after we rotate the
$x$-axis around $o$ counterclockwise by $\alpha$ the $x$-axis has
the same direction as $\rho$ (see Fig.~\ref{fig:angle}(a)).
For any (undirected) line $l$ through
$o$, we refer to the {\em
angle} of $l$ as the angle $\alpha$  in $[0,\pi)$ such that after we rotate the
$x$-axis around $o$ counterclockwise by $\alpha$ the $x$-axis is collinear with
$l$.

Recall $\rho(u)$ is the ray through $o$ with direction $u$.
We define the {\em angle} of $u$ as the angle of the ray $\rho(u)$, denoted by $\theta_u$.
For ease of discussion, we assume $\theta_u$ is in $[0,\pi)$ since the the case $\theta_u\in [\pi,2\pi)$ can be handled similarly.

\begin{figure}[t]
\begin{minipage}[t]{\linewidth}
\begin{center}
\includegraphics[totalheight=1.4in]{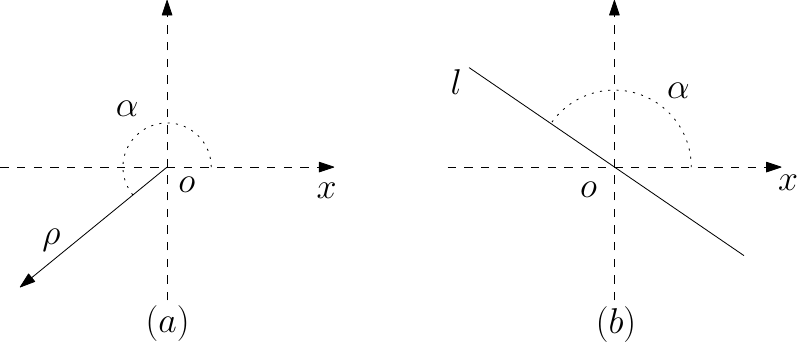}
\caption{\footnotesize Illustrating the definition of the angle
$\alpha$ of: (a) a ray $\rho$ and (b) a line $l$.}
\label{fig:angle}
\end{center}
\end{minipage}
\end{figure}

%Consider any direction $u$. As explained in the proof of Lemma \ref{lem:10}, we assume $u$ is an angle in $[0,\pi)$ and $l(u)$ is the line through $o$ whose angle is $u$.
%in $[0,\pi)$ formed by $u$ and the positive direction of the $x$-axis (i.e., if
%we rotate the $x$-axis counterclockwise around the origin, after we rotate
%$\alpha(u)$ degree, the $x$-axis becomes collinear with $u$).
%Below, we follow the notation defined in the proof of Lemma \ref{lem:10}.
%with the convention that when $u$ is used as a
%parameter, $u$ is replaced by $\alpha$ (i.e., we use $l(\alpha)$,
%$\Pr^L(v,\alpha)$, $\omega(\calP,\alpha)$, instead of $l(u)$).

We call the order of the points of $\calP$ sorted by the coordinates
of their projections on the ray $\rho(u)$ the {\em canonical order} of $\calP$
with respect to $u$. An easy observation is that when we increase the
angle $u$, the canonical order of $\calP$ does not change until $u$ is
perpendicular to a line containing two points of $\calP$. There are
$O(n^2)$ lines in the plane each of which contains two points of
$\calP$ and the directions of these lines partition $[0,\pi)$ into
$O(n^2)$ intervals such that if $\theta_u$ changes in each interval the
canonical order of $\calP$ does not change. In the following, we show
that for each of the above intervals, the value of $f(\calP,u)$ is a function of
the angle $\theta_u$, and more specifically $f(\calP,u)=a\cdot
\cos(\theta_u)+b\cdot \sin(\theta_u)$ where $a$ and $b$ are constants when $\theta_u$
changes in the
interval. As preprocessing for the lemma, we will compute the function
$f(\calP,u)$ for each interval; for each query direction $u$, we first
find the interval that contains $\theta_u$ by binary search in $O(\log n)$
time and then obtain
the value $f(\calP,u)$ in constant time using the function for
the interval. The details are given below.

For simplicity of discussion,
we make a general position assumption that no three points of $\calP$
are collinear.
For any two points $v$ and $v'$ in $\calP$,
let $\beta(v,v')$ denote the angle of the line perpendicular to the
line containing $v$ and $v'$, and we also say $\beta(v,v')$ is {\em defined}
by $v$ and $v'$. We sort all $O(n^2)$ angles $\beta(v,v')$ for
$v,v'\in \calP$ in increasing order, and let
$\beta_1,\beta_2,\ldots,\beta_h$ be the sorted list with $h=O(n^2)$.
For simplicity, let $\beta_0=0$ and $\beta_{h+1}=\pi$. These angles
partition $[0,\pi)$ into $h+1$ intervals.
Consider an interval $I_i=(\beta_i,\beta_{i+1})$ for any $0\leq i\leq
h$. Below we compute the
function $f(\calP,u)=a\cdot \cos(\theta_u)+b\cdot \sin(\theta_u)$ for
$\theta_u\in
(\beta_i,\beta_{i+1})$. Again, note that when $\theta_u$ changes in $I_i$, the
canonical order of $\calP$ does not change.

According to the proof of Lemma \ref{lem:10},
$f(\calP,u)= \sum_{v\in \calP}\Pr^R(v,u) \innerprod{v}{u}$.
Since the canonical order of $\calP$ does not change for any $\theta_u\in
I_i$,  for any $v\in \calP$, $\Pr^R(v,u)$ is a constant when $\theta_u$ changes in $I_i$.
Next, we consider the coordinate $\innerprod{v}{u}$ on $\rho(u)$.

For each point $v\in \calP$, let $\alpha_v$ be the angle of the
ray originating from $o$ and containing $v$ (i.e., directed from $o$ to $v$),
and let $d_v$ be the length of the
line segment $\overline{vo}$. Note that $\alpha_v$ and $d_v$ are fixed
for the input.  Then, we have (see Fig.~\ref{fig:coordinate})
$$\innerprod{v}{u}=d_v\cdot \cos(\alpha_v-\theta_u)=d_v\cdot \cos(\alpha_v)\cdot \cos(\theta_u)
+d_v\cdot \sin(\alpha_v)\cdot \sin(\theta_u).$$

Hence, we have the following
\begin{equation*}
\begin{split}
f(\calP,u)
&=\sum_{v\in \calP}\Pr^R(v,u)\innerprod{v}{u}\\
&=\sum_{v\in
\calP}\Pr^R(v,u)\cdot d_v\cdot [\cos(\alpha_v)\cdot
\cos(\theta_u)+\sin(\alpha_v)\cdot \sin(\theta_u)]\\
%&=\sum_{v\in \calP}\Pr^R(v,u)\cdot d_v\cdot \cos(\alpha_v)\cdot
%\cos(\theta_u)+\sum_{v\in \calP}\Pr^R(v,u)\cdot d_v\cdot\sin(\alpha_v)\cdot \sin(\theta_u)\\
&=\cos(\theta_u)\cdot \Big[\sum_{v\in \calP}\Pr^R(v,u)\cdot d_v\cdot
\cos(\alpha_v)\Big]+\sin(\theta_u)\cdot
\Big[\sum_{v\in \calP}\Pr^R(v,u)\cdot d_v\cdot\sin(\alpha_v)\Big].\\
\end{split}
\end{equation*}

Let $a=\sum_{v\in \calP}\Pr^R(v,u)\cdot d_v\cdot \cos(\alpha_v)$
and $b=\sum_{v\in \calP}\Pr^R(v,u)\cdot d_v\cdot\sin(\alpha_v)$.
Hence, $a$ and $b$ are constants when $\theta_u$ changes in $I_i$.
%Similarly, define $a^R=\sum_{v\in \calP}\Pr^R(v,u)\cdot d_v\cdot
%\cos(\alpha_v)$ and $b^R=\sum_{v\in \calP}\Pr^R(v,u)\cdot
%d_v\cdot\sin(\alpha_v)$, and
Then, we have $f(\calP,u)=a\cdot \cos(\theta_u)+b\cdot \sin(\theta_u)$, for any $\theta_u\in I_i$. Therefore, if we know the
two values $a$ and $b$, we can compute $f(\calP,u)$ in constant
time for any direction $\theta_u\in I_i$.

\begin{figure}[t]
\begin{minipage}[t]{\linewidth}
\begin{center}
\includegraphics[totalheight=1.4in]{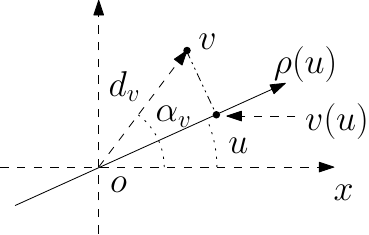}
\caption{\footnotesize Illustrating the computation of the coordinate
$x(v,u)$ on $l(u)$: $v(u)$ is the perpendicular projection of $v$ on
$l(u)$. The length of $\overline{ov}$ is $d_v$. }
\label{fig:coordinate}
\end{center}
\end{minipage}
\end{figure}

In the sequel, we show that we can compute $a$ and $b$ for all
intervals $I_i=(\beta_i,\beta_{i+1})$ with $i=0,1,\ldots,h$ in
$O(n^2)$ time. For each
interval $I_i$, we use $a(I_i)$ and $b(I_i)$ to denote the
corresponding $a$ and $b$ respectively for the interval $I_i$.
%We define $a^L(I_i)$, $b^L(I_i)$, $a^R(I_i)$, and $b^R(I_i)$ similarly.

Suppose we have computed $a(I_i)$ and $b(I_i)$ for the interval $I_i$, and also suppose
have computed the value $\Pr^R(v,u)$ for each point
$v\in \calP$ when $\theta_u\in I_i$ (note that $\Pr^R(v,u)$
is a constant for any $\theta_u\in I_i$).
Initially, we can compute these values for the interval $I_0$ in $O(n\log n)$ time
by Lemma \ref{lem:10}.
Below, we show that we can obtain $a(I_{i+1})$ and $b(I_{i+1})$ in
constant time, based on the above values maintained for $I_i$.

Recall that $I_i=(\beta_i,\beta_{i+1})$ and
$I_{i+1}=(\beta_{i+1},\beta_{i+2})$. Suppose the angle $\beta_{i+1}$
is defined by the two points $v_1$ and $v_2$ of $\calP$. In other words,
$\beta_{i+1}$ is the angle of the line perpendicular to the
line through $v_1$ and $v_2$. If we increase the angle $\theta_u$ in
$(\beta_i,\beta_{i+2})$, the canonical order of $\calP$ does not
change except that $v_1$ and $v_2$ exchange their order when $\theta_u$ passes the value
$\beta_{i+1}$. Therefore, for each point $v\in \calP\setminus
\{v_1,v_2\}$, the value $\Pr^R(v,u)$ is a constant
for any $\theta_u\in (\beta_i,\beta_{i+2})$.
Based on this observation, we can compute $a(I_{i+1})$ in the following way.

We first analyze the change of the values $\Pr^R(v_1,u)$ and
$\Pr^R(v_2,u)$ when $\theta_u$ changes from $I_i$ to $I_{i+1}$. Let
$u$ and $u'$ be any two directions such that $\theta_u$  is
in $I_i$ and $\theta_{u'}$ is in $I_{i+1}$. Without loss of
generality, we assume $\innerprod{v_1}{u}<\innerprod{v_2}{u}$, and thus,
$\innerprod{v_1}{u'}>\innerprod{v_2}{u'}$ since $v_1$ and $v_2$ exchange their order.
Observe that $\Pr^R(v_1,u')=\Pr^R(v_1,u)/
(1-p_{v_2})$ and $\Pr^R(v_2,u')=\Pr^R(v_2,u)\cdot (1-p_{v_1})$. Thus, we
can obtain $\Pr^R(v_1,u')$ and $\Pr^R(v_2,u')$ in constant time since
we already maintain $\Pr^R(v_1,u)$ and $\Pr^R(v_2,u)$. Consequently, we have
\begin{align}
\label{eq:aR}
a(I_{i+1}) &=
 a(I_i)-[\Pr^R(v_1,u)d_{v_1}\cos(\alpha_{v_1})+\Pr^R(v_1,u)d_{v_2}\cos(\alpha_{v_2})]\\ \nonumber
&+[\Pr^R(v_1,u')d_{v_1}\cos(\alpha_{v_1})+\Pr^R(v_1,u')d_{v_2}\cos(\alpha_{v_2})]\\ \nonumber
&=a(I_i)+d_{v_1}\cos(\alpha_{v_1})\cdot [\Pr^R(v_1,u')-\Pr^R(v_1,u)]
+d_{v_2}\cos(\alpha_{v_2})\cdot [\Pr^R(v_2,u')-\Pr^R(v_2,u)].
\end{align}

Hence, after we compute $\Pr^R(v_1,u')$ and $\Pr^R(v_2,u')$, we can obtain $a(I_{i+1})$ in constant time.

Similarly, we can obtain $b(I_{i+1})$ in constant time. Also, the values
$\Pr^R(v_1,u)$ and $\Pr^R(v_2,u)$ are updated for $\theta_u\in I_{i+1}$.

In summary, after the $O(n^2)$ angles $\beta(v,v')$ are sorted in $O(n^2\log
n)$ time, the above computes the functions
$f(\calP,u)=a(I_i)\cdot \cos(\theta_u)+b(I_i)\cdot \sin(\theta_u)$ for all intervals $I_i$
with $i=0,1,\ldots,h$, in additional $O(n^2)$ time. This finishes our
preprocessing.

Consider any query direction $u$. By binary search, we first find the
two angles $\beta_i$ and $\beta_{i+1}$ such that $\beta_i\leq
\theta_u<\beta_{i+1}$. If $\beta_i\neq \theta_u$, then $\theta_u$ is in $I_i$ and we can
use the function $f(\calP,u)=a(I_i)\cos(\theta_u)+b(I_i)\sin(\theta_u)$ to
compute $f(\calP,u)$ in constant time. If $\beta_i= \theta_u$, then the
function $f(\calP,u)=a(I_i)\cos(\theta_u)+b(I_i)\sin(\theta_u)$ still gives the
correct value of  $f(\calP,u)$ since when $\theta_u=\beta_i$ the
projections of the two points of $\calP$ defining $\beta_i$ on $\rho(u)$
overlap and we can still consider the canonical order of $\calP$ for
$\theta_u=\beta_i$ the same as that for $\theta_u\in I_i$. Hence, the query time is
$O(\log n)$.

Next, we show how to compute $\grad f(\calP,u)$. Recall that $\grad
f(\calP,u)=\sum_{v\in \calP}\Pr^R(v,u)v$ by the proof of Lemma
\ref{lm:complexity}. As preprocessing, we compute the value
$\sum_{v\in \calP}\Pr^R(v,u)v$ for each interval $\theta_u\in
(\beta_i,\beta_{i+1})$ for $i=0,1,\ldots,h$. This can be done in
$O(n^2)$ time (after we sort all angles), by using the similar idea as
above. Specifically, suppose we already have $\grad
f(\calP,u)=\sum_{v\in \calP}\Pr^R(v,u)v$ for $\theta_u\in
(\beta_i,\beta_{i+1})$; then we can compute $\grad
f(\calP,u)=\sum_{v\in \calP}\Pr^R(v,u)v$ for $\theta_u\in
(\beta_{i+1},\beta_{i+2})$ in constant time. This is because when $\theta_u$
changes from $(\beta_i,\beta_{i+1})$ to $(\beta_{i+1},\beta_{i+2})$,
$\Pr^R(v,u)$ does not change for any $v\in \calP\setminus\{v_1,v_2\}$
and $\Pr^R(v,u)$ for $v\in \{v_1,v_2\}$ can be updated in constant
time, as shown above. Due to the above preprocessing, given any
direction $u$, we can compute $\grad f(\calP,u)$ in $O(\log n)$ time
by binary search, similar to that of computing $f(\calP,u)$. Further,
according to Lemma \ref{lm:complexity}, the above preprocessing
essentially computes $M$, in totally $O(n^2\log n)$ time.
The lemma thus follows.
\end{proof}

\subsection{Computing Expected Width for Locational Uncertainty}

In this setting let $\P$ be a set of $n$ uncertain points each taking
one of several locations from a set of $m$ locations in $S$.  The
probability that a point $v\in \P$ is in location $s \in S$ is denoted
$p_{v,s}$.  To simplify analysis and discussion, we assume each
location $s \in S$ only has the potential to be realized by any one
uncertain point $v \in \P$.

%For any direction $u$, denote by $\omega(\P,u)$ the expected width of $\P$ along the direction $u$.
We now replicate the lemmas in the previous section for this setting.  We use the same notation and structure when possible.

\begin{lemma}\label{lem:10-loc}
For any direction $u$, we can compute $\omega(\P,u)$, $f(\P,u)$, and $\grad f(\calP,u)$ in $O(m\log
m)$ time; if the points of $S$ are already sorted along the
direction $u$, then we can compute them in $O(m)$ time.
\end{lemma}
\begin{proof}
Again, we first compute $f(\P,u)$ since $\omega(\calP,u) = f(\calP,u) - f(\calP,-u)$ and $f(\calP,-u)$ can be computed similarly.

We follow the structure and proof of Lemma \ref{lem:10} and just note the changes.
The first change is that we need to keep a bit more structure since there is now dependence between the different locations of each uncertain point $v$.
Denote by $S^R(s,u)$ the subset of $s' \in S$ such that $\innerprod{s'}{u} > \innerprod{s}{u}$ and denote by $\Pr^R_\emptyset(s,u)$ as the probability that no point $v \in \P$ appears at a larger location than $s \in S$ along direction $u$.  To describe this probability we first define a vector $A_v$ indexed by $s$ as $A_v[s] = 1 - \sum_{s' \in S^R(s,u)} p_{v,s'}$ as the probability that uncertain point $v$ does not appear in any of its possible locations which are after $s$ along direction $u$.
Now we can define
\[
\Pr^R_\emptyset(s,u) = \prod_{v \in \P} A_v[s].
\]
Finally, we define $\Pr^R(s,v,u)$ as the probability that the largest point along $u$ is uncertain point $v \in \P$ at location $s \in S$.  This updates equation (\ref{eq:PrR}) to be
\[
\Pr^R(v,s,u)=p_{v,s} \cdot \Pr^R_\emptyset(s,u) / A_v[s].
\]
Note the two key differences.  First we need to sum the probabilities for each location of $v$ since they are mutually exclusive.  Second, value $A_v[s]$ needs to be factored out of $\Pr^R_\emptyset(s,u)$ because it is already accounted for in $p_{v,s}$ locating $v$ at $s$, again since they are mutually exclusive.

%Similarly for the largest point we define $S^R(s,u) = \{s' \in P \mid \innerprod{s'}{u} > \innerprod{s}{u}\}$, the vectors $B_v[s] = 1- \sum_{s' \in S^R(s,u)} p_{s',v}$, the empty set probability as
%\[
%\Pr^R_\emptyset(s,u) = \prod_{v \in \P} B_v[s],
%\]
%and updating equation (\ref{eq:PrR}), the probability that $v$ is the rightmost point at $s$ as
%\[
%\Pr^R(v,s,u) = p_{v,s} \cdot \Pr^R_\emptyset(s,u) / B_v[s].
%\]

It follows
that $f(\P,u) = \sum_{s \in S,v \in \P} \Pr^R(v,s,u) \innerprod{s}{u}$.

To compute $f(\P,u)$ we again start by projecting all points $P$ onto $\rho(u)$ obtaining coordinates $\innerprod{u}{s}$ and sorting, if needed.  This takes $O(m \log m)$ time.
Given these coordinates, sorted, it now takes a bit more work in the locational setting to show that $f(\P,u)$ can be computed in $O(m)$ additional time.  We focus on computing all $m$ values $\Pr^R(v,s,u)$; from there is straightforward to compute $f(\P,u)$ in $O(m)$ time.

We sweep over the points $s \in P$ from largest $\innerprod{s}{u}$ value to smallest, and we maintain each $A_v[s]$ and $\Pr^R_\emptyset(s,u)$ along the way.  Given these, it is not hard to calculate $\Pr^R(v,s',u)$ in constant time with $p_{v,s'}$ for $s' \in S$ as the next smallest value $\innerprod{s'}{u}$.  The important observation is that we only need to update $A_v[s]^{\text{new}} = A_v[s]^{\text{old}} - p_{v,s}$ if $p_{v,s} > 0$ (which by assumption holds for only one $v \in \P$).  Then $\Pr^R_\emptyset(s,u)$ is updated by multiplying by $A_v[s]^{\text{new}}/A_v[s]^{\text{old}}$.  Both operations can be done in constant time as needed to complete the proof.

It remains to compute $\grad f(\calP,u)$. It is not hard to see
that in the locational model $$\grad f(\calP,u)=\sum_{s\in S, v\in \calP}\Pr^R(v,s,u)v.$$
Note that the above has already computed the $m$ values $\Pr^R(v,s,u)$ for all $s\in S$ and $v\in \calP$. Therefore, $\grad f(\calP,u)$ can be computed in additional $O(m)$ time.
\end{proof}

\begin{lemma}\label{lem:20loc}
We can build a data structure of $O(m^2)$ size in $O(m^2\log m)$ time
that can compute $\omega(\P,u)$, $f(\P,u)$,  and $\grad f(\calP,u)$ in $O(\log m)$ time for any query
direction $u$. Further, we can construct $M$ explicitly in $O(m^2\log m)$ time.
\end{lemma}
\begin{proof}
Again we first discuss the case for computing $f(\P,u)$. For ease of discussion, we assume the angle $\theta_u$ is in $[0,\pi)$.
We again follow the structure of Lemma \ref{lem:20}.  The geometry is largely the same, except that there are $h = O(m^2)$ angles $\beta_1, \beta_2, \ldots, \beta_h$ since each pair $s, s' \in S$ now defines an angle $\beta(s,s')$.
But it remains to compute $f(\P,u) = a \cdot \cos(\theta_u) + b \cdot
\sin (\theta_u)$ for some constants $a$ and $b$ for any $\theta_u$ in each
$(\beta_i, \beta_{i+1})$.  The argument is virtually the same,
replacing $\Pr^R(v,u)$ with $\Pr^R(v,s,u)$.

It remains to show that we can calculate the constants $a(I_{i+1})$
and $b(I_{i+1})$ for an interval $I_{i+1} = (\beta_{i+1},
\beta_{i+2})$ efficiently, given the values for interval $I_i =
(\beta_i, \beta_{i+1})$.  Assume $\beta_{i+1}$ is defined for two
points $s_1, s_2 \in S$, where $\innerprod{s_1}{u} <
\innerprod{s_2}{u}$ for $\theta_u \in I_i$ and $\innerprod{s_1}{ u'} >
\innerprod{s_2}{u'}$ for $\theta_{u'} \in I_{i+1}$.  By definition, the
ordering among all other pairs of points is unchanged within
$(\beta_i,\beta_{i+2})$.    Let only $v_1$ take location $s_1$ with
positive probability and only $v_2$ take location $s_2$ with positive
probability.  We focus on the more general case where $v_1 \neq v_2$;
when $v_1 = v_2$, it is easier to update.

We again focus on updating $a$ and the algorithm for $b$ is symmetric.  By the $v_1 \neq v_2$ assumption $A_{v_1}[s_1]$ and $A_{v_2}[s_2]$ are unchanged in interval $(\beta_i,\beta_{i+2})$.  However from directions $u$ to $u'$ $A_{v_2}[s_1]$ increases by $p_{v_2,s_2}$ and $A_{v_1}[s_2]$ decreases by $p_{v_1,s_1}$; all other such values are unchanged.
Let $A_{v}[s]^u$ denote the value in direction $u$.
Hence
\[
\Pr^R_\emptyset(s_1,u') = \Pr^R_\emptyset(s_1,u) \frac{A_{v_2}^{u'}[s_1]}{A_{v_2}^u[s_1]},
\]
and so if we can update $A_{v_2}[s_1]$, we can update $\Pr^R_\emptyset(s_1,u')$ in constant time.  Only $A_{v_1}[s_2]$ and $\Pr^R_\emptyset(s_2,u')$ also need to be updated.
And then
\[
\Pr^R(v_1,s_1,u')
=
p_{v_1,s_1} \cdot \Pr^R_\emptyset(s_1,u')/A_{v_1}[s_1]
=
\Pr^R(v_1,s_1,u) \frac{\Pr^R_\emptyset(s_1,u')}{\Pr^R_\emptyset(s_1,u)}
\]
can also be updated in constant time, and similar for $\Pr^R(v_2,s_2,u')$.  Thus the only remaining difficulty is accessing and updating $A_{v_1}[s_2]$ and $A_{v_2}[s_1]$.  We can easily do this for if we store the full size $m$ array $A_v[\cdot]$ for each $v \in \V$.  Note that this takes $O(m \cdot n)$ space, but since the output is a structure of size $O(m^2)$ and $n \leq m$, this is not prohibitive. (We note that these full arrays are not explicitly required in Lemma \ref{lem:10-loc}, which only requires $O(m)$ space.)

Finally, we can update $a(I_{i+1})$ from $a(I_i)$ similarly to equation (\ref{eq:aR}) using $\Pr^R(v,s,u)$ in place of $\Pr^R(v,u)$.  Thus in $O(m^2)$ time, after sorting all interval breakpoints in $O(m^2 \log m)$ time, we can build a data structure that allows calculation of $f(\P,u)$ for any $u$ in $O(\log m)$ time.

Next, to compute $\grad f(\calP,u)$, recall that in the locational
model $\grad f(\calP,u)=\sum_{s\in S, v\in \calP}\Pr^R(v,s,u)v$ by the
proof of Lemma \ref{lm:complexity}. As preprocessing, we compute
the value $\sum_{s\in S, v\in \calP}\Pr^R(v,u)v$ for each interval
$\theta_u\in (\beta_i,\beta_{i+1})$ for $i=0,1,\ldots,h$. This can be done in
$O(m^2)$ time (after we sort all angles), by using the similar idea as
above. The argument is similar to that in Lemma \ref{lem:20} and we
ignore the details. Due to the above preprocessing, given any
direction $u$, we can compute $\grad f(\calP,u)$ in $O(\log m)$ time
by binary search. Further, the above preprocessing essentially
computes $M$, in totally $O(m^2\log m)$ time.
The lemma thus follows.
\end{proof}

%\end{spacing}
\end{document}